\pdfoutput=1

\documentclass[letterpaper,10pt,oneside]{article}

\usepackage[utf8]{inputenc}
\usepackage[english]{babel}
\usepackage[english=american]{csquotes}
\MakeOuterQuote{"}
\usepackage[nottoc]{tocbibind}
\usepackage{lmodern}

\usepackage{authblk}
\usepackage{mathtools}
\usepackage{amsfonts}
\usepackage{amsmath}
\usepackage{amssymb}
\usepackage{amsthm}
\usepackage{cite}
\usepackage{enumitem}

\usepackage{hyperref}
\hypersetup{
    colorlinks=true,
    linkcolor=blue,
    filecolor=magenta,      
    urlcolor=cyan,
}
\usepackage{geometry}

\usepackage{pgf,tikz,pgfplots}
\pgfplotsset{compat=1.18}
\usetikzlibrary{arrows}
\usetikzlibrary{matrix}
\usetikzlibrary{arrows}

\newtheorem{thm}{Theorem}[section]
\newtheorem{lem}[thm]{Lemma}
\newtheorem{prop}[thm]{Proposition}
\newtheorem{ex}[thm]{Example}
\newtheorem{cor}[thm]{Corollary}
\newtheorem{notation}[thm]{Notation}

\theoremstyle{definition}
\newtheorem{definition}[thm]{Definition}
\newtheorem{obs}[thm]{Observation}
\newtheorem{rem}[thm]{Remark}

\newcommand{\blackged}{\hfill$\blacksquare$}
\newcommand{\whiteged}{\hfill$\square$}
\newcounter{proofcount}

\renewenvironment{proof}[1][\proofname.]{\par
  \ifnum \theproofcount>0 \pushQED{\whiteged} \else \pushQED{\blackged} \fi%
  \refstepcounter{proofcount}
  \normalfont 
  \trivlist
  \item[\hskip\labelsep
        \itshape
    {\bf\em #1}]\ignorespaces
}{%
  \addtocounter{proofcount}{-1}
  \popQED\endtrivlist
}

\newcommand{\defeq}{\coloneqq}

\newcommand{\N}{\mathbb{N}}
\newcommand{\Z}{\mathbb{Z}}
\newcommand{\R}{\mathbb{R}}
\providecommand{\C}{}
\renewcommand{\C}{\mathbb{C}}
\newcommand{\tens}{\otimes}
\newcommand{\cH}{\mathcal{H}}
\newcommand{\sig}[1]{[#1]}
\newcommand{\fdg}[1]{|#1|}
\newcommand{\ds}{\circ}
\newcommand{\id}{\mathrm{id}}
\newcommand{\ev}{\mathrm{ev}}
\newcommand{\coev}{\mathrm{coev}}
\newcommand{\ptr}{\mathrm{pTr}}
\newcommand{\tr}{\mathrm{Tr}}
\newcommand{\act}{\triangleright}
\newcommand{\Calg}{\mathcal{C}_{\mathrm{alg}}}
\newcommand{\Cclass}{\mathcal{C}^{\mathrm{class}}}
\newcommand{\Ccalg}{\mathcal{C}^{\mathrm{class}}_{\mathrm{alg}}}
\newcommand{\Csq}{\mathcal{C}}
\newcommand{\xd}{\mathrm{d}}
\newcommand{\one}{\mathbf{1}}

\begin{document}


\title{\textbf{Compositional Quantum Field Theory:\\ An axiomatic presentation}}
\author[1]{Robert Oeckl\footnote{email: robert@matmor.unam.mx}}
\author[1,2]{Juan Orendain Almada\footnote{email: jxo245@case.edu}}
\affil[1]{Centro de Ciencias Matemáticas,
  Universidad Nacional Autónoma de México,
  C.P.~58190, Morelia, Michoacán, Mexico}
 \affil[2]{Mathematics, Applied Mathematics and Statistics Department, Case Western Reserve University, Cleveland, Ohio, USA}
\date{UNAM-CCM-2022-1\\ 22 August 2022\\ 12 June 2023 (v2)\\ 31 January 2024 (v3)}

\maketitle

\vspace{\stretch{1}}

\begin{abstract}
We introduce Compositional Quantum Field Theory (CQFT) as an axiomatic model of Quantum Field Theory, based on the principles of locality and compositionality. Our model is a refinement of the axioms of General Boundary Quantum Field Theory, and is phrased in terms of correspondences between certain commuting diagrams of gluing identifications between manifolds and corresponding commuting diagrams of state-spaces and linear maps, thus making it amenable to formalization in terms of involutive symmetric monoidal functors and operad algebras. The underlying novel framework for gluing leads to equivariance of CQFT. We study CQFTs in dimension 2 and the algebraic structure they define on open and closed strings. We also consider, as a particular case, the compositional structure of 2-dimensional pure quantum Yang-Mills theory.

\end{abstract}

\vspace{\stretch{1}}
\newpage

\tableofcontents


\section{Introduction}

Compositionality describes how information about complex systems can be pieced together from information about their simpler components. In this paper we propose an axiomatic system describing Quantum Field Theory (QFT) following the principles of compositionality. 

A compositional model for QFT should allow for the study of states, amplitudes and observables in "small regions" of spacetime, in such a way that when two of these regions are glued into a "more complex region" the physical information of the larger region can be deduced from the physical information in the smaller ones. Pictorially, if we know the physics of two regions: 
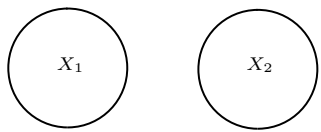
\begin{center}

	\tikzset{every picture/.style={line width=0.75pt}} 
	
	\begin{tikzpicture}[x=0.75pt,y=0.75pt,yscale=-1,xscale=1]
		
		\draw  [draw opacity=0] (246.57,183) .. controls (230.2,182.76) and (217,169.42) .. (217,153) .. controls (217,136.48) and (230.36,123.08) .. (246.86,123) -- (247,153) -- cycle ; \draw  [color={rgb, 255:red, 0; green, 0; blue, 0 }  ,draw opacity=1 ] (246.57,183) .. controls (230.2,182.76) and (217,169.42) .. (217,153) .. controls (217,136.48) and (230.36,123.08) .. (246.86,123) ;  
		\draw  [draw opacity=0] (246.17,123.01) .. controls (246.26,123.01) and (246.34,123.01) .. (246.42,123.01) .. controls (262.99,122.69) and (276.68,135.86) .. (276.99,152.42) .. controls (277.31,168.99) and (264.14,182.68) .. (247.58,182.99) .. controls (247.27,183) and (246.97,183) .. (246.67,183) -- (247,153) -- cycle ; \draw  [color={rgb, 255:red, 0; green, 0; blue, 0 }  ,draw opacity=1 ] (246.17,123.01) .. controls (246.26,123.01) and (246.34,123.01) .. (246.42,123.01) .. controls (262.99,122.69) and (276.68,135.86) .. (276.99,152.42) .. controls (277.31,168.99) and (264.14,182.68) .. (247.58,182.99) .. controls (247.27,183) and (246.97,183) .. (246.67,183) ;  
		\draw  [draw opacity=0] (342.43,183.66) .. controls (326.06,183.43) and (312.87,170.09) .. (312.87,153.67) .. controls (312.87,137.15) and (326.22,123.74) .. (342.73,123.67) -- (342.87,153.67) -- cycle ; \draw  [color={rgb, 255:red, 0; green, 0; blue, 0 }  ,draw opacity=1 ] (342.43,183.66) .. controls (326.06,183.43) and (312.87,170.09) .. (312.87,153.67) .. controls (312.87,137.15) and (326.22,123.74) .. (342.73,123.67) ;  
		\draw  [draw opacity=0] (342.04,123.68) .. controls (342.12,123.68) and (342.21,123.67) .. (342.29,123.67) .. controls (358.86,123.35) and (372.54,136.53) .. (372.86,153.09) .. controls (373.18,169.66) and (360.01,183.34) .. (343.44,183.66) .. controls (343.14,183.67) and (342.84,183.67) .. (342.54,183.67) -- (342.87,153.67) -- cycle ; \draw  [color={rgb, 255:red, 0; green, 0; blue, 0 }  ,draw opacity=1 ] (342.04,123.68) .. controls (342.12,123.68) and (342.21,123.67) .. (342.29,123.67) .. controls (358.86,123.35) and (372.54,136.53) .. (372.86,153.09) .. controls (373.18,169.66) and (360.01,183.34) .. (343.44,183.66) .. controls (343.14,183.67) and (342.84,183.67) .. (342.54,183.67) ;  
		
		\draw (240,146) node [anchor=north west][inner sep=0.75pt]  [font=\scriptsize]  {$X_{1}$};
		\draw (335.6,146) node [anchor=north west][inner sep=0.75pt]  [font=\scriptsize]  {$X_{2}$};

	\end{tikzpicture}
	
\end{center}
we should be able to deduce, from that, the physics of the region:
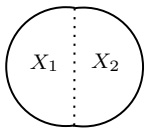
\begin{center}

	\tikzset{every picture/.style={line width=0.75pt}} 
	
	\begin{tikzpicture}[x=0.75pt,y=0.75pt,yscale=-1,xscale=1]
		
		\draw  [draw opacity=0] (354.72,221.69) .. controls (353.31,221.89) and (351.87,222) .. (350.4,222) .. controls (333.83,222) and (320.4,208.57) .. (320.4,192) .. controls (320.4,175.43) and (333.83,162) .. (350.4,162) .. controls (351.96,162) and (353.5,162.12) .. (355,162.35) -- (350.4,192) -- cycle ; \draw  [color={rgb, 255:red, 0; green, 0; blue, 0 }  ,draw opacity=1 ] (354.72,221.69) .. controls (353.31,221.89) and (351.87,222) .. (350.4,222) .. controls (333.83,222) and (320.4,208.57) .. (320.4,192) .. controls (320.4,175.43) and (333.83,162) .. (350.4,162) .. controls (351.96,162) and (353.5,162.12) .. (355,162.35) ;  
		\draw  [draw opacity=0] (354.88,162.54) .. controls (356.17,162.34) and (357.49,162.23) .. (358.82,162.21) .. controls (375.39,161.89) and (389.08,175.06) .. (389.39,191.62) .. controls (389.71,208.19) and (376.54,221.88) .. (359.98,222.19) .. controls (357.98,222.23) and (356.02,222.07) .. (354.12,221.74) -- (359.4,192.2) -- cycle ; \draw  [color={rgb, 255:red, 0; green, 0; blue, 0 }  ,draw opacity=1 ] (354.88,162.54) .. controls (356.17,162.34) and (357.49,162.23) .. (358.82,162.21) .. controls (375.39,161.89) and (389.08,175.06) .. (389.39,191.62) .. controls (389.71,208.19) and (376.54,221.88) .. (359.98,222.19) .. controls (357.98,222.23) and (356.02,222.07) .. (354.12,221.74) ;  
		\draw [color={rgb, 255:red, 0; green, 0; blue, 0 }  ,draw opacity=1 ] [dash pattern={on 0.84pt off 2.51pt}]  (354.88,162.54) -- (354.72,221.69) ;
		
		\draw (330.75,184.2) node [anchor=north west][inner sep=0.75pt]  [font=\footnotesize]  {$X_{1}$};
		\draw (361.75,183.95) node [anchor=north west][inner sep=0.75pt]  [font=\footnotesize]  {$X_{2}$};

	\end{tikzpicture}
	
\end{center}
obtained by gluing the smaller regions. In this paper, by physics of a region we will mean state spaces and amplitude maps. State spaces are associated to boundaries of regions, and more generally to hypersurfaces. State spaces should compose in a way compatible with the way we glue bounding regions. In particular, we should be able to piece together the state space of the boundary of a region from any decomposition as gluing of smaller pieces. Pictorially, if we know the state spaces of two "smaller hypersurfaces":

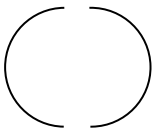
\begin{center}
	
	\tikzset{every picture/.style={line width=0.75pt}} 
	
	\begin{tikzpicture}[x=0.75pt,y=0.75pt,yscale=-1,xscale=1]
		
		\draw  [draw opacity=0] (293.93,160.66) .. controls (277.56,160.43) and (264.36,147.09) .. (264.36,130.67) .. controls (264.36,114.15) and (277.72,100.74) .. (294.22,100.67) -- (294.36,130.67) -- cycle ; \draw  [color={rgb, 255:red, 0; green, 0; blue, 0 }  ,draw opacity=1 ] (293.93,160.66) .. controls (277.56,160.43) and (264.36,147.09) .. (264.36,130.67) .. controls (264.36,114.15) and (277.72,100.74) .. (294.22,100.67) ;  
		\draw  [draw opacity=0] (306.87,100.68) .. controls (306.95,100.68) and (307.04,100.67) .. (307.12,100.67) .. controls (323.69,100.35) and (337.37,113.53) .. (337.69,130.09) .. controls (338.01,146.66) and (324.84,160.34) .. (308.27,160.66) .. controls (307.97,160.67) and (307.67,160.67) .. (307.37,160.67) -- (307.7,130.67) -- cycle ; \draw  [color={rgb, 255:red, 0; green, 0; blue, 0 }  ,draw opacity=1 ] (306.87,100.68) .. controls (306.95,100.68) and (307.04,100.67) .. (307.12,100.67) .. controls (323.69,100.35) and (337.37,113.53) .. (337.69,130.09) .. controls (338.01,146.66) and (324.84,160.34) .. (308.27,160.66) .. controls (307.97,160.67) and (307.67,160.67) .. (307.37,160.67) ;

	\end{tikzpicture}

\end{center}
we should be able to reconstruct, from that information, the state space of the "larger hypersurface":
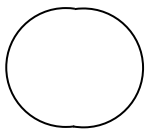
\begin{center}

	\tikzset{every picture/.style={line width=0.75pt}} 
	
	\begin{tikzpicture}[x=0.75pt,y=0.75pt,yscale=-1,xscale=1]
		
		\draw  [draw opacity=0] (226.41,160.73) .. controls (225,160.93) and (223.56,161.03) .. (222.1,161.03) .. controls (205.53,161.03) and (192.1,147.6) .. (192.1,131.03) .. controls (192.1,114.46) and (205.53,101.03) .. (222.1,101.03) .. controls (223.66,101.03) and (225.2,101.15) .. (226.7,101.38) -- (222.1,131.03) -- cycle ; \draw  [color={rgb, 255:red, 0; green, 0; blue, 0 }  ,draw opacity=1 ] (226.41,160.73) .. controls (225,160.93) and (223.56,161.03) .. (222.1,161.03) .. controls (205.53,161.03) and (192.1,147.6) .. (192.1,131.03) .. controls (192.1,114.46) and (205.53,101.03) .. (222.1,101.03) .. controls (223.66,101.03) and (225.2,101.15) .. (226.7,101.38) ;  
		\draw  [draw opacity=0] (226.58,101.57) .. controls (227.87,101.38) and (229.18,101.26) .. (230.52,101.24) .. controls (247.09,100.92) and (260.77,114.09) .. (261.09,130.66) .. controls (261.41,147.22) and (248.24,160.91) .. (231.67,161.23) .. controls (229.67,161.27) and (227.72,161.11) .. (225.82,160.77) -- (231.1,131.23) -- cycle ; \draw  [color={rgb, 255:red, 0; green, 0; blue, 0 }  ,draw opacity=1 ] (226.58,101.57) .. controls (227.87,101.38) and (229.18,101.26) .. (230.52,101.24) .. controls (247.09,100.92) and (260.77,114.09) .. (261.09,130.66) .. controls (261.41,147.22) and (248.24,160.91) .. (231.67,161.23) .. controls (229.67,161.27) and (227.72,161.11) .. (225.82,160.77) ;

	\end{tikzpicture}
	
\end{center}
obtained by gluing along corresponding boundary components. A compositional model of QFT should thus describe how to translate the arithmetic of gluing regions and hypersurfaces, into physically significant data reflecting these operations. Observe that any such model should automatically be required to be \emph{local}, in the sense that we should be able to isolate the physics of a region or a hypersurface from all other regions or hypersurfaces. 

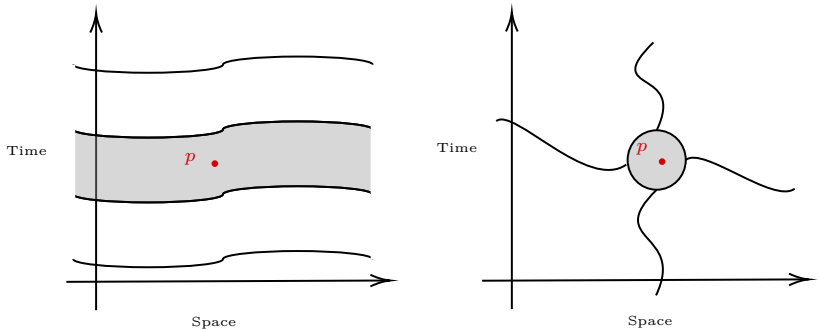
\begin{figure}
	\centering
\tikzset{every picture/.style={line width=0.75pt}} 

\begin{tikzpicture}[x=0.75pt,y=0.75pt,yscale=-1.5,xscale=1.5]
	
	\draw [color={rgb, 255:red, 0; green, 0; blue, 0 }  ,draw opacity=1 ]   (60.58,230.13) -- (60.58,132.13) ;
	\draw [shift={(60.58,130.13)}, rotate = 90] [color={rgb, 255:red, 0; green, 0; blue, 0 }  ,draw opacity=1 ][line width=0.75]    (6.56,-1.97) .. controls (4.17,-0.84) and (1.99,-0.18) .. (0,0) .. controls (1.99,0.18) and (4.17,0.84) .. (6.56,1.97)   ;
	\draw   (52.89,191.13) .. controls (52.89,192.63) and (64.16,193.85) .. (78.06,193.85) .. controls (91.95,193.85) and (103.22,192.63) .. (103.22,191.13) .. controls (103.22,189.62) and (114.49,188.4) .. (128.39,188.4) .. controls (142.29,188.4) and (153.56,189.62) .. (153.56,191.13) -- (153.56,212.94) .. controls (153.56,211.43) and (142.29,210.21) .. (128.39,210.21) .. controls (114.49,210.21) and (103.22,211.43) .. (103.22,212.94) .. controls (103.22,214.45) and (91.95,215.67) .. (78.06,215.67) .. controls (64.16,215.67) and (52.89,214.45) .. (52.89,212.94) -- cycle ;
	\draw  [color={rgb, 255:red, 0; green, 0; blue, 0 }  ,draw opacity=1 ][fill={rgb, 255:red, 155; green, 155; blue, 155 }  ,fill opacity=0.4 ] (52.89,169.31) .. controls (52.89,170.82) and (64.16,172.04) .. (78.06,172.04) .. controls (91.95,172.04) and (103.22,170.82) .. (103.22,169.31) .. controls (103.22,167.81) and (114.49,166.59) .. (128.39,166.59) .. controls (142.29,166.59) and (153.56,167.81) .. (153.56,169.31) -- (153.56,191.13) .. controls (153.56,189.62) and (142.29,188.4) .. (128.39,188.4) .. controls (114.49,188.4) and (103.22,189.62) .. (103.22,191.13) .. controls (103.22,192.63) and (91.95,193.85) .. (78.06,193.85) .. controls (64.16,193.85) and (52.89,192.63) .. (52.89,191.13) -- cycle ;
	\draw   (52.89,147.5) .. controls (52.89,149.01) and (64.16,150.23) .. (78.06,150.23) .. controls (91.95,150.23) and (103.22,149.01) .. (103.22,147.5) .. controls (103.22,145.99) and (114.49,144.77) .. (128.39,144.77) .. controls (142.29,144.77) and (153.56,145.99) .. (153.56,147.5) -- (153.56,169.31) .. controls (153.56,167.81) and (142.29,166.59) .. (128.39,166.59) .. controls (114.49,166.59) and (103.22,167.81) .. (103.22,169.31) .. controls (103.22,170.82) and (91.95,172.04) .. (78.06,172.04) .. controls (64.16,172.04) and (52.89,170.82) .. (52.89,169.31) -- cycle ;
	\draw [color={rgb, 255:red, 255; green, 255; blue, 255 }  ,draw opacity=1 ][line width=1.5]    (153.56,147.5) -- (153.56,212.94) ;
	\draw [color={rgb, 255:red, 255; green, 255; blue, 255 }  ,draw opacity=1 ][line width=1.5]    (52.89,147.5) -- (52.89,212.94) ;
	\draw [color={rgb, 255:red, 0; green, 0; blue, 0 }  ,draw opacity=1 ]   (50.58,220.58) -- (157.36,220.14) ;
	\draw [shift={(159.36,220.13)}, rotate = 179.77] [color={rgb, 255:red, 0; green, 0; blue, 0 }  ,draw opacity=1 ][line width=0.75]    (6.56,-1.97) .. controls (4.17,-0.84) and (1.99,-0.18) .. (0,0) .. controls (1.99,0.18) and (4.17,0.84) .. (6.56,1.97)   ;
	\draw  [color={rgb, 255:red, 214; green, 2; blue, 6 }  ,draw opacity=1 ][fill={rgb, 255:red, 214; green, 2; blue, 6 }  ,fill opacity=1 ] (99.73,180.8) .. controls (99.73,180.37) and (100.08,180.02) .. (100.51,180.02) .. controls (100.94,180.02) and (101.29,180.37) .. (101.29,180.8) .. controls (101.29,181.23) and (100.94,181.58) .. (100.51,181.58) .. controls (100.08,181.58) and (99.73,181.23) .. (99.73,180.8) -- cycle ;
	\draw [color={rgb, 255:red, 0; green, 0; blue, 0 }  ,draw opacity=1 ]   (200.36,229.76) -- (200.36,131.76) ;
	\draw [shift={(200.36,129.76)}, rotate = 90] [color={rgb, 255:red, 0; green, 0; blue, 0 }  ,draw opacity=1 ][line width=0.75]    (6.56,-1.97) .. controls (4.17,-0.84) and (1.99,-0.18) .. (0,0) .. controls (1.99,0.18) and (4.17,0.84) .. (6.56,1.97)   ;
	\draw [color={rgb, 255:red, 0; green, 0; blue, 0 }  ,draw opacity=1 ]   (190.36,220.2) -- (298.13,219.76) ;
	\draw [shift={(300.13,219.76)}, rotate = 179.77] [color={rgb, 255:red, 0; green, 0; blue, 0 }  ,draw opacity=1 ][line width=0.75]    (6.56,-1.97) .. controls (4.17,-0.84) and (1.99,-0.18) .. (0,0) .. controls (1.99,0.18) and (4.17,0.84) .. (6.56,1.97)   ;
	\draw  [fill={rgb, 255:red, 155; green, 155; blue, 155 }  ,fill opacity=0.4 ] (239.29,179.58) .. controls (239.29,174.05) and (243.67,169.58) .. (249.07,169.58) .. controls (254.47,169.58) and (258.84,174.05) .. (258.84,179.58) .. controls (258.84,185.1) and (254.47,189.58) .. (249.07,189.58) .. controls (243.67,189.58) and (239.29,185.1) .. (239.29,179.58) -- cycle ;
	\draw    (249.07,169.58) .. controls (259.29,148.02) and (229.73,157.8) .. (247.96,140.02) ;
	\draw    (248.84,225.13) .. controls (259.07,203.58) and (230.84,207.36) .. (249.07,189.58) ;
	\draw    (295.47,189.23) .. controls (288.36,194.57) and (264.62,174.69) .. (258.84,179.58) ;
	\draw    (238.84,181.13) .. controls (225.96,190.24) and (200.91,161.61) .. (195.13,166.5) ;
	\draw  [color={rgb, 255:red, 214; green, 2; blue, 6 }  ,draw opacity=1 ][fill={rgb, 255:red, 214; green, 2; blue, 6 }  ,fill opacity=1 ] (250.11,180.12) .. controls (250.11,179.69) and (250.46,179.35) .. (250.88,179.35) .. controls (251.31,179.35) and (251.66,179.69) .. (251.66,180.12) .. controls (251.66,180.55) and (251.31,180.9) .. (250.88,180.9) .. controls (250.46,180.9) and (250.11,180.55) .. (250.11,180.12) -- cycle ;
	
	\draw (89.36,176.29) node [anchor=north west][inner sep=0.75pt]  [font=\scriptsize,color={rgb, 255:red, 214; green, 2; blue, 6 }  ,opacity=1 ]  {$p$};
	\draw (241.26,172.88) node [anchor=north west][inner sep=0.75pt]  [font=\scriptsize,color={rgb, 255:red, 214; green, 2; blue, 6 }  ,opacity=1 ]  {$p$};
	\draw (91.67,231.33) node [anchor=north west][inner sep=0.75pt]  [font=\tiny] [align=left] {Space};
	\draw (29.33,173.67) node [anchor=north west][inner sep=0.75pt]  [font=\tiny] [align=left] {Time};
	\draw (174,172.67) node [anchor=north west][inner sep=0.75pt]  [font=\tiny] [align=left] {Time};
	\draw (238.33,231) node [anchor=north west][inner sep=0.75pt]  [font=\tiny] [align=left] {Space};
\end{tikzpicture}
	\caption{Compositional locality: With \emph{temporal locality} spacetime regions are time-slices (left-hand side), in full \emph{spacetime locality}, spacetime regions can be arbitrary (right-hand side). With the latter notion we can describe physics in an arbitrarily small spacetime neighborhood of an event $p$. This is crucial for implementing Wilsonian renormalization, as known from Lattice Gauge Theory. It would also be essential for quantum gravity.}
	\label{fig:locality}
\end{figure}

Models for or inspired by QFT satisfying some of the above observations abound in the mathematical literature. Sheaf-theoretical models such as Algebraic Quantum Field Theory (AQFT), describe observables compositionally, but usually assume a fixed space of states, or no explicit compositional descriptions of state spaces is provided. Topological Quantum Field Theory (TQFT) on the other hand provides compositional models of state spaces and amplitudes. However, due to its in-out structure, compositionality is coherently directed, and thus usually associated with time evolution only. This means that instead of full spacetime locality (Figure~\ref{fig:locality} right-hand side) only temporal locality can be implemented (Figure~\ref{fig:locality} left-hand side), somewhat similar to traditional formulations of QFT in curved spacetime in the physics literature.\footnote{This particular shortcoming of TQFT is remedied in certain versions of \emph{extended TQFT}.} Moreover, the application of TQFT to QFT faces important obstacles. One is the restriction of state spaces in TQFT to be finite-dimensional, in contrast to the requirements of QFT. Even more serious is the lack of proper rules for the prediction of measurement outcomes. This is fine (and even expected) for a purely mathematical framework, but indispensable when modeling a quantum theory is intended.\footnote{Unfortunately, the rules of the standard formulation of quantum theory do not make sense for most objects appearing in TQFT. Only for globally hyperbolic manifolds or time-slices of these can these rules be applied, precluding spacetime-local compositionality for physical quantities.} On the other hand, TQFT has very attractive mathematical features, not least its functorial formulation.

In this paper we focus on General Boundary Quantum Field Theory (GBQFT), developed by the first author in the framework of the General Boundary Formulation of Quantum Theory (GBF) \cite{Oe:boundary,Oe:gbqft,Oe:holomorphic,Oe:freefermi,Oe:feynobs,CoOe:locgenvac}. GBQFT is a model for QFT, inspired by ideas of extended TQFT but implementing a notion of locality closer to sheaf-theoretical models such as factorization homology. One can think of GBQFT as a minimal incarnation of the ideas of compositionality and locality described above, having simple versions of realistic QFTs as examples. This is the first installment of a series papers, whose main goal is to provide a precise categorical interpretation of the GBQFT axioms and related theories.

Compositional Quantum Field Theory (CQFT) is an axiomatic system refining the axioms of GBQFT, formulated as a correspondence between certain commuting diagrams, representing the algebra of spacetime pieces, and certain other commuting diagrams, representing the corresponding physics. One elementary but novel characteristic of our approach is the way we model the operation of gluing. We consider specific witnesses of gluings, which we call gluing functions, and we account for gluing ambiguities by considering isomorphisms as part of our algebraic data. Gluing on the nose and implementing up-to-isomorphism conditions as another part of the data, suggests the use of symmetric monoidal double categories, and more specifically decorated or structured cospans, see \cite{Fong, BaezCourser} to model the compositional structure of spacetime. The work presented in this paper should be thought as a step previous to this. We consider the appropriate categorical structures in the followup paper \cite{RobertJuan2}. CQFT considers key observations not considered in previous GBQFT installments, it is a first step towards categorification of the ideas of GBQFT, and its formulation makes certain pieces of structure, previously unknown to us, evident, e.g.\ in Section~\ref{sec:2dim} we provide the state space of an interval in quantum Yang-Mills theory with area with the structure of a centrally tracial Hilbert *-algebra.

\subsection{General Boundary Quantum Field Theory}

Quantum Field Theory (QFT) is a framework for constructing and presenting relativistic quantum theories. It is used to describe phenomena in different areas of physics, but most prominently in high-energy physics, where it forms the basis of the Standard Model of Elementary Particle Physics. However, in spite of its quantitative predictive success, no rigorous mathematical formulation of QFT has been found to date. Prominent approaches (both emerging in the 1960s) at axiomatizing and making rigorous QFT are the Wightman Axioms \cite{StWi:pct} and Algebraic Quantum Field Theory \cite{HaKa:aqft}.
For physicists, a problem even more important than making QFT rigorous is the incompatibility of QFT with the principles of General Relativity. Also here, one might hope that an axiomatic approach would help. 

Topological Quantum Field Theory (TQFT), which emerged in the late 1980s, was originally inspired by QFT. It was thought at the time that it might lead to a framework addressing both of the mentioned problems. There are a number of works, going into the 1990s where it is proposed that TQFT might form the basis of a theory of quantum gravity. While hugely successful in mathematics, the early expectations for TQFT to catalyze developments in physics have not been born out at the time. A new attempt, inspired by TQFT, was started in 2003 with the programme of the general boundary formulation of quantum theory (GBF) \cite{Oe:catandclock,Oe:boundary}, meant as an axiomatic approach not only to QFT, but to quantum theory more generally. The most crucial aspect of the GBF compared to previous approaches based on TQFT is, that it comprehensively establishes a relation between the mathematical formalism and physical predictions \cite{Oe:gbqft,Oe:posfound}. We shall refer to the main axiomatic system that has been proposed so far for the GBF as General Boundary Quantum Field Theory (GBQFT) \cite{Oe:gbqft}.

One important application of TQFT is to produce invariants of manifolds. For most physical applications, spacetimes have the topology of $\R^4$. Only for certain astrophysical applications involving black holes or for cosmological applications (considering the universe as a whole) might we consider other topologies, and the interesting choices would be very limited. As for pieces of spacetime, these are principally of interest to implement locality, and describe physics in a “small” region. Their typical topology would be that of a ball. This also means that we need to be able to glue along parts of boundaries, as in \emph{gluing two balls to one ball} and \emph{self-gluing a ball to another ball}. To this end not only regions need boundaries, but also hypersurfaces. This suggests an extended TQFT setting. However, it is not clear in general what data should be associated to manifolds of codimension $2$ or higher. What is more, there does not seem to be any immediate physical motivation for such data. GBQFT as considered here allows for the implicit appearance of manifolds of codimension $2$ or larger, but is minimalistic in not explicitly associating any data to them. Another “unphysical” aspect of TQFT for the present purposes is the fact that regions are \emph{cobordisms}. Thus, the morphism associated to the cobordism in the category of state spaces represents a time evolution map from initial to final states. However, as explained previously, this is not tenable if we want to implement spacetime locality in QFT. Thus, the consequential move, implemented in GBQFT, is to abolish the distinction between “in” and “out” on the boundary. This also means that \emph{orientation} plays a more important role.

Returning to the question of values associated to closed manifolds, these are not of any intrinsic physical interest, except possibly for yielding renormalization constants.\footnote{In QFT these values generally arise from a path integral. While the functional dependence of the path integral \emph{on insertions} such as sources or observables is physically important, the pure path integral is not. Physical quantities arise from quotients of the path integral with insertions by the path integral without insertions. In the present article we will not consider such insertions.} What is of interest, are the linear maps associated to manifolds with boundary. This implies more freedom in “adjusting” values associated to closed manifolds. In fact, we shall assume that the value associated to a closed manifold is always $1$, as has proven fruitful in constructive applications of GBQFT, see e.g.~\cite{Oe:holomorphic}. Related to this, we explicitly admit in GBQFT a \emph{gluing anomaly} in the sense of Turaev \cite{Tur:qinv}. We mention however, that for the purposes of the present paper this assumption could easily be relaxed.

In most applications of TQFT the manifolds are topological manifolds or essentially smooth manifolds.\footnote{By essentially smooth we also subsume those cases where the manifolds carry a (usually Riemannian) metric, but the TQFT only sees at most a finite number of degrees of freedom of the metric on top of the smooth structure.} However, in spite of the name “topological” the possibility of adding more structure to the manifolds was foreseen from the beginning \cite{Ati:tqft}. The application to QFT that we are interested in here certainly suggests that we should consider manifolds with (Lorentzian) metric structure. The bare axiomatic system of GBQFT is deliberately agnostic in this respect. Throughout the paper we shall use the term \textbf{region} to refer to an oriented manifold of maximal dimension $n$, possibly with boundary, and we shall use the term \textbf{hypersurface} to refer to an oriented manifold of dimension $n-1$, also possibly with boundary. These manifolds may carry additional structure, and we suppose notions of gluing that take this additional structure into account. For more motivation in this respect see for example \cite{Oe:2dqym,RuSz:areaqft,StolzTeichnerSusy}. 

\subsection{Axioms}

We now present the axioms for GBQFT. Quantum field theories come in two flavors: \emph{bosonic} and \emph{fermionic} ones. As has been known for a long time, the fermionic case translates to a non-trivial symmetric structure on the relevant monoidal category, while the bosonic case translates to the trivial symmetric structure provided by the trivial flip isomorphism. In the fermionic case we use a $\Z_2$-grading to specify the symmetric structure. This is straightforward enough for an implementation in TQFT. What is rather more surprising is that to physically correctly encode fermionic theories, we also need to replace the Hilbert spaces ubiquitous in quantum theory with Krein spaces. The latter are indefinite (but non-degenerate) inner product spaces. As far as we are aware this was discovered only in the work \cite{Oe:freefermi}. In the following we present the axioms of GBQFT in a version essentially as in that work. 

\begin{definition}\label{AFdefinitionClassic}
A GBQFT consists of the following data:

\begin{description}

\item[\textbf{(T1)}]\label{(T1)} A complex
separable Krein superspace $\cH_\Sigma$, called the \emph{state space} of
$\Sigma$, for each hypersurface $\Sigma$. We denote its indefinite inner product by
$\langle\cdot,\cdot\rangle_\Sigma$. We require $\cH_\emptyset=\C$.

\item[\textbf{(T1b)}]\label{(T1B)} An anti-linear graded isometric isomorphism $\iota_\Sigma:\cH_\Sigma\to\cH_{\overline{\Sigma}}$ satisfying the equation $\iota_{\overline{\Sigma}}\iota_\Sigma=id_{\cH_\Sigma}$ for each hypersurface $\Sigma$.

\item[\textbf{(T2)}]\label{(T2)} Given a decomposition of a hypersurface $\Sigma$ as a union $\Sigma=\Sigma_1\cup...\cup\Sigma_n$ of hypersurfaces, a graded epimorphism $\tau_{\Sigma_1,...,\Sigma_n;\Sigma}:\cH_{\Sigma_1}\otimes...\otimes\cH_{\Sigma_n}\to \cH_\Sigma$. We require $\tau_{\Sigma_1,...,\Sigma_n;\Sigma}$ to be associative. If the decomposition is disjoint we also require $\tau_{\Sigma_1,...,\Sigma_n;\Sigma}$ to be a graded isometric isomorphism.

\item[\textbf{(T2a)}\label{(T2a)}]
The map $\tau_{\Sigma_2,\Sigma_1;\Sigma}^{-1}\circ \tau_{\Sigma_1,\Sigma_2;\Sigma}:\cH_{\Sigma_1}\tens\cH_{\Sigma_2}\to\cH_{\Sigma_2}\tens\cH_{\Sigma_1}$ is the graded transposition,
\begin{equation}\label{eq:TensorTransformationSymmetry}
     \psi_1\tens\psi_2\mapsto (-1)^{\fdg{\psi_1}\cdot\fdg{\psi_2}}\psi_2\tens\psi_1 .
\end{equation}

\item[\textbf{(T2b)}]\label{(T2B)} Orientation change and decomposition are compatible. In particular, for a disjoint decomposition of hypersurfaces $\Sigma=\Sigma_1\sqcup\Sigma_2$ we have,
\begin{equation}
\tau_{\overline{\Sigma}_1,\overline{\Sigma}_2;\overline{\Sigma}}
\left((\iota_{\Sigma_1}\tens\iota_{\Sigma_2})(\psi_1\tens\psi_2)\right)
  =\iota_\Sigma\left(\tau_{\Sigma_2,\Sigma_1;\Sigma}(\psi_2\tens\psi_1)\right) .
\end{equation}

\item[\textbf{(T4)}]\label{(T4)} A graded map $\rho_X:\cH_{\partial X}\to \C$ for every region $X$. We call $\rho_X$ the amplitude map associated to $X$. We require $\rho_\emptyset=id_\C$, and we also require that the amplitude $\rho_{\overline{X}}$ of the region $\overline{X}$ obtained from $X$ by inverting orientation, is the conjugate composition $\overline{\rho_X\iota_{\partial \overline{X}}}$.

\item[\textbf{(T3x)}]\label{(T3x)} Let $\Sigma$ be a hypersurface. Let $\hat{\Sigma}$ be its slice region. The boundary $\partial \hat{\Sigma}$ of $\hat{\Sigma}$ decomposes as the union $\overline{\Sigma}\cup\Sigma$ of two copies of $\Sigma$ with opposite orientation. The bilinear map $(\cdot,\cdot)_\Sigma:\cH_{\overline{\Sigma}}\otimes\cH_\Sigma\to\C$ given by $\rho_{\hat{\Sigma}}\tau$ is related to the inner product of $\cH_\Sigma$ via the equation $\langle\cdot,\cdot\rangle_\Sigma=(\iota_\Sigma(\cdot),\cdot)_\Sigma$. 

\item[\textbf{(T5a)}]\label{(T5a)} Let $X_1,X_2$ be regions. Let $X=X_1\sqcup X_2$ be their disjoint union. For all $\psi_1\in\cH_{\partial X_1}$ and $\psi_2\in\cH_{\partial X_2}$ the following equation holds:
\begin{equation}
    \rho_X\tau_{\partial X_1,\partial X_2}(\psi_1\otimes\psi_2)=\rho_{X_1}(\psi_1)\rho_{X_2}(\psi_2) .
\end{equation}

\item[\textbf{(T5b)}]\label{(T5b)} Let $X$ be a region with its boundary decomposed as $\partial X=\Lambda\cup\Sigma\cup\overline{\Sigma}$. Let $X'$ denote the gluing of $X$ to itself along $\Sigma$ and $\overline{\Sigma}$. Suppose $X'$ is a region. There exists a non-zero complex number $c_{\Lambda,\Sigma,\overline{\Sigma}}$, which we call the gluing anomaly of the gluing of $X$ into $X'$, such that for any orthonormal basis $\{\zeta_k\}_{k\in I}$ of $\cH_\Sigma$ and for every $\psi\in\cH_{X'}$ the following equation holds:
\begin{equation}
    \rho_{\partial X'}\tau(\psi)=c_{\Lambda,\Sigma,\overline{\Sigma}}\sum_{k\in I}(-1)^{\sig{\zeta_k}}\rho_X\tau_{\Lambda,\Sigma,\overline{\Sigma}}(\psi\tens\zeta_k\tens \iota_\Sigma(\zeta_k)) .
    \label{eq:comprule}
\end{equation}
\end{description}
\end{definition}

We explain some notation used in the axioms. A superspace is a $\Z_2$-graded vector space. A Krein superspace is a Krein space that also carries the structure of a superspace. To make this explicit we write $\cH_{\Sigma}=\cH_{\Sigma}^0\oplus\cH_{\Sigma}^1$ and $|\psi|=0,1$ for $\psi\in\cH_{\Sigma}^{0,1}$. The inner product and grading are compatible in the sense that $\langle \psi,\psi'\rangle_{\Sigma}=0$ if $|\psi|\neq |\psi'|$. Graded linear maps (homomorphisms) between superspaces are linear maps that preserve the grading. $\C$ is canonically a Krein superspace with the odd part trivial and the inner product $\langle 1,1\rangle=1$. In particular, $\rho_X(\psi)=0$ if $|\psi|=1$. An anti-linear graded isometric isomorphism between Krein superspaces is a map $\alpha$ satisfying $\langle\alpha(\psi),\alpha(\psi')\rangle=(-1)^{|\psi| |\psi'|}\overline{\langle\psi,\psi'\rangle}$.
An orthonormal basis of a Krein superspace is a decomposition $\cH=\cH^{+}\oplus\cH^{-}$ into complete positive-definite and negative-definite parts that are separately superspaces and an orthonormal basis for each part. For basis elements we use in Axiom~\textbf{(T5b)} the notation $[\zeta]=0$ if $\zeta\in \cH^{+}$ and $[\zeta]=1$ if $\zeta\in\cH^{-}$.

We proceed to comment on the difference between the version of the axioms presented here and that of the work \cite{Oe:freefermi}. The Axiom~\textbf{(T2)} of \cite{Oe:freefermi} is split here into two axioms, a more basic Axiom~\textbf{(T2)} and Axiom~\textbf{(T2a)}. This is convenient in allowing us to specifically refer to this more basic version of Axiom~\textbf{(T2)}, which was used previously in other works. What is more, we should expect that in certain theories, such as those linked to the representation theory of quantum groups only the more basic version holds, i.e., what we call Axiom~\textbf{(T2a)} would not hold there. Also, what we call here Axiom~\textbf{(T2b)} arises from a combination of the new Axiom~\textbf{(T2a)} and what was called Axiom~\textbf{(T2b)} in \cite{Oe:freefermi}. Conversely, Axiom~\textbf{(T2b)} in \cite{Oe:freefermi} arises from combining Axioms~\textbf{(T2a)} and \textbf{(T2b)} as presented here. The advantage of our presentation here is that Axiom~\textbf{(T2b)} as presented still makes sense, even if Axiom~\textbf{(T2a)} does not hold.

The axioms for GBQFT describe ways of coherently associating algebraic data, namely Krein superspaces and their maps, to composite geometric systems, namely hypersurfaces and regions obtained from simpler manifolds through the operations of gluing and taking disjoint unions. While the GBQFT axioms are meant to be applied in contexts where the manifolds carry various types of additional structure, we shall restrict in the present paper to the simplest case of interest: \emph{oriented topological manifolds}. In fact, one characteristic of our approach, to be explored in subsequent work, is that we wish to introduce a systematic way to add additional structure.

\begin{rem}
	In Definition~\ref{QFTAFfin}, we take an agnostic point of view on our notion of gluing, i.e.\ the axioms involving gluings of hypersurfaces and regions are purposely phrased in a vague notion of gluing. We will develop the specific notion of gluing of manifolds we will use throughout the paper and subsequent work in Section~\ref{gluingfunctionssubsection}. The gluing operation of regions used to state Axiom~\textbf{(T5b)} is of a specific kind: It describes the self-gluing of a region along two copies of the same hypersurface, with opposite orientations, contained in the boundary. By combining this operation with the operation of taking disjoint unions, the notion of gluing described in Axiom~\textbf{(T5b)} subsumes the usual gluing of cobordisms. We develop a formal notion of gluing expressing gluings appearing in Axiom~\textbf{(T5b)}, which will be easily generalized to situations where more structure is involved.
\end{rem}

A \textbf{slice region} $\hat{\Sigma}$, as appearing in Axiom~\textbf{(T3x)} of Definition \ref{QFTAFfin} is a special type of region that should be thought of as an infinitesimal thickening of $\Sigma$, which, under the action of a GBQFT relates inner products to amplitudes. Slice regions were introduced originally in \cite{Oe:gbqft} under the name “empty regions”. The following is a pictorial representation of a slice region on the circle $\Sigma$ with its usual orientation:
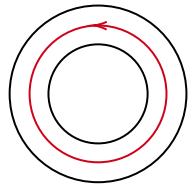
\begin{center}

	\tikzset{every picture/.style={line width=0.75pt}} 
	
	\begin{tikzpicture}[x=0.75pt,y=0.75pt,yscale=-1,xscale=1]
		
		\draw   (135.47,185.36) .. controls (135.47,160.74) and (155.43,140.78) .. (180.05,140.78) .. controls (204.67,140.78) and (224.63,160.74) .. (224.63,185.36) .. controls (224.63,209.98) and (204.67,229.95) .. (180.05,229.95) .. controls (155.43,229.95) and (135.47,209.98) .. (135.47,185.36) -- cycle ;
		\draw   (155.05,185.36) .. controls (155.05,171.55) and (166.24,160.36) .. (180.05,160.36) .. controls (193.86,160.36) and (205.05,171.55) .. (205.05,185.36) .. controls (205.05,199.17) and (193.86,210.36) .. (180.05,210.36) .. controls (166.24,210.36) and (155.05,199.17) .. (155.05,185.36) -- cycle ;
		\draw  [color={rgb, 255:red, 208; green, 2; blue, 27 }  ,draw opacity=1 ] (145.53,185.36) .. controls (145.53,166.29) and (160.98,150.84) .. (180.05,150.84) .. controls (199.12,150.84) and (214.57,166.29) .. (214.57,185.36) .. controls (214.57,204.43) and (199.12,219.89) .. (180.05,219.89) .. controls (160.98,219.89) and (145.53,204.43) .. (145.53,185.36) -- cycle ;
		\draw [color={rgb, 255:red, 208; green, 2; blue, 27 }  ,draw opacity=1 ]   (180.05,150.84) -- (179.86,150.83) ;
		\draw [shift={(177.86,150.79)}, rotate = 1.36] [color={rgb, 255:red, 208; green, 2; blue, 27 }  ,draw opacity=1 ][line width=0.75]    (6.56,-1.97) .. controls (4.17,-0.84) and (1.99,-0.18) .. (0,0) .. controls (1.99,0.18) and (4.17,0.84) .. (6.56,1.97)   ;

	\end{tikzpicture}
	
\end{center}
where the arrowed circle (red) represents $\Sigma$ and the outer annulus (non-arrowed circles) is meant to be thought of as having infinitesimally small width.

\begin{rem}
	Similarly to the case of gluings, the notion of slice region appearing in Definition~\ref{QFTAFfin} is intentionally phrased in a vague manner. In Section~\ref{sec:compareaxioms} we present a formal definition of slice region on a hypersurface, completely stated in terms of the desired formal properties a slice region should satisfy. In the context of the present paper, i.e.\ in the context of oriented topological manifolds, the slice region $\hat{\Sigma}$ of a hypersurface $\Sigma$, can be identified simply with the cylinder over $\Sigma$, as we shall see in Section~\ref{sliceregionssubsection}. Stating the definition of a slice region in terms of its formal properties will allow us to extend our notion of slice region to contexts where more spacetime structure is involved.
\end{rem}

\subsection{Examples}

In vanilla TQFT the gluing in codimension $0$ is conveniently described in terms of the gluing together of two distinct cobordisms (although it is also permissible to glue a cobordism to itself). In the present setting, where manifolds of codimension $1$ may have boundaries, it is more convenient to describe gluing in codimension $0$ in terms of two separate cases: A disjoint gluing as in Axiom~\textbf{(T5a)}, e.g.\ gluing two $n$-balls, along a common $n-1$-ball, and a self-gluing as in Axiom~\textbf{(T5b)}, e.g.\ self-gluing an $n$-ball into an $n$-ball. We illustrate these two examples in the following for GBQFT of dimension $2$.

\begin{ex}\label{ExampleTwoDisks}
	Assume we are given a GBQFT of dimension $2$. Let $X_1$ and $X_2$ be two copies of the standard disk $\mathbb{D}^2$. Let $X$ denote the disjoint union $X_1\sqcup X_2$ of $X_1$ and $X_2$. Consider the boundary decompositions of $X_1$ and $X_2$ as two intervals $I^i_\ell$ and $I^i_r$, $i=1,2$ glued together along their endpoints. Pictorially:
	\begin{center}

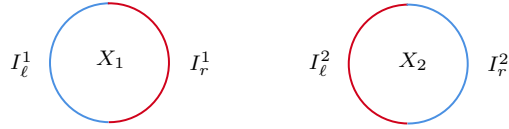
		\tikzset{every picture/.style={line width=0.75pt}} 
		
		\begin{tikzpicture}[x=0.75pt,y=0.75pt,yscale=-1,xscale=1]
			
			\draw  [draw opacity=0] (279.57,189) .. controls (263.2,188.76) and (250,175.42) .. (250,159) .. controls (250,142.48) and (263.36,129.08) .. (279.86,129) -- (280,159) -- cycle ; \draw  [color={rgb, 255:red, 74; green, 144; blue, 226 }  ,draw opacity=1 ] (279.57,189) .. controls (263.2,188.76) and (250,175.42) .. (250,159) .. controls (250,142.48) and (263.36,129.08) .. (279.86,129) ;  
			\draw  [draw opacity=0] (279.17,129.01) .. controls (279.26,129.01) and (279.34,129.01) .. (279.42,129.01) .. controls (295.99,128.69) and (309.68,141.86) .. (309.99,158.42) .. controls (310.31,174.99) and (297.14,188.68) .. (280.58,188.99) .. controls (280.27,189) and (279.97,189) .. (279.67,189) -- (280,159) -- cycle ; \draw  [color={rgb, 255:red, 208; green, 2; blue, 27 }  ,draw opacity=1 ] (279.17,129.01) .. controls (279.26,129.01) and (279.34,129.01) .. (279.42,129.01) .. controls (295.99,128.69) and (309.68,141.86) .. (309.99,158.42) .. controls (310.31,174.99) and (297.14,188.68) .. (280.58,188.99) .. controls (280.27,189) and (279.97,189) .. (279.67,189) ;  
			\draw  [draw opacity=0] (430.23,189.66) .. controls (413.86,189.43) and (400.67,176.09) .. (400.67,159.67) .. controls (400.67,143.15) and (414.02,129.74) .. (430.53,129.67) -- (430.67,159.67) -- cycle ; \draw  [color={rgb, 255:red, 208; green, 2; blue, 27 }  ,draw opacity=1 ] (430.23,189.66) .. controls (413.86,189.43) and (400.67,176.09) .. (400.67,159.67) .. controls (400.67,143.15) and (414.02,129.74) .. (430.53,129.67) ;  
			\draw  [draw opacity=0] (429.84,129.68) .. controls (429.92,129.68) and (430.01,129.67) .. (430.09,129.67) .. controls (446.66,129.35) and (460.34,142.53) .. (460.66,159.09) .. controls (460.98,175.66) and (447.81,189.34) .. (431.24,189.66) .. controls (430.94,189.67) and (430.64,189.67) .. (430.34,189.67) -- (430.67,159.67) -- cycle ; \draw  [color={rgb, 255:red, 74; green, 144; blue, 226 }  ,draw opacity=1 ] (429.84,129.68) .. controls (429.92,129.68) and (430.01,129.67) .. (430.09,129.67) .. controls (446.66,129.35) and (460.34,142.53) .. (460.66,159.09) .. controls (460.98,175.66) and (447.81,189.34) .. (431.24,189.66) .. controls (430.94,189.67) and (430.64,189.67) .. (430.34,189.67) ;  
			
			\draw (271.78,151.51) node [anchor=north west][inner sep=0.75pt]  [font=\footnotesize]  {$X_{1}$};
			\draw (424.22,152.18) node [anchor=north west][inner sep=0.75pt]  [font=\footnotesize]  {$X_{2}$};
			\draw (228.44,151.84) node [anchor=north west][inner sep=0.75pt]  [font=\footnotesize]  {$I_{\ell }^{1}$};
			\draw (318.78,152.18) node [anchor=north west][inner sep=0.75pt]  [font=\footnotesize]  {$I_{r}^{1}$};
			\draw (379.11,151.51) node [anchor=north west][inner sep=0.75pt]  [font=\footnotesize]  {$I_{\ell }^{2}$};
			\draw (469.11,152.84) node [anchor=north west][inner sep=0.75pt]  [font=\footnotesize]  {$I_{r}^{2}$};

		\end{tikzpicture}
		
	\end{center}
	The above boundary decompositions of $X_1$ and $X_2$ define a boundary decomposition of $X$. Write $\Lambda$ for $I^1_\ell\sqcup I^2_r$, i.e. $\Lambda$ is the disjoint union of the outer (blue) intervals above. Write $\Sigma$ for the interval $I^1_r$ and identify $I^2_\ell$ with $\overline{\Sigma}$. Doing this we are in a situation in which we can apply Axiom \textbf{(T5b)} to obtain a formula involving the amplitude maps of the standard disk. Let $X'$ be the region obtained by gluing $X$ along $\Sigma$ and $\overline{\Sigma}$. Thus defined $X$ is a new copy $X'$ of the standard disk $\mathbb{D}^2$. We represent $X'$ pictorially as:
	\begin{center}

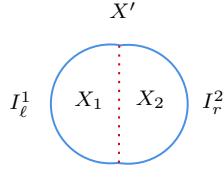
		\tikzset{every picture/.style={line width=0.75pt}} 
		
		\begin{tikzpicture}[x=0.75pt,y=0.75pt,yscale=-1,xscale=1]
			
			\draw  [draw opacity=0] (334.72,208.89) .. controls (333.31,209.09) and (331.87,209.2) .. (330.4,209.2) .. controls (313.83,209.2) and (300.4,195.77) .. (300.4,179.2) .. controls (300.4,162.63) and (313.83,149.2) .. (330.4,149.2) .. controls (331.96,149.2) and (333.5,149.32) .. (335,149.55) -- (330.4,179.2) -- cycle ; \draw  [color={rgb, 255:red, 74; green, 144; blue, 226 }  ,draw opacity=1 ] (334.72,208.89) .. controls (333.31,209.09) and (331.87,209.2) .. (330.4,209.2) .. controls (313.83,209.2) and (300.4,195.77) .. (300.4,179.2) .. controls (300.4,162.63) and (313.83,149.2) .. (330.4,149.2) .. controls (331.96,149.2) and (333.5,149.32) .. (335,149.55) ;  
			\draw  [draw opacity=0] (334.88,149.74) .. controls (336.17,149.54) and (337.49,149.43) .. (338.82,149.41) .. controls (355.39,149.09) and (369.08,162.26) .. (369.39,178.82) .. controls (369.71,195.39) and (356.54,209.08) .. (339.98,209.39) .. controls (337.98,209.43) and (336.02,209.27) .. (334.12,208.94) -- (339.4,179.4) -- cycle ; \draw  [color={rgb, 255:red, 74; green, 144; blue, 226 }  ,draw opacity=1 ] (334.88,149.74) .. controls (336.17,149.54) and (337.49,149.43) .. (338.82,149.41) .. controls (355.39,149.09) and (369.08,162.26) .. (369.39,178.82) .. controls (369.71,195.39) and (356.54,209.08) .. (339.98,209.39) .. controls (337.98,209.43) and (336.02,209.27) .. (334.12,208.94) ;  
			\draw [color={rgb, 255:red, 208; green, 2; blue, 27 }  ,draw opacity=1 ] [dash pattern={on 0.84pt off 2.51pt}]  (334.88,149.74) -- (334.72,208.89) ;
			
			\draw (328.5,126.09) node [anchor=north west][inner sep=0.75pt]  [font=\footnotesize]  {$X'$};
			\draw (310.75,171.4) node [anchor=north west][inner sep=0.75pt]  [font=\footnotesize]  {$X_{1}$};
			\draw (341.75,171.15) node [anchor=north west][inner sep=0.75pt]  [font=\footnotesize]  {$X_{2}$};
			\draw (278,171.4) node [anchor=north west][inner sep=0.75pt]  [font=\footnotesize]  {$I_{\ell }^{1}$};
			\draw (375.25,170.9) node [anchor=north west][inner sep=0.75pt]  [font=\footnotesize]  {$I_{r}^{2}$};

		\end{tikzpicture}
		
	\end{center}
	Every state $\psi\in\mathcal{H}_{\partial X'}$ is a linear combination of states of the form $\tau(\psi_1\otimes\psi_2)$, where $\psi_1\in\mathcal{H}_{I^1_\ell}$ and $\psi_2\in\mathcal{H}_{I^2_r}$, where $\tau$ is the isometric isomorphism in Axiom~\textbf{(T2)} associated to the boundary decomposition $\partial X'=I_\ell^1\cup I_r^2$ depicted in the picture above. Thus, in order to compute $\rho_{X'}$ on every state of $\mathcal{H}_{\partial X'}$ it is enough to compute $\rho_{X'}$ on states of the form $\tau(\psi_1\otimes\psi_2)$ as above. Let $\left\{\zeta_k:k\in I\right\}$ be an orthonormal basis for $\mathcal{H}_\Sigma$. Axiom~\textbf{(T5b)} yields the formula:
	\begin{equation}\label{GluingFormula2disks1}
		\rho_{X'}\tau(\psi_1\otimes\psi_2)=c\sum_{k\in I}(-1)^{[\zeta_k]}\rho_X\tau(\psi_1\otimes\psi_2\otimes\zeta_k\otimes\iota_\Sigma(\zeta_k))
	\end{equation}
	for a complex number $c$. Identifying $X', X_1$ and $X_2$ with $\mathbb{D}^2$, $\Sigma$ with $[0,1]^2$ and using Axiom~\textbf{(T5a)}, we obtain the formula:
	\begin{equation}\label{GluingFormula2disks2}
		\rho_{\mathbb{D}^2}\tau(\psi_1\otimes\psi_2)=c\sum_{k\in I}(-1)^{[\zeta_k]}\rho_{\mathbb{D}^2}(\psi_1\otimes\zeta_k)\rho_{\mathbb{D}^2}(\psi_2\otimes\iota_{[0,1]}(\zeta_k))
	\end{equation}
	which involves the amplitude of the disk $\mathbb{D}^2$ on the two sides of the equation. The evaluation of a GBQFT of general dimension $n$ on the gluing of two copies of the standard $n$-ball onto a third copy of the standard $n$-ball will yield a similar formula.
\end{ex}

The following example illustrates the behavior of self-gluing a disk into itself in the context of a GBQFT of dimension 2. We obtain another formula for the amplitude of the standard disk $\mathbb{D}^2$.

\begin{ex}\label{examplediskondiskclassic}
	As in Example~\ref{ExampleTwoDisks} assume we are given a GBQFT of dimension $2$. Let $X$ denote the standard disk $\mathbb{D}^2$. Consider the boundary decomposition $\partial X=\Lambda\sqcup\Sigma\sqcup\overline{\Sigma}$ where $\Sigma$ and $\Lambda$ are copies of the standard interval $[0,1]$ pictorially described as:
	\begin{center}

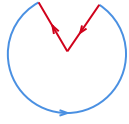
		\tikzset{every picture/.style={line width=0.75pt}} 
		
		\begin{tikzpicture}[x=0.75pt,y=0.75pt,yscale=-0.5,xscale=0.5]
			
			\draw  [draw opacity=0] (362.75,123.4) .. controls (378.77,133.88) and (389.45,151.86) .. (389.71,172.42) .. controls (390.13,205.26) and (363.79,232.23) .. (330.88,232.65) .. controls (297.96,233.07) and (270.94,206.78) .. (270.53,173.94) .. controls (270.23,151.15) and (282.83,131.19) .. (301.56,120.97) -- (330.12,173.18) -- cycle ; \draw  [color={rgb, 255:red, 74; green, 144; blue, 226 }  ,draw opacity=1 ] (362.75,123.4) .. controls (378.77,133.88) and (389.45,151.86) .. (389.71,172.42) .. controls (390.13,205.26) and (363.79,232.23) .. (330.88,232.65) .. controls (297.96,233.07) and (270.94,206.78) .. (270.53,173.94) .. controls (270.23,151.15) and (282.83,131.19) .. (301.56,120.97) ;
			\draw  [color={rgb, 255:red, 74; green, 144; blue, 226 }  ,draw opacity=1 ] (322.69,229.77) -- (330.69,232.64) -- (322.61,235.31) ;
			\draw [color={rgb, 255:red, 208; green, 2; blue, 27 }  ,draw opacity=1 ]   (301.56,120.97) -- (330.54,170.44) ;
			\draw [color={rgb, 255:red, 208; green, 2; blue, 27 }  ,draw opacity=1 ]   (362.75,123.4) -- (330.54,170.44) ;
			\draw  [color={rgb, 255:red, 208; green, 2; blue, 27 }  ,draw opacity=1 ] (317.01,152.85) -- (315.51,144.48) -- (321.85,150.14) ;
			\draw  [color={rgb, 255:red, 208; green, 2; blue, 27 }  ,draw opacity=1 ] (349.9,147.43) -- (343.62,151.58) -- (345.35,144.25) ;

		\end{tikzpicture}
		
	\end{center}
	where $\Sigma, \overline{\Sigma}$ are pictured as striaght lines (red) and $\Lambda$ is pictured as a circle segment (blue). The region $X'$ obtained from $X$ by gluing itself along the above decomposition is again the standard disk $\mathbb{D}^2$. We picture the gluing of $X$ into $X'$ as:
	\begin{center}
		\tikzset{every picture/.style={line width=0.75pt}} 
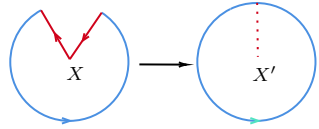
		\begin{tikzpicture}[x=0.75pt,y=0.75pt,yscale=-0.5,xscale=0.5]
			
			\draw  [draw opacity=0][fill={rgb, 255:red, 255; green, 255; blue, 255 }  ,fill opacity=1 ] (223.08,108.06) .. controls (239.1,118.54) and (249.79,136.53) .. (250.05,157.09) .. controls (250.47,189.93) and (224.12,216.89) .. (191.21,217.31) .. controls (158.3,217.73) and (131.28,191.45) .. (130.86,158.6) .. controls (130.57,135.82) and (143.16,115.86) .. (161.89,105.63) -- (190.45,157.84) -- cycle ; \draw  [color={rgb, 255:red, 74; green, 144; blue, 226 }  ,draw opacity=1 ] (223.08,108.06) .. controls (239.1,118.54) and (249.79,136.53) .. (250.05,157.09) .. controls (250.47,189.93) and (224.12,216.89) .. (191.21,217.31) .. controls (158.3,217.73) and (131.28,191.45) .. (130.86,158.6) .. controls (130.57,135.82) and (143.16,115.86) .. (161.89,105.63) ;
			\draw  [color={rgb, 255:red, 74; green, 144; blue, 226 }  ,draw opacity=1 ] (183.02,214.43) -- (191.02,217.31) -- (182.95,219.98) ;
			\draw [color={rgb, 255:red, 208; green, 2; blue, 27 }  ,draw opacity=1 ]   (161.89,105.63) -- (190.88,155.1) ;
			\draw [color={rgb, 255:red, 208; green, 2; blue, 27 }  ,draw opacity=1 ]   (223.08,108.06) -- (190.88,155.1) ;
			\draw  [color={rgb, 255:red, 208; green, 2; blue, 27 }  ,draw opacity=1 ] (177.35,137.52) -- (175.84,129.15) -- (182.19,134.81) ;
			\draw  [color={rgb, 255:red, 208; green, 2; blue, 27 }  ,draw opacity=1 ] (210.23,132.1) -- (203.95,136.25) -- (205.68,128.92) ;
			\draw  [draw opacity=0] (353.02,104.36) .. controls (360.68,100.63) and (369.27,98.49) .. (378.36,98.37) .. controls (411.28,97.96) and (438.3,124.24) .. (438.71,157.09) .. controls (439.13,189.93) and (412.79,216.89) .. (379.88,217.31) .. controls (346.96,217.73) and (319.94,191.45) .. (319.53,158.6) .. controls (319.22,134.94) and (332.82,114.32) .. (352.75,104.49) -- (379.12,157.84) -- cycle ; \draw  [color={rgb, 255:red, 74; green, 144; blue, 226 }  ,draw opacity=1 ] (353.02,104.36) .. controls (360.68,100.63) and (369.27,98.49) .. (378.36,98.37) .. controls (411.28,97.96) and (438.3,124.24) .. (438.71,157.09) .. controls (439.13,189.93) and (412.79,216.89) .. (379.88,217.31) .. controls (346.96,217.73) and (319.94,191.45) .. (319.53,158.6) .. controls (319.22,134.94) and (332.82,114.32) .. (352.75,104.49) ;
			\draw  [color={rgb, 255:red, 80; green, 227; blue, 194 }  ,draw opacity=1 ] (373.02,214.43) -- (381.02,217.31) -- (372.95,219.98) ;
			\draw    (261.08,160.33) -- (306.58,160.53) ;
			\draw [shift={(308.58,160.54)}, rotate = 180.25] [color={rgb, 255:red, 0; green, 0; blue, 0 }  ][line width=0.75]    (7.65,-2.3) .. controls (4.86,-0.97) and (2.31,-0.21) .. (0,0) .. controls (2.31,0.21) and (4.86,0.98) .. (7.65,2.3)   ;
			\draw [color={rgb, 255:red, 208; green, 2; blue, 27 }  ,draw opacity=1 ] [dash pattern={on 0.84pt off 2.51pt}]  (379.89,98.13) -- (380.58,152.67) ;
			
			\draw (184.88,159.82) node [anchor=north west][inner sep=0.75pt]  [font=\footnotesize,xscale=0.9,yscale=0.9]  {$X$};
			\draw (372.38,159.32) node [anchor=north west][inner sep=0.75pt]  [font=\footnotesize,xscale=0.9,yscale=0.9]  {$X'$};
		\end{tikzpicture}
	\end{center}
	Let $\{\zeta_k\}_{k\in I}$ be an orthonormal basis for $\cH_\Sigma$. Axiom \textbf{(T5b)} of Definition \ref{AFdefinitionClassic} yields the following equation:
	\begin{equation}\label{GluingFormula1disk1}
		\rho_{X'}(\psi)=c\sum_{k\in I}(-1)^{[\zeta_k]}\rho_X\tau(\psi\otimes\zeta_k\otimes\iota_\Sigma(\zeta_k))
	\end{equation}
	for every $\psi\in\mathcal{H}_{\partial X'}$. Identifying $X$ and $X'$ with the standard disk $\mathbb{D}^2$, $\partial X'$ with the standard circle $S^1$ and $\Sigma$ with the standard interval $[0,1]$ we obtain the formula:
	\begin{equation}\label{GluingFormula1disk2}
		\rho_{\mathbb{D}^2}(\psi)=c\sum_{k\in I}(-1)^{[\zeta_k]}\rho_{\mathbb{D}^2}\tau(\psi\otimes\zeta_k\otimes \iota_\Sigma(\zeta_k))
	\end{equation}
	for every $\psi\in\cH_{S^1}$.
\end{ex}

Finally, the following example illustrates the evaluation of the usual self-gluing of the standard cylinder $S^1\times [0,1]$ onto the standard torus $\mathbb{T}^2$ that occurs in vanilla TQFT, under a GBQFT of dimension 2. This will provide an explicit formula for gluing anomalies.

\begin{ex}\label{examplefinitedimensionality}
As in Examples~\ref{ExampleTwoDisks} and \ref{examplediskondiskclassic}, assume we are given a GBQFT of dimension 2. Denote the standard circle $S^1$ by $\Sigma$, and denote the standard cylinder $S^1\times [0,1]$ by $X$. In the notation of Axiom~\textbf{(T5b)} let $\Lambda$ be the empty hypersurface $\emptyset$, and let $\partial(S^1\times [0,1])=S^1\sqcup \overline{S^1}$ be the boundary decomposition $\partial X=\Sigma\sqcup\overline{\Sigma}\sqcup\Lambda$ of $X$ along which we wish to glue $X$. Pictorially: 
	\begin{center}
			
		\tikzset{every picture/.style={line width=0.75pt}} 
		
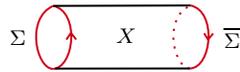
		\begin{tikzpicture}[x=0.75pt,y=0.75pt,yscale=-0.8,xscale=0.8]
			
			\draw [color={rgb, 255:red, 208; green, 2; blue, 27 }  ,draw opacity=1 ]   (288.87,162.75) .. controls (278.46,162.18) and (271.7,130.72) .. (289.26,123.01) ;
			\draw [color={rgb, 255:red, 208; green, 2; blue, 27 }  ,draw opacity=1 ]   (289.26,123.01) .. controls (306.12,126.63) and (302.68,162.39) .. (288.87,162.75) ;
			\draw [color={rgb, 255:red, 208; green, 2; blue, 27 }  ,draw opacity=1 ]   (375.79,122.75) .. controls (392.65,126.36) and (389.21,162.13) .. (375.4,162.49) ;
			\draw [color={rgb, 255:red, 208; green, 2; blue, 27 }  ,draw opacity=1 ] [dash pattern={on 0.84pt off 2.51pt}]  (375.4,162.49) .. controls (364.99,161.91) and (358.23,130.46) .. (375.79,122.75) ;
			\draw    (289.26,123.01) -- (375.79,122.75) ;
			\draw    (290.42,162.57) -- (375.4,162.49) ;
			\draw  [color={rgb, 255:red, 208; green, 2; blue, 27 }  ,draw opacity=1 ] (297.49,144.29) -- (300.44,140.07) -- (303.68,144.23) ;
			\draw  [color={rgb, 255:red, 208; green, 2; blue, 27 }  ,draw opacity=1 ] (390.05,140.48) -- (386.82,144.21) -- (384.35,140.32) ;
			
			\draw (260,136.73) node [anchor=north west][inner sep=0.75pt]  [font=\footnotesize]  {$\Sigma $};
			\draw (395,135.07) node [anchor=north west][inner sep=0.75pt]  [font=\footnotesize]  {$\overline{\Sigma }$};
			\draw (327,135.4) node [anchor=north west][inner sep=0.75pt]  [font=\footnotesize]  {$X$};

		\end{tikzpicture}
		
	\end{center}
Here, in the pictorial notation of Examples~\ref{ExampleTwoDisks} and \ref{examplediskondiskclassic} the circle segment (blue) boundary component, i.e.\ $\Lambda$ is empty. Let $X'$ be the region obtained by gluing $X$ to itself along the above boundary decomposition. We represent $X'$ pictorially as:
	\begin{center}

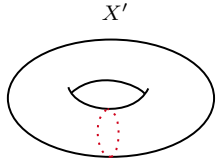
		\tikzset{every picture/.style={line width=0.75pt}} 
		
		\begin{tikzpicture}[x=0.75pt,y=0.75pt,yscale=-0.8,xscale=0.8]
			
			\draw   (262.8,149.17) .. controls (262.8,128.73) and (291.9,112.17) .. (327.8,112.17) .. controls (363.7,112.17) and (392.8,128.73) .. (392.8,149.17) .. controls (392.8,169.6) and (363.7,186.17) .. (327.8,186.17) .. controls (291.9,186.17) and (262.8,169.6) .. (262.8,149.17) -- cycle ;
			\draw    (300.95,142.22) .. controls (305.69,157.84) and (341.75,162.62) .. (351.35,143.02) ;
			\draw    (302.97,145.67) .. controls (320.15,130.62) and (344.95,141.02) .. (348.95,146.22) ;
			\draw [color={rgb, 255:red, 208; green, 2; blue, 27 }  ,draw opacity=1 ] [dash pattern={on 0.84pt off 2.51pt}]  (327.8,186.17) .. controls (318.34,185.87) and (315.77,161.3) .. (325.62,156.09) ;
			\draw [color={rgb, 255:red, 208; green, 2; blue, 27 }  ,draw opacity=1 ] [dash pattern={on 0.84pt off 2.51pt}]  (327.8,186.17) .. controls (334.06,184.15) and (335.2,157.01) .. (325.62,156.09) ;
			
			\draw (320.14,87.54) node [anchor=north west][inner sep=0.75pt]  [font=\footnotesize]  {$X'$};

		\end{tikzpicture}
		
	\end{center}
	We obtain a formula relating the dimension of the state space $\mathcal{H}_{S^1}$ of $S^1$ and the corresponding gluing anomaly $c$. Let $\{\zeta_k\}_{k\in I}$ be an orthonormal basis for $\cH_{\Sigma}$. A combination of Axioms~\textbf{(T4)} and \textbf{(T5b)} yields the equation:
	\begin{equation}
		1=c\sum_{k\in I}(-1)^{[\zeta_k]}\rho_{X'}\tau(\zeta_k\otimes\iota_\Sigma(\zeta_k))
	\end{equation}
	Identifying $\Sigma$ with $S^1$ and $X$ with $S^1\times [0,1]$ we obtain the formula:
	\begin{equation}
		c^{-1}=\sum_{k\in I}(-1)^{[\zeta_k]}\rho_{X}\tau(\zeta_k\otimes \iota_{S^1}(\zeta_k))
	\end{equation}
	Following the comments after Definition \ref{QFTAFfin}, the standard cylinder $S^1\times [0,1]$ is a slice region $\hat{S^1}$ for $S^1$. Therefore, using Axioms~\textbf{(T2)} and \textbf{(T3x)} we can rewrite this as,
	\begin{equation}\label{finitedimensionalS1equation}
		c^{-1}
		=\sum_{k\in I}(-1)^{[\zeta_k]+|\zeta_k|}\langle \zeta_k, \zeta_k\rangle_{S^1}
		=\sum_{k\in I}(-1)^{|\zeta_k|}
		=\dim\cH_{S^1}^0 - \dim\cH_{S^1}^1 .
	\end{equation}
\end{ex}

\begin{obs}
\label{obs:findim}
Equation~(\ref{finitedimensionalS1equation}) implies that in the situation of Example~\ref{examplefinitedimensionality}, the state space $\cH_{S^1}$ of $S^1$ is finite-dimensional. Since the cylinder over any closed hypersurface $\Sigma$ can be glued to itself in the same manner, its corresponding state space $\cH_{\Sigma}$ must be finite-dimensional. The above is a modification of the standard argument for the finite-dimensionality of state spaces in TQFT.
\end{obs}

Observation~\ref{obs:findim} appears to indicate a serious problem. We cannot have infinite-dimensional state spaces, and so the theory must apparently be trivial. To resolve this problem, GBQFT introduces another ingredient that we have not mentioned so far. State spaces are generically infinite-dimensional in QFT. So, we need to deal with this situation, even apart from Observation~\ref{obs:findim}. To address this, GBQFT introduces a notion of \emph{admissibility} for manifolds and for gluings. Simply put, the regions and hypersurfaces mentioned in the axioms are not to be taken to be all topological manifolds of dimension $n$ and $n-1$, but only certain classes of them. What is more, not all gluings are admissible. In particular, in Example~\ref{examplefinitedimensionality} the hypersurface $S^1$ would be admissible as well as the cylinder $S^1\times [0,1]$, but not the self-gluing that yields the torus. Of course this notion of admissibility must be coherent in the sense that boundaries of admissible regions are admissible hypersurfaces, compositions of admissible gluings are admissible gluings etc. The price to pay is the diminished usefulness of GBQFT to deal with a wide variety of topologies and underlying gluings. However, as stated previously, such topologies do not actually appear in the use case of primary interest here, QFT.
It tuns out that in order to make the GBQFT axioms work successfully with infinite-dimensional state spaces in realistic QFT, further adjustments have to be made. In particular, the amplitude map $\rho_X$ of Axiom~\textbf{(T4)} will not in general be defined on the whole Krein superspace $\cH_{\partial X}$, but only on a dense subspace \cite{Oe:holomorphic}. These details will be taken into account in Section~\ref{sec:axiomatics}.

\begin{rem}
	Formulas~(\ref{GluingFormula2disks2}) and (\ref{GluingFormula1disk2}) are obtained from formulas~(\ref{GluingFormula2disks1}) and (\ref{GluingFormula1disk1}) by identifying certain homeomorphic hypersurfaces and regions. This suggests that the items in Definition~\ref{QFTAFfin} are invariant under orientation preserving homeomorphisms, and thus formulas obtained from the GBQFT axioms should be defined "up to isomorphisms". This is intentionally not explicitly stated in Definition~\ref{QFTAFfin}. The gluing formalism developed in Section 3 will allow for a finer identification of formulas "up to isomorphisms" while not strictly identifying homeomorphic hypersurfaces and regions. This is done by making spacetime symmetries native to our formalism.
\end{rem}

\subsection{Plan for the paper}

In Section~\ref{usualgluings} we present the definitions of gluing function, relative gluing diagram, and slice region. We present examples and study their formal properties. Gluing functions, relative gluing diagrams, and slice regions form the gluing formalism we use to formalize the notion of composite spacetime system in CQFT. In Section~\ref{sec:axiomatics} we present the definition of CQFT. We provide examples of computations and compare them to corresponding computations in GBQFT. In Section~\ref{sec:2dim} we specialize the definition of CQFT to dimension 2. We study relevant structures defined by 2-dimensional CQFTs, and after adapting our axioms to regions with area, we describe 2-dimensional quantum Yang-Mills theory with corners as a special case of a 2-dimensional CQFT in Section~\ref{sec:2dpqym}. We end the paper with an Outlook.


\section{Gluings: A new perspective}
\label{usualgluings}

In Atiyah's original TQFT axioms the notion of gluing of manifolds is implicit in the choice of a suitable cobordism category. The hypersurfaces are the objects and the regions, having the structure of cobordisms, correspond to morphisms. Gluing of regions, performed with the aid of a choice of collar structure, is defined "up to isomorphisms", and encodes the composition of morphisms. The TQFT set-up no longer makes sense once we forget the in-out distinction of cobordisms and allow hypersurfaces to have boundaries. A choice of collar structure can be a subtle problem in more involved cases, e.g. on Riemannian or conformal cobordisms, see \cite{StolzTeichnerSusy}, and considering regions "up to isomorphisms" suppresses natural spacetime symmetry axioms. In order to promote the categorification of notions of manifolds and their gluings in the context of the GBF, and to overcome some of the issues described above, we propose a new formalism centered around a notion of \emph{gluing functions}. This formalism is based on four principles:
\begin{enumerate}
	\item The way one glues a manifold into another manifold contains topological data we wish to record. We thus consider the identification performing a gluing as our fundamental operation, rather than the gluing itself.
	\item The boundary components along which one glues a region contains vital information in a QFT and thus our gluing formalism should allow us to naturally keep track of "boundary components along which we glue a region".
	\item Gluing of regions should be unital, i.e., our formalism should accommodate the existence of regions acting trivially by gluing. These are the \emph{slice regions} appearing in the GBQFT axioms.
	\item Symmetries of manifolds, i.e., homeomorphisms, should be seen as integral parts of the description of the gluing process. In particular, pure symmetries should appear as "trivial" gluings. This should lead to an equivariant formalism.
\end{enumerate}

The proposed formalism for treating gluings of hypersurfaces and regions in the context of the axioms of GBQFT will provide a roadmap on how to systematically treat gluings in more involved and more general settings.

\subsection{Gluing functions}
\label{gluingfunctionssubsection}

We formally introduce our notion of gluing function. Gluing functions as presented in this section describe the most basic notion of gluing we consider. We provide examples and computations and set notational conventions for the rest of the paper.

Fix a positive integer $n$. If not indicated otherwise, in the following a \emph{hypersurface} is any oriented compact $(n-1)$-dimensional topological manifold, possibly with boundary. Similarly, a \emph{region} is any oriented compact $n$-dimensional topological manifold, possibly with boundary.
We consider identifications between manifolds in which only points in the boundary are non-trivially identified. The fact that we only consider compact manifolds allows us to characterize identifications as closed, epic, continuous functions. The following definition captures our notion of gluing function.

\begin{definition}\label{gluingfunctiondefinition}
Let $\Sigma,\Sigma'$ be hypersurfaces or regions. Let $f:\Sigma\to\Sigma'$ be a continuous function. We will say that $f$ is a gluing function if $f$ is an identification, i.e., $f$ is epic and closed; if the restriction $f|_{\Sigma^\circ}$ of $f$ to the interior $\Sigma^\circ$ of $\Sigma$ is an embedding; and if $f$ preserves orientation.
\end{definition}

In the following, we sometimes omit to explicitly mention the orientation of manifolds. In these cases it is understood that a coherent choice of orientations is made.

We present examples of gluing functions. In what follows we consider first manifolds of dimension 1. Compact, oriented, connected topological manifolds of dimension 1 with boundary are usually referred to as abstract intervals, see \cite{Bartels1}, or open strings, see \cite{Oe:2dqym}. Every open string is homeomorphic to the standard unit interval $[0,1]$. Compact, oriented, connected, closed topological manifolds of dimension 1 are referred to as abstract circles, or closed strings, and every closed string is homeomorphic to the standard circle $S^1$. We will use the words interval and open string interchangeably, and we will use the words circle and closed string interchangeably. 

\begin{ex}\label{exampletwointervals}
Definition~\ref{gluingfunctiondefinition} supports the operation of composing open strings into composite open strings. Given two open strings $O_1$ and $O_2$ we can form the glued open string $O_1\vee O_2$ by considering the quotient of $O_1\sqcup O_2$ by identifying boundary components of $O_1$ and $O_2$ in such a way that $O_1\vee O_2$ inherits an orientation from $O_1$ and $O_2$, see \cite{Bartels1,Oe:2dqym}. Let $\alpha:O_1\sqcup O_2\to O_1\vee O_2$ be the identification map of $O_1\sqcup O_2$ onto $O_1\vee O_2$. Thus defined $\alpha$ is a gluing function. Pictorially, the operation of composing two open strings $O_1$ and $O_2$ into a composite open string $O_1\vee O_2$ is described as:
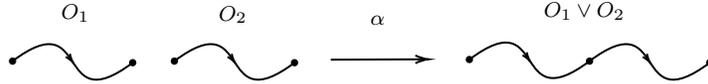
\begin{center}
		\tikzset{every picture/.style={line width=0.75pt}} 
		
		\begin{tikzpicture}[x=0.75pt,y=0.75pt,yscale=-1,xscale=1]
			
			\draw    (299.69,249.47) -- (346.92,249.29) ;
			\draw [shift={(348.92,249.28)}, rotate = 179.78] [color={rgb, 255:red, 0; green, 0; blue, 0 }  ][line width=0.75]    (6.56,-1.97) .. controls (4.17,-0.84) and (1.99,-0.18) .. (0,0) .. controls (1.99,0.18) and (4.17,0.84) .. (6.56,1.97)   ;
			\draw    (139.2,250.2) .. controls (179.2,220.2) and (159.8,281) .. (199.8,251) ;
			\draw [shift={(199.8,251)}, rotate = 323.13] [color={rgb, 255:red, 0; green, 0; blue, 0 }  ][fill={rgb, 255:red, 0; green, 0; blue, 0 }  ][line width=0.75]      (0, 0) circle [x radius= 1.34, y radius= 1.34]   ;
			\draw [shift={(169.65,250.83)}, rotate = 234.02] [color={rgb, 255:red, 0; green, 0; blue, 0 }  ][line width=0.75]    (4.37,-1.32) .. controls (2.78,-0.56) and (1.32,-0.12) .. (0,0) .. controls (1.32,0.12) and (2.78,0.56) .. (4.37,1.32)   ;
			\draw [shift={(139.2,250.2)}, rotate = 323.13] [color={rgb, 255:red, 0; green, 0; blue, 0 }  ][fill={rgb, 255:red, 0; green, 0; blue, 0 }  ][line width=0.75]      (0, 0) circle [x radius= 1.34, y radius= 1.34]   ;
			\draw    (220.8,250.2) .. controls (260.8,220.2) and (241.4,281) .. (281.4,251) ;
			\draw [shift={(281.4,251)}, rotate = 323.13] [color={rgb, 255:red, 0; green, 0; blue, 0 }  ][fill={rgb, 255:red, 0; green, 0; blue, 0 }  ][line width=0.75]      (0, 0) circle [x radius= 1.34, y radius= 1.34]   ;
			\draw [shift={(251.25,250.83)}, rotate = 234.02] [color={rgb, 255:red, 0; green, 0; blue, 0 }  ][line width=0.75]    (4.37,-1.32) .. controls (2.78,-0.56) and (1.32,-0.12) .. (0,0) .. controls (1.32,0.12) and (2.78,0.56) .. (4.37,1.32)   ;
			\draw [shift={(220.8,250.2)}, rotate = 323.13] [color={rgb, 255:red, 0; green, 0; blue, 0 }  ][fill={rgb, 255:red, 0; green, 0; blue, 0 }  ][line width=0.75]      (0, 0) circle [x radius= 1.34, y radius= 1.34]   ;
			\draw    (369.6,249.4) .. controls (409.6,219.4) and (390.2,280.2) .. (430.2,250.2) ;
			\draw [shift={(430.2,250.2)}, rotate = 323.13] [color={rgb, 255:red, 0; green, 0; blue, 0 }  ][fill={rgb, 255:red, 0; green, 0; blue, 0 }  ][line width=0.75]      (0, 0) circle [x radius= 1.34, y radius= 1.34]   ;
			\draw [shift={(400.05,250.03)}, rotate = 234.02] [color={rgb, 255:red, 0; green, 0; blue, 0 }  ][line width=0.75]    (4.37,-1.32) .. controls (2.78,-0.56) and (1.32,-0.12) .. (0,0) .. controls (1.32,0.12) and (2.78,0.56) .. (4.37,1.32)   ;
			\draw [shift={(369.6,249.4)}, rotate = 323.13] [color={rgb, 255:red, 0; green, 0; blue, 0 }  ][fill={rgb, 255:red, 0; green, 0; blue, 0 }  ][line width=0.75]      (0, 0) circle [x radius= 1.34, y radius= 1.34]   ;
			\draw    (430.2,250.2) .. controls (470.2,220.2) and (450.8,281) .. (490.8,251) ;
			\draw [shift={(490.8,251)}, rotate = 323.13] [color={rgb, 255:red, 0; green, 0; blue, 0 }  ][fill={rgb, 255:red, 0; green, 0; blue, 0 }  ][line width=0.75]      (0, 0) circle [x radius= 1.34, y radius= 1.34]   ;
			\draw [shift={(460.65,250.83)}, rotate = 234.02] [color={rgb, 255:red, 0; green, 0; blue, 0 }  ][line width=0.75]    (4.37,-1.32) .. controls (2.78,-0.56) and (1.32,-0.12) .. (0,0) .. controls (1.32,0.12) and (2.78,0.56) .. (4.37,1.32)   ;
			
			\draw (318.21,225.55) node [anchor=north west][inner sep=0.75pt]  [font=\footnotesize]  {$\alpha$};
			\draw (162.29,219.97) node [anchor=north west][inner sep=0.75pt]  [font=\footnotesize]  {$O_{1}$};
			\draw (241.71,220.83) node [anchor=north west][inner sep=0.75pt]  [font=\footnotesize]  {$O_{2}$};
			\draw (405.71,219.11) node [anchor=north west][inner sep=0.75pt]  [font=\footnotesize]  {$O_{1} \lor O_{2}$};

		\end{tikzpicture}
\end{center}
Observe that the gluing function $\alpha:O_1\sqcup O_2\to O_1\vee O_2$, is, "up to isomorphism", the identification of the disjoint union $[0,1/2]\sqcup [1/2,1]$ of the two halves of the unit interval, onto $[0,1]$, by identifying the two boundary components $\left\{1/2\right\}$ of $[0,1/2]$ and $[1/2,1]$, or equivalently is, "up to isomorphism", the identification $[0,1]\vee [0,1]$ of two copies of the standard interval $[0,1]$ along $0$ and $1$.
\end{ex}

\begin{ex}\label{ExampleClosedStrings}
	Definition \ref{gluingfunctiondefinition} also supports the operation of closing an open string: Given an open string $O$, with boundary $\partial O=\left\{p,q\right\}$ we can form a closed string $C$ by identifying $p$ and $q$ in $O$, making $C$ inherit the orientation of $O$. Let $\alpha$ now be the identification map $\alpha:O\to C$. Thus defined $\alpha$ is a gluing function. Pictorially, the operation of closing an open string $O$ into a closed string $C$ is described as:
	\begin{center}

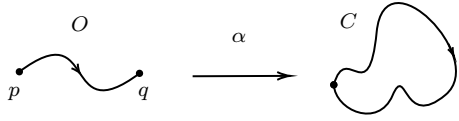
		\tikzset{every picture/.style={line width=0.75pt}} 
		
		\begin{tikzpicture}[x=0.75pt,y=0.75pt,yscale=-1,xscale=1]
			
			\draw    (297.79,140.47) -- (345.02,140.29) ;
			\draw [shift={(347.02,140.28)}, rotate = 179.78] [color={rgb, 255:red, 0; green, 0; blue, 0 }  ][line width=0.75]    (6.56,-1.97) .. controls (4.17,-0.84) and (1.99,-0.18) .. (0,0) .. controls (1.99,0.18) and (4.17,0.84) .. (6.56,1.97)   ;
			\draw    (426.78,124.11) -- (428.73,128.46) ;
			\draw [shift={(429.55,130.29)}, rotate = 245.82] [color={rgb, 255:red, 0; green, 0; blue, 0 }  ][line width=0.75]    (4.37,-1.32) .. controls (2.78,-0.56) and (1.32,-0.12) .. (0,0) .. controls (1.32,0.12) and (2.78,0.56) .. (4.37,1.32)   ;
			\draw    (210.4,138.2) .. controls (250.4,108.2) and (231,169) .. (271,139) ;
			\draw [shift={(271,139)}, rotate = 323.13] [color={rgb, 255:red, 0; green, 0; blue, 0 }  ][fill={rgb, 255:red, 0; green, 0; blue, 0 }  ][line width=0.75]      (0, 0) circle [x radius= 1.34, y radius= 1.34]   ;
			\draw [shift={(240.85,138.83)}, rotate = 234.02] [color={rgb, 255:red, 0; green, 0; blue, 0 }  ][line width=0.75]    (4.37,-1.96) .. controls (2.78,-0.92) and (1.32,-0.27) .. (0,0) .. controls (1.32,0.27) and (2.78,0.92) .. (4.37,1.96)   ;
			\draw [shift={(210.4,138.2)}, rotate = 323.13] [color={rgb, 255:red, 0; green, 0; blue, 0 }  ][fill={rgb, 255:red, 0; green, 0; blue, 0 }  ][line width=0.75]      (0, 0) circle [x radius= 1.34, y radius= 1.34]   ;
			\draw    (368.99,144.79) .. controls (375.4,122.2) and (386.6,158.6) .. (390.6,118.2) .. controls (394.6,77.8) and (447.01,131.59) .. (425.4,149.8) .. controls (403.79,168.01) and (405.8,133.4) .. (398.6,150.2) .. controls (391.4,167) and (369.8,157.8) .. (368.99,144.79) -- cycle ;
			\draw [shift={(368.99,144.79)}, rotate = 266.46] [color={rgb, 255:red, 0; green, 0; blue, 0 }  ][fill={rgb, 255:red, 0; green, 0; blue, 0 }  ][line width=0.75]      (0, 0) circle [x radius= 1.34, y radius= 1.34]   ;
			\draw [shift={(368.99,144.79)}, rotate = 285.83] [color={rgb, 255:red, 0; green, 0; blue, 0 }  ][fill={rgb, 255:red, 0; green, 0; blue, 0 }  ][line width=0.75]      (0, 0) circle [x radius= 1.34, y radius= 1.34]   ;
			
			\draw (316.05,117.02) node [anchor=north west][inner sep=0.75pt]  [font=\footnotesize]  {$\alpha$};
			\draw (235.09,108.77) node [anchor=north west][inner sep=0.75pt]  [font=\footnotesize]  {$O$};
			\draw (370.46,107.57) node [anchor=north west][inner sep=0.75pt]  [font=\footnotesize]  {$C$};
			\draw (203.37,144.2) node [anchor=north west][inner sep=0.75pt]  [font=\footnotesize]  {$p$};
			\draw (268.8,144.49) node [anchor=north west][inner sep=0.75pt]  [font=\footnotesize]  {$q$};

		\end{tikzpicture}
		
	\end{center}
	The closing string operation described above is, "up to isomorphisms" the operation of closing the standard interval $[0,1]$ to the standard circle $S^1$ through the gluing function $\alpha(t)=e^{2\pi it}$.
\end{ex}

The following trivial observation follows directly from Definition~\ref{gluingfunctiondefinition}.

\begin{obs}\label{observationgluingfunctionsclosed}
	The collection of gluing functions is closed under disjoint unions, composition, and orientation change, and these operations satisfy the obvious strict associativity conditions, i.e. given gluing functions $\alpha_1,\alpha_2$ and $\alpha_3$, such that the compositions $(\alpha_1\sqcup \alpha_2)\sqcup \alpha_3$ and $\alpha_1\sqcup (\alpha_2\sqcup \alpha_3)$ are defined, the equation $(\alpha_1\sqcup \alpha_2)\sqcup \alpha_3=\alpha_1\sqcup (\alpha_2\sqcup \alpha_3)$ holds, and if the compositions $\gamma(\beta\alpha)$ and $(\gamma\beta)\alpha$ are defined, then the equation $\gamma(\beta\alpha)=(\gamma\beta)\alpha$ holds. The operation of changing orientation is involutive in the sense that the equation $\overline{\overline{\alpha}}=\alpha$ holds for every gluing function $\alpha$. The operations of disjoint union, composition, and orientation change of gluing functions distribute over each other, i.e. given gluing functions $\alpha,\beta,\gamma,\delta$, the equation $(\alpha\sqcup \beta)(\gamma\sqcup \delta)=\alpha\gamma\sqcup \beta\delta$ holds, whenever the three compositions are defined, given two gluing functions $\alpha,\beta$, the equation $\overline{\alpha\sqcup \beta}=\overline{\alpha}\sqcup\overline{\beta}$ always holds, and for any pair of gluing functions $\alpha,\beta$ such that the composition $\beta\alpha$ is defined the equation $\overline{\beta\alpha}=\overline{\beta}\overline{\alpha}$ always holds.
\end{obs}

Observation~\ref{observationgluingfunctionsclosed} says that we can form gluing functions from arrangements of simpler gluing functions through taking disjoint unions, compositions and by inverting orientations; and that any expression formed by these operations is independent on the order in which the operations are performed, as long as they are well-defined. The following observation also trivially follows from Definition~\ref{gluingfunctiondefinition}.

\begin{obs}\label{ObservationHomeomorphisms}
	Homeomorphisms satisfy the conditions of Definition~\ref{gluingfunctiondefinition} and thus are examples of gluing functions. Moreover, homeomorphisms are precisely the invertible gluing functions. 
\end{obs}

Homeomorphisms encode spacetime symmetries. Homeomorphisms are precisely the gluing functions, where the "gluing" performed is trivial. Considering symmetries as gluing functions allows us to formally consider formulas such as (\ref{GluingFormula2disks2}) and (\ref{GluingFormula1disk2}), and allows us to provide formal meaning to statements such as the "up to isomorphisms" statements made at the end of Examples~\ref{exampletwointervals} and \ref{ExampleClosedStrings}. The following example explains how to do this.

\begin{ex}\label{Exhomeomorphisms}
	Let $\alpha:\Sigma\to \Sigma'$ be a gluing function, and let $\beta:\Lambda\to \Sigma$ and $\gamma:\Sigma'\to \Lambda'$ be homeomorphisms. The composition $\gamma\alpha\beta$ is a gluing function from $\Lambda$ to $\Lambda'$. Thus, if we know how to glue a manifold $\Sigma$ to itself in order to produce a manifold $\Sigma'$, then we unambiguously know how to glue any manifold $\Lambda$ homeomorphic to $\Sigma$ to itself in order to produce any manifold $\Lambda'$ homeomorphic to $\Sigma'$. As a particular case of this consider the following arrangement $\Sigma$ of open strings:
	\begin{center}

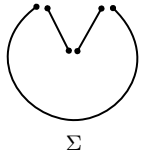
		\tikzset{every picture/.style={line width=0.75pt}} 
		
		\begin{tikzpicture}[x=0.75pt,y=0.75pt,yscale=-0.8,xscale=0.8]
			
			\draw    (218.6,122.16) .. controls (183.43,148.14) and (204.86,188.43) .. (237.14,193.29) .. controls (269.43,198.14) and (302.38,157) .. (266.86,123) ;
			\draw [shift={(266.86,123)}, rotate = 223.75] [color={rgb, 255:red, 0; green, 0; blue, 0 }  ][fill={rgb, 255:red, 0; green, 0; blue, 0 }  ][line width=0.75]      (0, 0) circle [x radius= 1.34, y radius= 1.34]   ;
			\draw [shift={(218.6,122.16)}, rotate = 143.55] [color={rgb, 255:red, 0; green, 0; blue, 0 }  ][fill={rgb, 255:red, 0; green, 0; blue, 0 }  ][line width=0.75]      (0, 0) circle [x radius= 1.34, y radius= 1.34]   ;
			\draw    (225.57,123.14) -- (239,149.86) ;
			\draw [shift={(239,149.86)}, rotate = 63.31] [color={rgb, 255:red, 0; green, 0; blue, 0 }  ][fill={rgb, 255:red, 0; green, 0; blue, 0 }  ][line width=0.75]      (0, 0) circle [x radius= 1.34, y radius= 1.34]   ;
			\draw [shift={(225.57,123.14)}, rotate = 63.31] [color={rgb, 255:red, 0; green, 0; blue, 0 }  ][fill={rgb, 255:red, 0; green, 0; blue, 0 }  ][line width=0.75]      (0, 0) circle [x radius= 1.34, y radius= 1.34]   ;
			\draw    (244.43,150.14) -- (259.57,123.29) ;
			\draw [shift={(259.57,123.29)}, rotate = 299.42] [color={rgb, 255:red, 0; green, 0; blue, 0 }  ][fill={rgb, 255:red, 0; green, 0; blue, 0 }  ][line width=0.75]      (0, 0) circle [x radius= 1.34, y radius= 1.34]   ;
			\draw [shift={(244.43,150.14)}, rotate = 299.42] [color={rgb, 255:red, 0; green, 0; blue, 0 }  ][fill={rgb, 255:red, 0; green, 0; blue, 0 }  ][line width=0.75]      (0, 0) circle [x radius= 1.34, y radius= 1.34]   ;
			
			\draw (235.2,201.4) node [anchor=north west][inner sep=0.75pt]  [font=\footnotesize]  {$\Sigma $};

		\end{tikzpicture}

	\end{center}
	and consider the obvious identification:
	\begin{center}

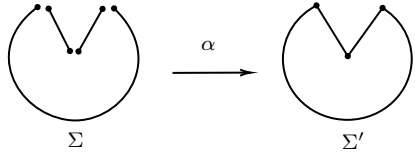
		\tikzset{every picture/.style={line width=0.75pt}} 
		
		\begin{tikzpicture}[x=0.75pt,y=0.75pt,yscale=-0.8,xscale=0.8]
			
			\draw  [draw opacity=0] (415.28,102.2) .. controls (426.66,108.88) and (434.3,120.76) .. (434.48,134.37) .. controls (434.77,155.59) and (416.82,173.01) .. (394.4,173.28) .. controls (371.98,173.55) and (353.57,156.57) .. (353.28,135.35) .. controls (353.08,120.33) and (362.01,107.21) .. (375.2,100.74) -- (393.88,134.86) -- cycle ; \draw  [color={rgb, 255:red, 0; green, 0; blue, 0 }  ,draw opacity=1 ] (415.28,102.2) .. controls (426.66,108.88) and (434.3,120.76) .. (434.48,134.37) .. controls (434.77,155.59) and (416.82,173.01) .. (394.4,173.28) .. controls (371.98,173.55) and (353.57,156.57) .. (353.28,135.35) .. controls (353.08,120.33) and (362.01,107.21) .. (375.2,100.74) ;
			\draw [color={rgb, 255:red, 0; green, 0; blue, 0 }  ,draw opacity=1 ]   (374.43,101.13) -- (394.17,133.09) ;
			\draw [shift={(394.17,133.09)}, rotate = 58.29] [color={rgb, 255:red, 0; green, 0; blue, 0 }  ,draw opacity=1 ][fill={rgb, 255:red, 0; green, 0; blue, 0 }  ,fill opacity=1 ][line width=0.75]      (0, 0) circle [x radius= 1.34, y radius= 1.34]   ;
			\draw [shift={(374.43,101.13)}, rotate = 58.29] [color={rgb, 255:red, 0; green, 0; blue, 0 }  ,draw opacity=1 ][fill={rgb, 255:red, 0; green, 0; blue, 0 }  ,fill opacity=1 ][line width=0.75]      (0, 0) circle [x radius= 1.34, y radius= 1.34]   ;
			\draw [color={rgb, 255:red, 0; green, 0; blue, 0 }  ,draw opacity=1 ]   (416.11,102.7) -- (394.17,133.09) ;
			\draw [shift={(416.11,102.7)}, rotate = 125.83] [color={rgb, 255:red, 0; green, 0; blue, 0 }  ,draw opacity=1 ][fill={rgb, 255:red, 0; green, 0; blue, 0 }  ,fill opacity=1 ][line width=0.75]      (0, 0) circle [x radius= 1.34, y radius= 1.34]   ;
			\draw    (283.69,143.67) -- (330.92,143.49) ;
			\draw [shift={(332.92,143.48)}, rotate = 179.78] [color={rgb, 255:red, 0; green, 0; blue, 0 }  ][line width=0.75]    (6.56,-1.97) .. controls (4.17,-0.84) and (1.99,-0.18) .. (0,0) .. controls (1.99,0.18) and (4.17,0.84) .. (6.56,1.97)   ;
			\draw    (198.6,102.16) .. controls (163.43,128.14) and (184.86,168.43) .. (217.14,173.29) .. controls (249.43,178.14) and (282.38,137) .. (246.86,103) ;
			\draw [shift={(246.86,103)}, rotate = 223.75] [color={rgb, 255:red, 0; green, 0; blue, 0 }  ][fill={rgb, 255:red, 0; green, 0; blue, 0 }  ][line width=0.75]      (0, 0) circle [x radius= 1.34, y radius= 1.34]   ;
			\draw [shift={(198.6,102.16)}, rotate = 143.55] [color={rgb, 255:red, 0; green, 0; blue, 0 }  ][fill={rgb, 255:red, 0; green, 0; blue, 0 }  ][line width=0.75]      (0, 0) circle [x radius= 1.34, y radius= 1.34]   ;
			\draw    (205.57,103.14) -- (219,129.86) ;
			\draw [shift={(219,129.86)}, rotate = 63.31] [color={rgb, 255:red, 0; green, 0; blue, 0 }  ][fill={rgb, 255:red, 0; green, 0; blue, 0 }  ][line width=0.75]      (0, 0) circle [x radius= 1.34, y radius= 1.34]   ;
			\draw [shift={(205.57,103.14)}, rotate = 63.31] [color={rgb, 255:red, 0; green, 0; blue, 0 }  ][fill={rgb, 255:red, 0; green, 0; blue, 0 }  ][line width=0.75]      (0, 0) circle [x radius= 1.34, y radius= 1.34]   ;
			\draw    (224.43,130.14) -- (239.57,103.29) ;
			\draw [shift={(239.57,103.29)}, rotate = 299.42] [color={rgb, 255:red, 0; green, 0; blue, 0 }  ][fill={rgb, 255:red, 0; green, 0; blue, 0 }  ][line width=0.75]      (0, 0) circle [x radius= 1.34, y radius= 1.34]   ;
			\draw [shift={(224.43,130.14)}, rotate = 299.42] [color={rgb, 255:red, 0; green, 0; blue, 0 }  ][fill={rgb, 255:red, 0; green, 0; blue, 0 }  ][line width=0.75]      (0, 0) circle [x radius= 1.34, y radius= 1.34]   ;
			
			\draw (299.87,121.47) node [anchor=north west][inner sep=0.75pt]  [font=\footnotesize]  {$\alpha$};
			\draw (215.5,179.4) node [anchor=north west][inner sep=0.75pt]  [font=\footnotesize]  {$\Sigma $};
			\draw (388.5,179.4) node [anchor=north west][inner sep=0.75pt]  [font=\footnotesize]  {$\Sigma '$};

		\end{tikzpicture}
		
	\end{center}
	That is, $\Sigma$ is closed into a closed string $\Sigma'$. We can unambiguously describe $\alpha$ as a gluing function from the disjoint union of the three copies $\Lambda=[0,1]\sqcup [0,1]\sqcup [0,1]$ to the standard circle $\Sigma'=S^1$. To do this, choose a homeomorphism $\beta:\Lambda\to\Sigma$ and choose a homeomorphism $\gamma$ from $\Sigma'$ to $\Lambda'$. In that case $\gamma\alpha\beta$ is a gluing function from $\Lambda=[0,1]\sqcup [0,1]\sqcup [0,1]$ to $\Lambda'=S^1$ obtained by deforming $\alpha$. We can also think of $\alpha$ as a gluing function from $\Lambda=[0,1]\sqcup [0,1]\sqcup [0,1]$ to $\Lambda'=S^1$ deformed by $\beta^{-1}$ and $\gamma^{-1}$ into a gluing function from $\Sigma$ to $\Sigma'$.  
\end{ex}

\subsection{Relative gluing diagrams}
\label{relativegluingsubsection}

We introduce relative gluing diagrams as a further ingredient in our gluing formalism, allowing us to phrase the GBQFT Axiom~\textbf{(T5b)} of Definition~\ref{AFdefinitionClassic} in terms of gluing functions. As presented in Section~\ref{gluingfunctionssubsection}, gluing functions are only used for hypersurface gluings. However, Axiom~\textbf{(T5b)} considers the gluing of a region to itself along two copies of a given hypersurface boundary component with opposite orientations, and is phrased in terms of this data explicitly. Relative gluing diagrams allow us to encode such gluings, while keeping track of the boundary components that are involved in the gluing operation. Relative gluing diagrams are types of decorated spans \cite{Fong} defined as coequalizers of pairs of continuous functions satisfying certain diagrammatic equations.

\begin{definition}\label{relativegluingdefinition}
Let $X$ and $X'$ be regions. Let $\Sigma$ and $\Lambda$ be hypersurfaces. Let $\alpha:\Sigma\sqcup\overline{\Sigma}\sqcup \Lambda\to \partial X$ be a gluing function. A gluing diagram on $X$, relative to $\alpha$, is a pair $(f,\beta)$, where $f:X\to X'$ and $\beta:\Lambda\to \partial X'$ are gluing functions, $f$ is a coequalizer of the diagram
\[\Sigma\rightrightarrows \Sigma\sqcup\overline{\Sigma}\sqcup \Lambda\xrightarrow[]{\alpha} \partial X\hookrightarrow X\]
and $\beta$ makes the following diagram commute:
\begin{center}

\begin{tikzpicture}
  \matrix (m) [matrix of math nodes,row sep=3em,column sep=3em,minimum width=2em]
  {\Lambda&&\partial X' \\
   \partial X&X&X'\\};
  \path[-stealth]
    (m-1-1) edge node [above] {$\beta$} (m-1-3)
            edge node [left]  {$\alpha|_{\Lambda}$} (m-2-1)
            
    (m-2-1) edge node [below] {} (m-2-2)
    (m-2-2) edge node [below]  {$f$} (m-2-3)
    (m-1-3) edge node [right] {} (m-2-3)
    
   ;
\end{tikzpicture}
\end{center}
\end{definition}

\begin{rem}
	\label{remarkrelativegluing}
	The condition of $f:X\to X'$ being a coequalizer in Definition~\ref{relativegluingdefinition} says that the function $f$ glues $X$ to $X'$ along the two copies of $\Sigma$ with opposite orientation appearing in $\alpha$. The condition on $\beta$ making the second diagram in Definition~\ref{relativegluingdefinition} commute says that the restriction of $f$ to $\Lambda$ provides a gluing decomposition of $\partial X'$. We summarize the data of a relative gluing diagram $(f,\beta)$ as in Definition~\ref{relativegluingdefinition} by a commutative diagram of the form:
	\begin{center}
		\begin{tikzpicture}
			\matrix (m) [matrix of math nodes,row sep=3em,column sep=3em,minimum width=2em]
			{\Lambda&&&\partial X' \\
				\Sigma\sqcup \overline{\Sigma}\sqcup\Lambda&\partial X&X&X'\\};
			\path[-stealth]
			(m-1-1) edge node [above] {$\beta$} (m-1-4)
			edge node [left]  {} (m-2-1)
			
			(m-2-1) edge node [below] {$\alpha$} (m-2-2)
			(m-2-2) edge node [below]  {} (m-2-3)
			(m-2-3) edge node [below] {$f$} (m-2-4)
			(m-1-4) edge node [right] {} (m-2-4)
			
			;
		\end{tikzpicture}
	\end{center}
\end{rem}

Observe that the main difference between Definition~\ref{gluingfunctiondefinition} and Definition~\ref{relativegluingdefinition} is the fact that relative gluing diagrams include relative gluing data as an explicit part of the defining structure. Relative gluing diagrams should be understood as the data of a self-gluing of a manifold along two copies of the same hypersurface with opposite orientation, as in Axiom~\textbf{(T5b)}, phrased within the gluing functions formalism developed in the previous subsection. We further clarify the diagram appearing in Remark~\ref{remarkrelativegluing} and the pieces of structure involved in Definition~\ref{relativegluingdefinition} with examples. The following example encodes the information necessary to make the computation presented in Example~\ref{ExampleTwoDisks} in terms of a relative gluing diagram.

\begin{ex}
	\label{ex:twodiskgluing}
	Let $X_1,X_2,X,X',\Sigma$, and $\Lambda$ be as in Example \ref{ExampleTwoDisks}. Let $f:X\to X'$ be the gluing function pictorially described by:
	\begin{center}
		\tikzset{every picture/.style={line width=0.75pt}} 
		

	\end{center}
	We organize the boundary data of $f$ used in Example~\ref{ExampleTwoDisks} to put formula (\ref{GluingFormula2disks2}) in terms of a relative gluing diagram. Let $\alpha:\Sigma\sqcup \overline{\Sigma}\sqcup\Lambda\to \partial X$ be the gluing function pictorially described by:
	
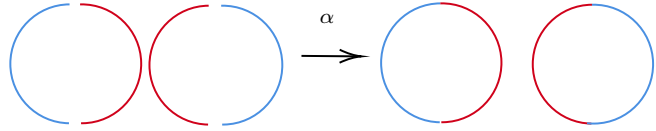
\begin{center}

		\tikzset{every picture/.style={line width=0.75pt}} 
		
		%

	\end{center}
	The function $f$ clearly satisfies the coequalizer condition with respect to $\alpha$. Let $\beta:\Lambda\to\partial X'$ be the gluing function pictorially described by:
	
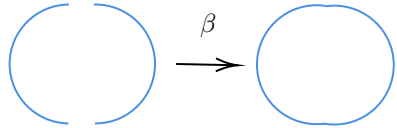
\begin{center}

		\tikzset{every picture/.style={line width=0.75pt}} 
		
		%

	\end{center}
	The function $\beta$ clearly makes the diagram appearing in Definition~\ref{relativegluingdefinition} commute. The pair $(f,\beta)$ is a relative gluing diagram with respect to the boundary decomposition $\alpha$. The diagram appearing in Remark~\ref{remarkrelativegluing} summarizing this is:
	
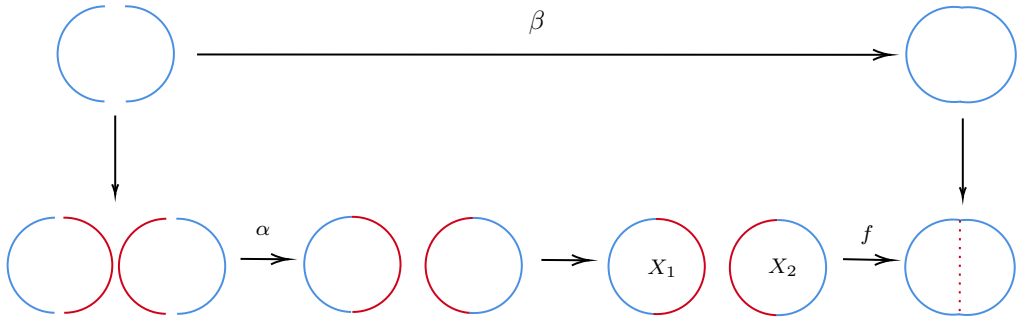
\begin{center}

		\tikzset{every picture/.style={line width=0.75pt}} 
		
		%
		
	\end{center}
\end{ex}

The following example encodes the necessary data to compute the formula appearing in Example~\ref{examplediskondiskclassic} as a relative gluing diagram.

\begin{ex}
	\label{diskondiskrelativegluing}
	Let $X,X',\Sigma$ and $\Lambda$ be as in Example~\ref{examplediskondiskclassic}. Let $f:X\to X'$ be the gluing function pictorially described by:
	\begin{center}
		\tikzset{every picture/.style={line width=1.0pt}} 
		
		%
	\end{center}
	Let $\alpha:\Sigma\sqcup \overline{\Sigma}\sqcup \Lambda\to\partial X$ be the gluing function of Example~\ref{Exhomeomorphisms}, i.e. let $\alpha$ be the gluing function pictorially described by:
	
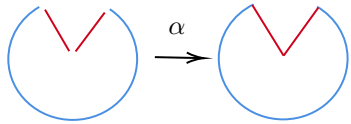
\begin{center}

		\tikzset{every picture/.style={line width=0.75pt}} 
		
		%
		
	\end{center}
	The function $f$ described above clearly satisfies the coequalizer condition with respect to $\alpha$. Let $\beta:\Lambda\to \partial X'$ be the gluing function appearing in Example~\ref{ExampleClosedStrings}. Pictorially:
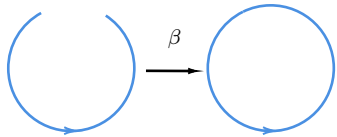
	\begin{center}
		\tikzset{every picture/.style={line width=1.0pt}} 
		
		%
	\end{center}
	This clearly makes the second diagram in Definition~\ref{relativegluingdefinition} commute. The pair $(f,\beta)$ is thus a gluing diagram relative to $\alpha$. We summarize this, as in Remark~\ref{remarkrelativegluing}, in the diagram:
	
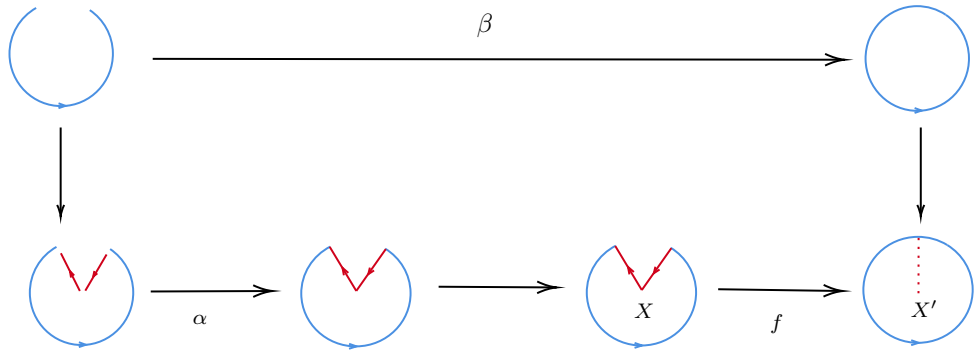
\begin{center}

		\tikzset{every picture/.style={line width=0.75pt}} 
		
		%
		
	\end{center}
\end{ex}

In TQFT a cylinder is an elementary cobordism. In our present setting the cylinder can be obtained by gluing the disk to itself.

\begin{ex}
\label{mainexamplerelativegluing}
Let $X,X'$ be the square $[0,1]^2$ and the cylinder $S^1\times [0,1]$ respectively. Let $f$ to be the function that glues $X$ along its top and bottom edges. Pictorially:
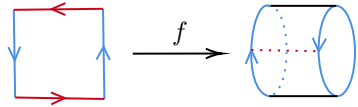
\begin{center}

\tikzset{every picture/.style={line width=0.75pt}} 

%

\end{center}
We let $\Sigma$ be the bottom edge of the disk $[0,1]^2$ (red), $\overline{\Sigma}$ the top edge of $[0,1]^2$ (also red), and we let $\Lambda$ be the disjoint union of the left and right edges of $[0,1]^2$ (blue). Let $\alpha:\Sigma\sqcup \overline{\Sigma}\sqcup \Lambda\to\partial X$ be the function described pictorially as:
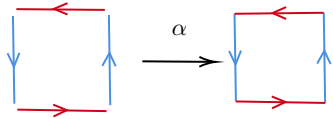
\begin{center}
\tikzset{every picture/.style={line width=0.75pt}} 

%
\end{center}
The function $f$ described above clearly satisfies the coequalizer condition with respect to $\alpha$. We let $\beta:\Lambda\to \partial X'$ be the function described pictorially by:
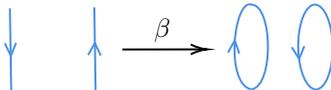
\begin{center}

\tikzset{every picture/.style={line width=0.75pt}} 

%

\end{center}
This is a gluing function by Example~\ref{ExampleClosedStrings} and Observation~\ref{observationgluingfunctionsclosed} and clearly makes the second diagram in Definition~\ref{relativegluingdefinition} commute. The pair $(f,\beta)$ is thus a gluing diagram relative to $\alpha$. We summarize this, as in Remark~\ref{remarkrelativegluing}, in the diagram:
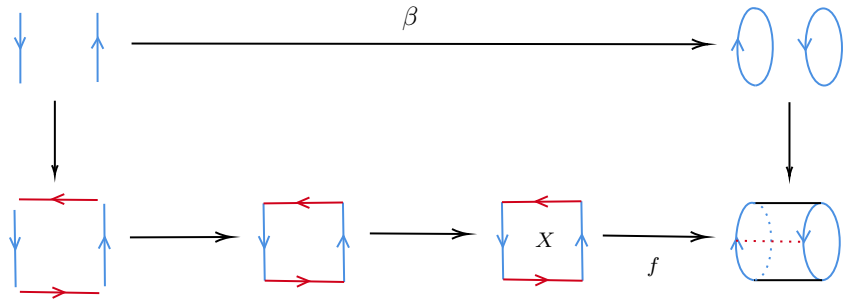
\begin{center}

\tikzset{every picture/.style={line width=0.75pt}} 

%

\end{center}
\end{ex}

The following example shows that Definition~\ref{relativegluingdefinition} generalizes the usual gluing operation between cobordisms \cite{Ati:tqft}. Example~\ref{diskondiskrelativegluing} shows that Definition~\ref{relativegluingdefinition} captures a more general notion of gluing.

\begin{ex}\label{examplecobordismsasrelativegluings}
	Let $\Sigma_i$, $i\in\left\{1,2,3\right\}$ be hypersurfaces. Let $X_1,X_2$ be cobordisms from $\Sigma_1$ to $\Sigma_2$ and from $\Sigma_2$ to $\Sigma_3$ respectively. Let $X_2\circ X_1$ be the cobordism obtained from $X_1$ and $X_2$ by gluing along a common collar neighborhood of $\Sigma_2$. In the notation of Definition~\ref{relativegluingdefinition} we take $X$ to be $X_1\sqcup X_2$, $X'$ to be $X_2\circ X_1$, $\Sigma$ to be $\Sigma_2$, $\Lambda$ to be $\Sigma_1\sqcup \overline{\Sigma_3}$, $\alpha$ to be the homeomorphism from $\Sigma\sqcup\overline{\Sigma}\sqcup \Lambda$ to $\partial X$ defining the cobordism structures of $X_1$ and $X_2$, and $\beta$ the homeomorphism from $\Lambda$ to $\partial(X_2\circ X_1)$ providing its cobordism structure. With this notation $(f,\beta)$ is a gluing diagram with respect to $\alpha$. An example of this is the composition of cobordisms described pictorially as:
	\begin{center}

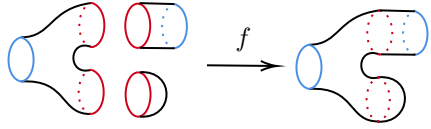
		\tikzset{every picture/.style={line width=0.75pt}} 
		
		\begin{tikzpicture}[x=0.75pt,y=0.75pt,yscale=-0.75,xscale=0.75]
			
			\draw  [color={rgb, 255:red, 208; green, 2; blue, 27 }  ,draw opacity=1 ] (265.7,139.81) .. controls (265.7,131.59) and (269.45,124.94) .. (274.08,124.94) .. controls (278.71,124.94) and (282.46,131.59) .. (282.46,139.81) .. controls (282.46,148.02) and (278.71,154.67) .. (274.08,154.67) .. controls (269.45,154.67) and (265.7,148.02) .. (265.7,139.81) -- cycle ;
			\draw  [color={rgb, 255:red, 208; green, 2; blue, 27 }  ,draw opacity=1 ] (265.93,186.06) .. controls (265.93,177.85) and (269.68,171.19) .. (274.31,171.19) .. controls (278.93,171.19) and (282.68,177.85) .. (282.68,186.06) .. controls (282.68,194.27) and (278.93,200.93) .. (274.31,200.93) .. controls (269.68,200.93) and (265.93,194.27) .. (265.93,186.06) -- cycle ;
			\draw    (275.51,171.03) .. controls (299.75,169.51) and (299.07,202.72) .. (275.5,200.94) ;
			\draw    (195.66,147.56) .. controls (216.78,147.56) and (214.32,127.99) .. (243.25,124.43) ;
			\draw    (195.66,177.29) .. controls (227.5,175.51) and (221.47,199.91) .. (242.8,199.15) ;
			\draw    (242.8,169.42) .. controls (226.38,174.24) and (227.05,147.56) .. (243.25,154.17) ;
			\draw [color={rgb, 255:red, 208; green, 2; blue, 27 }  ,draw opacity=1 ]   (243.25,124.43) .. controls (254.98,127.99) and (253.86,152.64) .. (243.25,154.17) ;
			\draw [color={rgb, 255:red, 208; green, 2; blue, 27 }  ,draw opacity=1 ]   (242.8,169.42) .. controls (254.53,172.97) and (253.42,197.63) .. (242.8,199.15) ;
			\draw [color={rgb, 255:red, 74; green, 144; blue, 226 }  ,draw opacity=1 ]   (195.66,147.56) .. controls (207.39,151.12) and (206.28,175.77) .. (195.66,177.29) ;
			\draw [color={rgb, 255:red, 74; green, 144; blue, 226 }  ,draw opacity=1 ]   (298.75,125.19) .. controls (310.48,128.75) and (309.36,153.4) .. (298.75,154.92) ;
			\draw [color={rgb, 255:red, 74; green, 144; blue, 226 }  ,draw opacity=1 ]   (195.66,177.29) .. controls (183.94,175.01) and (184.83,149.84) .. (195.66,147.56) ;
			\draw [color={rgb, 255:red, 74; green, 144; blue, 226 }  ,draw opacity=1 ] [dash pattern={on 0.84pt off 2.51pt}]  (298.88,154.94) .. controls (287.15,152.65) and (288.04,127.49) .. (298.88,125.2) ;
			\draw [color={rgb, 255:red, 208; green, 2; blue, 27 }  ,draw opacity=1 ] [dash pattern={on 0.84pt off 2.51pt}]  (242.8,199.15) .. controls (231.08,196.86) and (231.97,171.7) .. (242.8,169.42) ;
			\draw [color={rgb, 255:red, 208; green, 2; blue, 27 }  ,draw opacity=1 ] [dash pattern={on 0.84pt off 2.51pt}]  (243.25,154.17) .. controls (231.52,151.88) and (232.42,126.72) .. (243.25,124.43) ;
			\draw    (436.41,174.42) .. controls (460.65,172.89) and (458.77,206.09) .. (435.2,204.31) ;
			\draw    (389.27,152.56) .. controls (410.38,152.56) and (407.92,132.99) .. (436.85,129.43) ;
			\draw    (389.27,182.29) .. controls (421.1,180.51) and (415.07,204.91) .. (436.41,204.15) ;
			\draw    (436.41,174.42) .. controls (419.99,179.24) and (420.66,152.56) .. (436.85,159.17) ;
			\draw [color={rgb, 255:red, 208; green, 2; blue, 27 }  ,draw opacity=1 ] [dash pattern={on 0.84pt off 2.51pt}]  (436.85,129.43) .. controls (448.58,132.99) and (447.46,157.64) .. (436.85,159.17) ;
			\draw [color={rgb, 255:red, 208; green, 2; blue, 27 }  ,draw opacity=1 ] [dash pattern={on 0.84pt off 2.51pt}]  (435.2,174.58) .. controls (446.93,178.14) and (445.81,202.79) .. (435.2,204.31) ;
			\draw [color={rgb, 255:red, 74; green, 144; blue, 226 }  ,draw opacity=1 ]   (389.27,152.56) .. controls (401,156.12) and (399.88,180.77) .. (389.27,182.29) ;
			\draw [color={rgb, 255:red, 74; green, 144; blue, 226 }  ,draw opacity=1 ]   (389.27,182.29) .. controls (377.54,180.01) and (378.43,154.84) .. (389.27,152.56) ;
			\draw [color={rgb, 255:red, 208; green, 2; blue, 27 }  ,draw opacity=1 ] [dash pattern={on 0.84pt off 2.51pt}]  (436.41,204.15) .. controls (424.68,201.86) and (425.57,176.7) .. (436.41,174.42) ;
			\draw [color={rgb, 255:red, 208; green, 2; blue, 27 }  ,draw opacity=1 ] [dash pattern={on 0.84pt off 2.51pt}]  (436.85,159.17) .. controls (425.12,156.88) and (426.02,131.72) .. (436.85,129.43) ;
			\draw    (320.67,166.4) -- (365.75,166.85) ;
			\draw [shift={(367.75,166.88)}, rotate = 180.58] [color={rgb, 255:red, 0; green, 0; blue, 0 }  ][line width=0.75]    (10.93,-3.29) .. controls (6.95,-1.4) and (3.31,-0.3) .. (0,0) .. controls (3.31,0.3) and (6.95,1.4) .. (10.93,3.29)   ;
			\draw    (275.25,124.56) -- (298.75,125.19) ;
			\draw    (275.63,154.69) -- (298.88,154.94) ;
			\draw [color={rgb, 255:red, 74; green, 144; blue, 226 }  ,draw opacity=1 ]   (460.35,130.05) .. controls (472.08,133.61) and (471.36,157.73) .. (460.75,159.26) ;
			\draw [color={rgb, 255:red, 74; green, 144; blue, 226 }  ,draw opacity=1 ] [dash pattern={on 0.84pt off 2.51pt}]  (460.88,159.27) .. controls (449.15,156.98) and (450.04,131.82) .. (460.88,129.53) ;
			\draw    (436.85,129.43) -- (460.35,130.05) ;
			\draw    (437.63,159.02) -- (460.88,159.27) ;
			
			\draw (337.75,138.75) node [anchor=north west][inner sep=0.75pt]    {$f$};

		\end{tikzpicture}
	\end{center}
	where $X_1$ is the left pair of pants, $X_2$ is the union of the cylinder and the disk on the right, $\Lambda$ is the union of the two closed strings on the extreme left and right (blue), $\Sigma$ is the union of the two right-facing closed strings (red) on the left and $\overline{\Sigma}$ is the union of the two left-facing closed strings on the right (red). The relative gluing diagram describing the composition $X_2\circ X_1$ is:
	\begin{center}

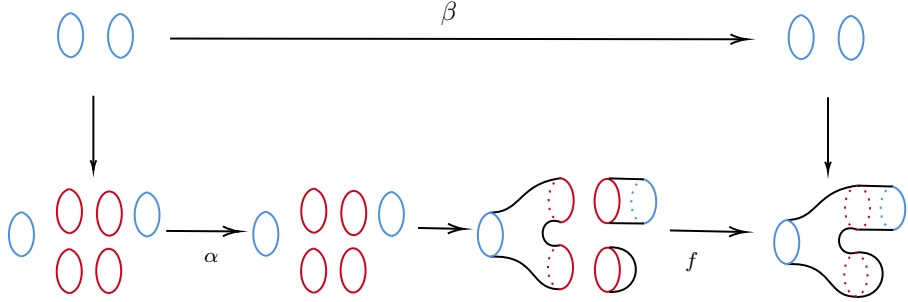
		\tikzset{every picture/.style={line width=0.75pt}} 
		
		\begin{tikzpicture}[x=0.75pt,y=0.75pt,yscale=-0.75,xscale=0.75]
			
			\draw    (120.33,200.5) -- (167.33,200.82) ;
			\draw [shift={(169.33,200.83)}, rotate = 180.39] [color={rgb, 255:red, 0; green, 0; blue, 0 }  ][line width=0.75]    (10.93,-3.29) .. controls (6.95,-1.4) and (3.31,-0.3) .. (0,0) .. controls (3.31,0.3) and (6.95,1.4) .. (10.93,3.29)   ;
			\draw    (122.92,71.03) -- (508.63,71.53) ;
			\draw [shift={(510.63,71.53)}, rotate = 180.07] [color={rgb, 255:red, 0; green, 0; blue, 0 }  ][line width=0.75]    (10.93,-3.29) .. controls (6.95,-1.4) and (3.31,-0.3) .. (0,0) .. controls (3.31,0.3) and (6.95,1.4) .. (10.93,3.29)   ;
			\draw    (289.67,198.5) -- (318.67,199.12) ;
			\draw [shift={(320.67,199.17)}, rotate = 181.23] [color={rgb, 255:red, 0; green, 0; blue, 0 }  ][line width=0.75]    (10.93,-3.29) .. controls (6.95,-1.4) and (3.31,-0.3) .. (0,0) .. controls (3.31,0.3) and (6.95,1.4) .. (10.93,3.29)   ;
			\draw    (459,200.17) -- (507.79,201.2) ;
			\draw [shift={(509.79,201.24)}, rotate = 181.21] [color={rgb, 255:red, 0; green, 0; blue, 0 }  ][line width=0.75]    (10.93,-3.29) .. controls (6.95,-1.4) and (3.31,-0.3) .. (0,0) .. controls (3.31,0.3) and (6.95,1.4) .. (10.93,3.29)   ;
			\draw    (71.33,109.53) -- (71.33,157.53) ;
			\draw [shift={(71.33,159.53)}, rotate = 270] [color={rgb, 255:red, 0; green, 0; blue, 0 }  ][line width=0.75]    (6.56,-1.97) .. controls (4.17,-0.84) and (1.99,-0.18) .. (0,0) .. controls (1.99,0.18) and (4.17,0.84) .. (6.56,1.97)   ;
			\draw    (565.33,110.53) -- (565.33,158.53) ;
			\draw [shift={(565.33,160.53)}, rotate = 270] [color={rgb, 255:red, 0; green, 0; blue, 0 }  ][line width=0.75]    (6.56,-1.97) .. controls (4.17,-0.84) and (1.99,-0.18) .. (0,0) .. controls (1.99,0.18) and (4.17,0.84) .. (6.56,1.97)   ;
			\draw  [color={rgb, 255:red, 208; green, 2; blue, 27 }  ,draw opacity=1 ] (238.04,188.47) .. controls (238.04,180.26) and (241.79,173.6) .. (246.41,173.6) .. controls (251.04,173.6) and (254.79,180.26) .. (254.79,188.47) .. controls (254.79,196.68) and (251.04,203.34) .. (246.41,203.34) .. controls (241.79,203.34) and (238.04,196.68) .. (238.04,188.47) -- cycle ;
			\draw [color={rgb, 255:red, 74; green, 144; blue, 226 }  ,draw opacity=1 ]   (271.75,174.19) .. controls (283.48,177.75) and (282.36,202.4) .. (271.75,203.92) ;
			\draw [color={rgb, 255:red, 74; green, 144; blue, 226 }  ,draw opacity=1 ]   (271.88,203.94) .. controls (260.15,201.65) and (261.04,176.49) .. (271.88,174.2) ;
			\draw [color={rgb, 255:red, 74; green, 144; blue, 226 }  ,draw opacity=1 ]   (56.06,53.91) .. controls (67.79,57.47) and (66.68,82.12) .. (56.06,83.64) ;
			\draw [color={rgb, 255:red, 74; green, 144; blue, 226 }  ,draw opacity=1 ]   (56.06,83.64) .. controls (44.34,81.36) and (45.23,56.19) .. (56.06,53.91) ;
			\draw [color={rgb, 255:red, 74; green, 144; blue, 226 }  ,draw opacity=1 ]   (89.55,54.34) .. controls (101.28,57.9) and (100.16,82.55) .. (89.55,84.07) ;
			\draw [color={rgb, 255:red, 74; green, 144; blue, 226 }  ,draw opacity=1 ]   (89.68,84.09) .. controls (77.95,81.8) and (78.84,56.64) .. (89.68,54.35) ;
			\draw [color={rgb, 255:red, 74; green, 144; blue, 226 }  ,draw opacity=1 ]   (547.33,55.24) .. controls (559.06,58.8) and (557.95,83.45) .. (547.33,84.98) ;
			\draw [color={rgb, 255:red, 74; green, 144; blue, 226 }  ,draw opacity=1 ]   (547.33,84.98) .. controls (535.6,82.69) and (536.5,57.53) .. (547.33,55.24) ;
			\draw [color={rgb, 255:red, 74; green, 144; blue, 226 }  ,draw opacity=1 ]   (580.82,55.67) .. controls (592.55,59.23) and (591.43,83.88) .. (580.82,85.41) ;
			\draw [color={rgb, 255:red, 74; green, 144; blue, 226 }  ,draw opacity=1 ]   (580.94,85.42) .. controls (569.21,83.13) and (570.11,57.97) .. (580.94,55.68) ;
			\draw    (584.74,214.99) .. controls (608.98,213.46) and (607.1,246.66) .. (583.53,244.88) ;
			\draw    (537.6,193.13) .. controls (558.71,193.13) and (556.25,173.56) .. (585.19,170) ;
			\draw    (537.6,222.87) .. controls (569.44,221.09) and (563.4,245.49) .. (584.74,244.72) ;
			\draw    (584.74,214.99) .. controls (568.32,219.82) and (568.99,193.13) .. (585.19,199.74) ;
			\draw [color={rgb, 255:red, 208; green, 2; blue, 27 }  ,draw opacity=1 ] [dash pattern={on 0.84pt off 2.51pt}]  (585.19,170) .. controls (596.92,173.56) and (595.8,198.21) .. (585.19,199.74) ;
			\draw [color={rgb, 255:red, 208; green, 2; blue, 27 }  ,draw opacity=1 ] [dash pattern={on 0.84pt off 2.51pt}]  (583.53,215.15) .. controls (595.26,218.71) and (594.14,243.36) .. (583.53,244.88) ;
			\draw [color={rgb, 255:red, 74; green, 144; blue, 226 }  ,draw opacity=1 ]   (537.6,193.13) .. controls (549.33,196.69) and (548.21,221.34) .. (537.6,222.87) ;
			\draw [color={rgb, 255:red, 74; green, 144; blue, 226 }  ,draw opacity=1 ]   (537.6,222.87) .. controls (525.87,220.58) and (526.76,195.42) .. (537.6,193.13) ;
			\draw [color={rgb, 255:red, 208; green, 2; blue, 27 }  ,draw opacity=1 ] [dash pattern={on 0.84pt off 2.51pt}]  (584.74,244.72) .. controls (573.01,242.44) and (573.9,217.27) .. (584.74,214.99) ;
			\draw [color={rgb, 255:red, 208; green, 2; blue, 27 }  ,draw opacity=1 ] [dash pattern={on 0.84pt off 2.51pt}]  (585.19,199.74) .. controls (573.46,197.45) and (574.35,172.29) .. (585.19,170) ;
			\draw [color={rgb, 255:red, 74; green, 144; blue, 226 }  ,draw opacity=1 ]   (608.69,170.63) .. controls (620.42,174.18) and (619.7,198.3) .. (609.08,199.83) ;
			\draw [color={rgb, 255:red, 74; green, 144; blue, 226 }  ,draw opacity=1 ] [dash pattern={on 0.84pt off 2.51pt}]  (609.21,199.84) .. controls (597.48,197.55) and (598.37,172.39) .. (609.21,170.11) ;
			\draw    (585.19,170) -- (608.69,170.63) ;
			\draw    (585.96,199.59) -- (609.21,199.84) ;
			\draw  [color={rgb, 255:red, 208; green, 2; blue, 27 }  ,draw opacity=1 ] (408.37,180.38) .. controls (408.37,172.17) and (412.12,165.51) .. (416.75,165.51) .. controls (421.38,165.51) and (425.13,172.17) .. (425.13,180.38) .. controls (425.13,188.59) and (421.38,195.25) .. (416.75,195.25) .. controls (412.12,195.25) and (408.37,188.59) .. (408.37,180.38) -- cycle ;
			\draw  [color={rgb, 255:red, 208; green, 2; blue, 27 }  ,draw opacity=1 ] (408.59,226.63) .. controls (408.59,218.42) and (412.34,211.77) .. (416.97,211.77) .. controls (421.6,211.77) and (425.35,218.42) .. (425.35,226.63) .. controls (425.35,234.85) and (421.6,241.5) .. (416.97,241.5) .. controls (412.34,241.5) and (408.59,234.85) .. (408.59,226.63) -- cycle ;
			\draw    (418.18,211.6) .. controls (442.42,210.08) and (441.74,243.29) .. (418.17,241.51) ;
			\draw    (338.33,188.13) .. controls (359.44,188.13) and (356.99,168.56) .. (385.92,165) ;
			\draw    (338.33,217.87) .. controls (370.17,216.09) and (364.14,240.49) .. (385.47,239.72) ;
			\draw    (385.47,209.99) .. controls (369.05,214.82) and (369.72,188.13) .. (385.92,194.74) ;
			\draw [color={rgb, 255:red, 208; green, 2; blue, 27 }  ,draw opacity=1 ]   (385.92,165) .. controls (397.65,168.56) and (396.53,193.21) .. (385.92,194.74) ;
			\draw [color={rgb, 255:red, 208; green, 2; blue, 27 }  ,draw opacity=1 ]   (385.47,209.99) .. controls (397.2,213.54) and (396.08,238.2) .. (385.47,239.72) ;
			\draw [color={rgb, 255:red, 74; green, 144; blue, 226 }  ,draw opacity=1 ]   (338.33,188.13) .. controls (350.06,191.69) and (348.94,216.34) .. (338.33,217.87) ;
			\draw [color={rgb, 255:red, 74; green, 144; blue, 226 }  ,draw opacity=1 ]   (441.42,165.76) .. controls (453.15,169.32) and (452.03,193.97) .. (441.42,195.5) ;
			\draw [color={rgb, 255:red, 74; green, 144; blue, 226 }  ,draw opacity=1 ]   (338.33,217.87) .. controls (326.6,215.58) and (327.5,190.42) .. (338.33,188.13) ;
			\draw [color={rgb, 255:red, 74; green, 144; blue, 226 }  ,draw opacity=1 ] [dash pattern={on 0.84pt off 2.51pt}]  (441.54,195.51) .. controls (429.81,193.22) and (430.71,168.06) .. (441.54,165.77) ;
			\draw [color={rgb, 255:red, 208; green, 2; blue, 27 }  ,draw opacity=1 ] [dash pattern={on 0.84pt off 2.51pt}]  (385.47,239.72) .. controls (373.74,237.44) and (374.64,212.27) .. (385.47,209.99) ;
			\draw [color={rgb, 255:red, 208; green, 2; blue, 27 }  ,draw opacity=1 ] [dash pattern={on 0.84pt off 2.51pt}]  (385.92,194.74) .. controls (374.19,192.45) and (375.08,167.29) .. (385.92,165) ;
			\draw    (417.92,165.13) -- (441.42,165.76) ;
			\draw    (418.29,195.26) -- (441.54,195.51) ;
			\draw [color={rgb, 255:red, 74; green, 144; blue, 226 }  ,draw opacity=1 ]   (187.33,187.56) .. controls (199.06,191.12) and (197.94,215.77) .. (187.33,217.29) ;
			\draw [color={rgb, 255:red, 74; green, 144; blue, 226 }  ,draw opacity=1 ]   (187.33,217.29) .. controls (175.6,215.01) and (176.5,189.84) .. (187.33,187.56) ;
			\draw [color={rgb, 255:red, 208; green, 2; blue, 27 }  ,draw opacity=1 ]   (219.47,212.33) .. controls (231.2,215.89) and (230.08,240.55) .. (219.47,242.07) ;
			\draw [color={rgb, 255:red, 208; green, 2; blue, 27 }  ,draw opacity=1 ]   (219.47,242.07) .. controls (207.74,239.78) and (208.64,214.62) .. (219.47,212.33) ;
			\draw [color={rgb, 255:red, 208; green, 2; blue, 27 }  ,draw opacity=1 ]   (219.58,172) .. controls (231.31,175.56) and (230.2,200.21) .. (219.58,201.74) ;
			\draw [color={rgb, 255:red, 208; green, 2; blue, 27 }  ,draw opacity=1 ]   (219.58,201.74) .. controls (207.86,199.45) and (208.75,174.29) .. (219.58,172) ;
			\draw  [color={rgb, 255:red, 208; green, 2; blue, 27 }  ,draw opacity=1 ] (237.59,227.06) .. controls (237.59,218.85) and (241.34,212.19) .. (245.97,212.19) .. controls (250.6,212.19) and (254.35,218.85) .. (254.35,227.06) .. controls (254.35,235.27) and (250.6,241.93) .. (245.97,241.93) .. controls (241.34,241.93) and (237.59,235.27) .. (237.59,227.06) -- cycle ;
			\draw  [color={rgb, 255:red, 208; green, 2; blue, 27 }  ,draw opacity=1 ] (73.7,188.81) .. controls (73.7,180.59) and (77.45,173.94) .. (82.08,173.94) .. controls (86.71,173.94) and (90.46,180.59) .. (90.46,188.81) .. controls (90.46,197.02) and (86.71,203.67) .. (82.08,203.67) .. controls (77.45,203.67) and (73.7,197.02) .. (73.7,188.81) -- cycle ;
			\draw [color={rgb, 255:red, 74; green, 144; blue, 226 }  ,draw opacity=1 ]   (107.42,174.52) .. controls (119.15,178.08) and (118.03,202.73) .. (107.42,204.26) ;
			\draw [color={rgb, 255:red, 74; green, 144; blue, 226 }  ,draw opacity=1 ]   (107.54,204.27) .. controls (95.81,201.98) and (96.71,176.82) .. (107.54,174.53) ;
			\draw [color={rgb, 255:red, 74; green, 144; blue, 226 }  ,draw opacity=1 ]   (23,187.89) .. controls (34.73,191.45) and (33.61,216.1) .. (23,217.63) ;
			\draw [color={rgb, 255:red, 74; green, 144; blue, 226 }  ,draw opacity=1 ]   (23,217.63) .. controls (11.27,215.34) and (12.16,190.18) .. (23,187.89) ;
			\draw [color={rgb, 255:red, 208; green, 2; blue, 27 }  ,draw opacity=1 ]   (55.14,212.67) .. controls (66.87,216.22) and (65.75,240.88) .. (55.14,242.4) ;
			\draw [color={rgb, 255:red, 208; green, 2; blue, 27 }  ,draw opacity=1 ]   (55.14,242.4) .. controls (43.41,240.12) and (44.3,214.95) .. (55.14,212.67) ;
			\draw [color={rgb, 255:red, 208; green, 2; blue, 27 }  ,draw opacity=1 ]   (55.25,172.33) .. controls (66.98,175.89) and (65.86,200.55) .. (55.25,202.07) ;
			\draw [color={rgb, 255:red, 208; green, 2; blue, 27 }  ,draw opacity=1 ]   (55.25,202.07) .. controls (43.52,199.78) and (44.42,174.62) .. (55.25,172.33) ;
			\draw  [color={rgb, 255:red, 208; green, 2; blue, 27 }  ,draw opacity=1 ] (73.26,227.4) .. controls (73.26,219.18) and (77.01,212.53) .. (81.64,212.53) .. controls (86.27,212.53) and (90.02,219.18) .. (90.02,227.4) .. controls (90.02,235.61) and (86.27,242.26) .. (81.64,242.26) .. controls (77.01,242.26) and (73.26,235.61) .. (73.26,227.4) -- cycle ;
			
			\draw (143.77,213.23) node [anchor=north west][inner sep=0.75pt]  [font=\footnotesize]  {$\alpha $};
			\draw (467.33,212.77) node [anchor=north west][inner sep=0.75pt]  [font=\footnotesize]  {$f$};
			\draw (302.86,44.1) node [anchor=north west][inner sep=0.75pt]    {$\beta $};

		\end{tikzpicture}

	\end{center}
\end{ex}

The following immediate observation says that we may construct complicated relative gluing diagrams out of simpler relative gluing diagrams through the operations of taking disjoint unions and reversing orientations.

\begin{obs}\label{observationrelativegluingclosed}
The collection of relative gluing diagrams is closed under orientation change and disjoint union. Precisely, if $(f,\beta)$ is a gluing diagram relative to $\alpha$, then $(\overline{f},\overline{\beta})$ is a gluing diagram relative to $\overline{\alpha}$, and if $(f_i,\beta_i)$ are gluing diagrams relative to $\alpha_i$, with $i\in\left\{1,2\right\}$, then $(f_1\sqcup f_2,\alpha_1\sqcup \alpha_2)$ is a gluing diagram relative to $\beta_1\sqcup\beta_2$.
\end{obs}

\begin{obs}\label{ObsHomeomorphismsRelative}
	Every homeomorphism fits in a relative gluing diagram. Precisely, let $X$ and $X'$ be regions. Let $f:X\to X'$ be a homeomorphism. If in the notation of Definition~\ref{relativegluingdefinition} we set $\Sigma$ to be the empty hypersurface $\emptyset$, $\Lambda$ to be the boundary $\partial X$ of $X$, $\alpha:\Sigma\sqcup\overline{\Sigma}\sqcup \Lambda\to \partial X$ be the identity and $\beta:\Lambda\to \partial X'$ to be the restriction $f|_{\partial X}$ of $f$ to $\partial X$, then it is easily seen that the diagram is a relative gluing diagram:
	\begin{center}
	\begin{tikzpicture}
			\matrix (m) [matrix of math nodes,row sep=3em,column sep=3em,minimum width=2em]
			{\Lambda&&&\partial X' \\
				\Sigma\sqcup \overline{\Sigma}\sqcup\Lambda&\partial X&X&X'\\};
			\path[-stealth]
			(m-1-1) edge node [above] {$\beta$} (m-1-4)
			edge node [left]  {} (m-2-1)
			
			(m-2-1) edge node [below] {$\alpha$} (m-2-2)
			(m-2-2) edge node [below]  {} (m-2-3)
			(m-2-3) edge node [below] {$f$} (m-2-4)
			(m-1-4) edge node [right] {} (m-2-4)
			
			;
		\end{tikzpicture}
	\end{center}
\end{obs}

The following observation says that every relative gluing diagram can be "deformed" by a homeomorphism on the left and on the right into another relative gluing diagram.

\begin{obs}
	\label{observationRelativeClosedCompositionHomeo}
	Let 
	\begin{center}
		
		\begin{tikzpicture}
			\matrix (m) [matrix of math nodes,row sep=3em,column sep=3em,minimum width=2em]
			{\Lambda&&&\partial X' \\
				\Sigma\sqcup \overline{\Sigma}\sqcup\Lambda&\partial X&X&X'\\};
			\path[-stealth]
			(m-1-1) edge node [above] {$\beta$} (m-1-4)
			edge node [left]  {} (m-2-1)
			
			(m-2-1) edge node [below] {$\alpha$} (m-2-2)
			(m-2-2) edge node [below]  {} (m-2-3)
			(m-2-3) edge node [below] {$f$} (m-2-4)
			(m-1-4) edge node [right] {} (m-2-4)
			
			;
		\end{tikzpicture}
	\end{center}
	be a relative gluing diagram. Let $g:Y\to X$ and $h:X'\to Y'$ be homeomorphisms. In that case it is easily seen that the diagrams
	\begin{center}
		\begin{tikzpicture}
			\matrix (m) [matrix of math nodes,row sep=3em,column sep=3em,minimum width=2em]
			{\Lambda&&&\partial Y' \\
				\Sigma\sqcup \overline{\Sigma}\sqcup\Lambda&\partial X&X&Y'\\};
			\path[-stealth]
			(m-1-1) edge node [above] {$h|_{\partial X'}\beta$} (m-1-4)
			edge node [left]  {} (m-2-1)
			
			(m-2-1) edge node [below] {$\alpha$} (m-2-2)
			(m-2-2) edge node [below]  {} (m-2-3)
			(m-2-3) edge node [below] {$hf$} (m-2-4)
			(m-1-4) edge node [right] {} (m-2-4)
			
			;
		\end{tikzpicture}
	\end{center}
	and
	\begin{center}
		\begin{tikzpicture}
			\matrix (m) [matrix of math nodes,row sep=3em,column sep=3em,minimum width=2em]
			{g^{-1}\Lambda&&&\partial X' \\
				g^{-1}\Sigma\sqcup g^{-1}\overline{\Sigma}\sqcup g^{-1}\Lambda&\partial Y&Y&X'\\};
			\path[-stealth]
			(m-1-1) edge node [above] {$\beta g|_{g^{-1}\Lambda}$} (m-1-4)
			edge node [left]  {} (m-2-1)
			
			(m-2-1) edge node [below] {$g^{-1}\alpha g$} (m-2-2)
			(m-2-2) edge node [below]  {} (m-2-3)
			(m-2-3) edge node [below] {$fg$} (m-2-4)
			(m-1-4) edge node [right] {} (m-2-4)
			
			;
		\end{tikzpicture}
	\end{center}
	are relative gluing diagrams. In the above cases we will usually write the composite boundary gluing functions $h|_{\partial X'}\beta$ and $\beta g|_{g^{-1}\Lambda}$ as $h\beta$ and $\beta g$. With this notation the "deformed" relative gluing diagrams above can be succinctly written as $(hf,h\beta)$ and $(fg,\beta g)$ respectively.
\end{obs}

We interpret Observation~\ref{ObsHomeomorphismsRelative} as the fact that any relative gluing diagram can be composed with compatible relative diagrams associated to homeomorphisms, as in Observation~\ref{ObsHomeomorphismsRelative} on the left and on the right. It is not clear however, how relative gluing data should behave with respect to composition in general. Consider for example the gluing function $f:[0,1]^2\to S^1\times [0,1]$ pictured as:
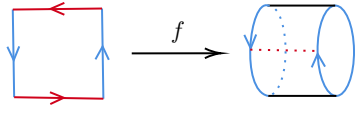
\begin{center}

	\tikzset{every picture/.style={line width=0.75pt}} 
	
	\begin{tikzpicture}[x=0.75pt,y=0.75pt,yscale=-0.75,xscale=0.75]
		
		\draw [color={rgb, 255:red, 74; green, 144; blue, 226 }  ,draw opacity=1 ]   (216,90) -- (216.92,149) ;
		\draw [color={rgb, 255:red, 74; green, 144; blue, 226 }  ,draw opacity=1 ]   (276.01,89.31) -- (276.92,150) ;
		\draw  [color={rgb, 255:red, 74; green, 144; blue, 226 }  ,draw opacity=1 ] (271.6,123.51) -- (276.4,114.84) -- (280.59,123.82) ;
		\draw  [color={rgb, 255:red, 74; green, 144; blue, 226 }  ,draw opacity=1 ] (220.35,114.79) -- (216.64,124.08) -- (212.36,115.05) ;
		\draw [color={rgb, 255:red, 208; green, 2; blue, 27 }  ,draw opacity=1 ]   (216,90) -- (276.01,89.31) ;
		\draw [color={rgb, 255:red, 208; green, 2; blue, 27 }  ,draw opacity=1 ]   (216.92,149) -- (276.92,150) ;
		\draw  [color={rgb, 255:red, 208; green, 2; blue, 27 }  ,draw opacity=1 ] (250.68,93.5) -- (241.52,89.5) -- (250.68,85.5) ;
		\draw  [color={rgb, 255:red, 208; green, 2; blue, 27 }  ,draw opacity=1 ] (241.01,145.63) -- (250.28,149.37) -- (241.23,153.62) ;
		\draw  [color={rgb, 255:red, 74; green, 144; blue, 226 }  ,draw opacity=1 ] (432.82,87.32) .. controls (439.59,87.23) and (445.25,100.77) .. (445.46,117.56) .. controls (445.68,134.34) and (440.37,148.02) .. (433.61,148.11) .. controls (426.84,148.2) and (421.19,134.66) .. (420.97,117.87) .. controls (420.75,101.08) and (426.06,87.41) .. (432.82,87.32) -- cycle ;
		\draw    (387.32,87.32) -- (432.82,87.32) ;
		\draw    (387.72,148.11) -- (433.61,148.11) ;
		\draw [color={rgb, 255:red, 74; green, 144; blue, 226 }  ,draw opacity=1 ] [dash pattern={on 0.84pt off 2.51pt}]  (387.32,87.32) .. controls (402.91,86.53) and (403.91,145.87) .. (387.72,148.11) ;
		\draw [color={rgb, 255:red, 74; green, 144; blue, 226 }  ,draw opacity=1 ]   (387.72,87.31) .. controls (375.14,82.86) and (368.47,144.53) .. (387.72,148.11) ;
		\draw  [color={rgb, 255:red, 74; green, 144; blue, 226 }  ,draw opacity=1 ] (417.01,128.91) -- (421,119.84) -- (425.99,128.41) ;
		\draw  [color={rgb, 255:red, 74; green, 144; blue, 226 }  ,draw opacity=1 ] (379.69,106.79) -- (375.98,116.08) -- (371.69,107.05) ;
		\draw [color={rgb, 255:red, 208; green, 2; blue, 27 }  ,draw opacity=1 ] [dash pattern={on 0.84pt off 2.51pt}]  (375.67,117) -- (420.97,117.87) ;
		\draw    (296,119.5) -- (354.1,119.42) ;
		\draw [shift={(356.1,119.41)}, rotate = 539.9200000000001] [color={rgb, 255:red, 0; green, 0; blue, 0 }  ][line width=0.75]    (10.93,-3.29) .. controls (6.95,-1.4) and (3.31,-0.3) .. (0,0) .. controls (3.31,0.3) and (6.95,1.4) .. (10.93,3.29)   ;
		
		\draw (320.5,95.9) node [anchor=north west][inner sep=0.75pt]  [xscale=0.9,yscale=0.9]  {$f$};
		
	\end{tikzpicture}
\end{center}
The red and blue boundary components make $f$ fit into a relative gluing diagram. Consider now the gluing function $g$ from $S^1\times [0,1]$ to the punctured torus $\mathbb{T}^2\setminus \mathbb{D}^2$ obtained by gluing two copies of the interval $[0,1]$ in the boundary of $S^1\times [0,1]$ intersecting the common interval $\Sigma$ forming $S^1\times [0,1]$. Pictorially:
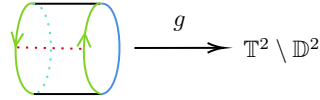
\begin{center}
	\tikzset{every picture/.style={line width=0.75pt}} 
	\begin{tikzpicture}[x=0.75pt,y=0.75pt,yscale=-0.75,xscale=0.75]
		
		\draw    (230.52,90.12) -- (276.02,90.12) ;
		\draw    (230.92,150.91) -- (276.81,150.91) ;
		\draw [color={rgb, 255:red, 80; green, 227; blue, 194 }  ,draw opacity=1 ] [dash pattern={on 0.84pt off 2.51pt}]  (230.52,90.12) .. controls (246.11,89.33) and (247.11,148.67) .. (230.92,150.91) ;
		\draw [color={rgb, 255:red, 126; green, 211; blue, 33 }  ,draw opacity=1 ]   (230.92,90.11) .. controls (218.34,85.66) and (211.67,147.33) .. (230.92,150.91) ;
		\draw  [color={rgb, 255:red, 126; green, 211; blue, 33 }  ,draw opacity=1 ] (260.27,118.87) -- (265.39,110.38) -- (269.25,119.51) ;
		\draw  [color={rgb, 255:red, 126; green, 211; blue, 33 }  ,draw opacity=1 ] (224.07,110.06) -- (218.83,118.58) -- (216.14,108.95) ;
		\draw [color={rgb, 255:red, 208; green, 2; blue, 27 }  ,draw opacity=1 ] [dash pattern={on 0.84pt off 2.51pt}]  (219.9,119.48) -- (264.17,120.47) ;
		\draw    (298.8,119.9) -- (356.9,119.82) ;
		\draw [shift={(358.9,119.81)}, rotate = 539.9200000000001] [color={rgb, 255:red, 0; green, 0; blue, 0 }  ][line width=0.75]    (10.93,-3.29) .. controls (6.95,-1.4) and (3.31,-0.3) .. (0,0) .. controls (3.31,0.3) and (6.95,1.4) .. (10.93,3.29)   ;
		\draw [color={rgb, 255:red, 126; green, 211; blue, 33 }  ,draw opacity=1 ]   (276.81,90.11) .. controls (264.23,85.66) and (257.56,147.32) .. (276.81,150.91) ;
		\draw [color={rgb, 255:red, 74; green, 144; blue, 226 }  ,draw opacity=1 ]   (276.81,90.11) .. controls (292.4,89.32) and (293.01,148.66) .. (276.81,150.91) ;
		
		\draw (323.3,96.3) node [anchor=north west][inner sep=0.75pt]  [xscale=0.9,yscale=0.9]  {$g$};
		\draw (370.37,111.47) node [anchor=north west][inner sep=0.75pt]  [xscale=0.9,yscale=0.9]  {$\mathbb{T}^{2} \setminus \mathbb{D}^{2}$};
		
	\end{tikzpicture}
\end{center}
where relative gluing is done along the front-facing half-circle intervals (green). It is unclear how to equip $gf$ with the structure of relative gluing diagram from the structure of $f$ and $g$. This can be done if we assume the green intervals above do not intersect the red interval $\Sigma$, i.e.\ if $g'$ represents the relative gluing diagram pictured as:
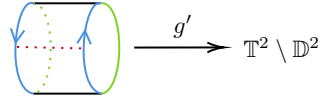
\begin{center}
	\tikzset{every picture/.style={line width=0.75pt}} 
	\begin{tikzpicture}[x=0.75pt,y=0.75pt,yscale=-0.75,xscale=0.75]
		
		\draw    (229.52,117.12) -- (275.02,117.12) ;
		\draw    (229.92,177.91) -- (275.81,177.91) ;
		\draw [color={rgb, 255:red, 126; green, 211; blue, 33 }  ,draw opacity=1 ] [dash pattern={on 0.84pt off 2.51pt}]  (229.52,117.12) .. controls (245.11,116.33) and (246.11,175.67) .. (229.92,177.91) ;
		\draw [color={rgb, 255:red, 74; green, 144; blue, 226 }  ,draw opacity=1 ]   (229.92,117.11) .. controls (217.34,112.66) and (210.67,174.33) .. (229.92,177.91) ;
		\draw  [color={rgb, 255:red, 74; green, 144; blue, 226 }  ,draw opacity=1 ] (259.45,144.05) -- (264.67,135.57) -- (268.63,144.7) ;
		\draw [color={rgb, 255:red, 208; green, 2; blue, 27 }  ,draw opacity=1 ] [dash pattern={on 0.84pt off 2.51pt}]  (218.9,146.48) -- (263.17,147.47) ;
		\draw    (297.8,146.9) -- (355.9,146.82) ;
		\draw [shift={(357.9,146.81)}, rotate = 539.9200000000001] [color={rgb, 255:red, 0; green, 0; blue, 0 }  ][line width=0.75]    (10.93,-3.29) .. controls (6.95,-1.4) and (3.31,-0.3) .. (0,0) .. controls (3.31,0.3) and (6.95,1.4) .. (10.93,3.29)   ;
		\draw [color={rgb, 255:red, 74; green, 144; blue, 226 }  ,draw opacity=1 ]   (275.81,117.11) .. controls (263.23,112.66) and (256.56,174.32) .. (275.81,177.91) ;
		\draw [color={rgb, 255:red, 126; green, 211; blue, 33 }  ,draw opacity=1 ]   (275.81,117.11) .. controls (291.4,116.32) and (292.01,175.66) .. (275.81,177.91) ;
		\draw  [color={rgb, 255:red, 74; green, 144; blue, 226 }  ,draw opacity=1 ] (223.92,135.38) -- (218.06,143.43) -- (214.83,134.01) ;
		
		\draw (322.3,123.3) node [anchor=north west][inner sep=0.75pt]  [xscale=0.9,yscale=0.9]  {$g'$};
		\draw (369.37,138.47) node [anchor=north west][inner sep=0.75pt]  [xscale=0.9,yscale=0.9]  {$\mathbb{T}^{2} \setminus \mathbb{D}^{2}$};

	\end{tikzpicture}
\end{center}
In that case a relative gluing structure for $g'f$ can be pictorially represented as:
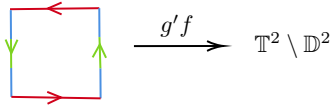
\begin{center}
	\tikzset{every picture/.style={line width=0.75pt}} 
	
	\begin{tikzpicture}[x=0.75pt,y=0.75pt,yscale=-0.75,xscale=0.75]
		
		\draw [color={rgb, 255:red, 74; green, 144; blue, 226 }  ,draw opacity=1 ]   (190.45,160.21) -- (190.58,179.67) ;
		\draw [color={rgb, 255:red, 74; green, 144; blue, 226 }  ,draw opacity=1 ]   (250.45,160.43) -- (250.58,180.67) ;
		\draw  [color={rgb, 255:red, 126; green, 211; blue, 33 }  ,draw opacity=1 ] (245.24,153.06) -- (250.05,144.39) -- (254.24,153.38) ;
		\draw  [color={rgb, 255:red, 126; green, 211; blue, 33 }  ,draw opacity=1 ] (194.69,141.23) -- (190.53,150.33) -- (186.69,141.09) ;
		\draw [color={rgb, 255:red, 208; green, 2; blue, 27 }  ,draw opacity=1 ]   (189.67,120.67) -- (249.68,119.98) ;
		\draw [color={rgb, 255:red, 208; green, 2; blue, 27 }  ,draw opacity=1 ]   (190.58,179.67) -- (250.58,180.67) ;
		\draw  [color={rgb, 255:red, 208; green, 2; blue, 27 }  ,draw opacity=1 ] (224.35,124.17) -- (215.18,120.17) -- (224.35,116.17) ;
		\draw  [color={rgb, 255:red, 208; green, 2; blue, 27 }  ,draw opacity=1 ] (214.68,176.29) -- (223.95,180.04) -- (214.89,184.29) ;
		\draw    (273,145.5) -- (331.1,145.42) ;
		\draw [shift={(333.1,145.41)}, rotate = 539.9200000000001] [color={rgb, 255:red, 0; green, 0; blue, 0 }  ][line width=0.75]    (10.93,-3.29) .. controls (6.95,-1.4) and (3.31,-0.3) .. (0,0) .. controls (3.31,0.3) and (6.95,1.4) .. (10.93,3.29)   ;
		\draw [color={rgb, 255:red, 74; green, 144; blue, 226 }  ,draw opacity=1 ]   (190.19,121.3) -- (190.32,140.76) ;
		\draw [color={rgb, 255:red, 74; green, 144; blue, 226 }  ,draw opacity=1 ]   (249.68,119.98) -- (249.81,139.44) ;
		\draw [color={rgb, 255:red, 126; green, 211; blue, 33 }  ,draw opacity=1 ]   (190.32,140.76) -- (190.45,160.21) ;
		\draw [color={rgb, 255:red, 126; green, 211; blue, 33 }  ,draw opacity=1 ]   (249.81,139.44) -- (250.45,160.43) ;
		
		\draw (289.1,122.3) node [anchor=north west][inner sep=0.75pt]  [xscale=0.9,yscale=0.9]  {$g'f$};
		\draw (352.37,134.47) node [anchor=north west][inner sep=0.75pt]  [xscale=0.9,yscale=0.9]  {$\mathbb{T}^{2} \setminus \mathbb{D}^{2}$};
		
	\end{tikzpicture}
\end{center}
where we identify red and green intervals. The axioms in Definition~\ref{AFdefinitionClassic} do not directly involve composition of relative gluings. We propose to address the issues concerning compositions of relative gluing diagrams in their full generality in subsequent work.


\subsection{Slice regions}
\label{sliceregionssubsection}

We introduce slice regions as the final ingredient of our gluing formalism. Slice regions are meant to behave as units with respect to the relative gluing operation. In GBQFT they provide the key ingredient in order to relate inner products with amplitudes via Axiom~\textbf{(T3x)} of Definition~\ref{AFdefinitionClassic}. Slice regions are, in our setting, analogous to identity cobordisms in different formulations of topological quantum field theory or to unit cylinders in Walker and Morrison's near-disk categories \cite{MoWa:blobhomology}. We use the relative gluing diagrams of Section~\ref{relativegluingsubsection} to define slice regions through their desired properties. We provide specific examples and construct models for slice regions.

\begin{definition}
	\label{sliceregionsdefinition}
	Let $\Sigma$ be a hypersurface. A slice region on $\Sigma$ is a pair $(S,s)$ where $S$ is a region and $s:\overline{\Sigma}\sqcup\Sigma\to \partial S$ is a gluing function, satisfying the following gluing triviality condition: For every region $X$, hypersurface $\Gamma$, and gluing function $\beta:\Sigma\sqcup \Gamma\to\partial X$ there exists a relative gluing diagram:
	\begin{center}
		\begin{tikzpicture}
			\matrix (m) [matrix of math nodes,row sep=3em,column sep=3em,minimum width=2em]
			{\Sigma\sqcup \Gamma&&&\partial X \\
				(\Sigma\sqcup \overline{\Sigma})\sqcup(\Sigma\sqcup \Gamma)&\partial(S\sqcup X)&S\sqcup X&X\\};
			\path[-stealth]
			(m-1-1) edge node [above] {$\beta$} (m-1-4)
			edge [bend left = 45]node [left]  {} (m-2-1)
			edge [bend right = 45] node {} (m-2-1)
			
			(m-2-1) edge node [below] {$s|_{\overline{\Sigma}}\sqcup \beta|_\Sigma$} (m-2-2)
			(m-2-2) edge node [below]  {} (m-2-3)
			(m-2-3) edge node [below] {$f$} (m-2-4)
			(m-1-4) edge node [right] {} (m-2-4)
			
			;
		\end{tikzpicture}
	\end{center}
	where the bent arrows represent the disjoint union of the inclusion of $\Sigma$ in $\Sigma\sqcup\overline{\Sigma}$ and the inclusion of $\Gamma$ in $\Sigma\sqcup \Gamma$ respectively.
\end{definition}

\begin{rem}
Let $\Sigma$ be a hypersurface. Let $(S,s)$ be a slice region on $\Sigma$. The gluing function $s$ provides a decomposition of $\partial S$ as a gluing of two copies of $\Sigma$ with opposite orientation. We call $s$ the \textbf{slice boundary decomposition} of $S$. The \textbf{gluing triviality} condition of Definition~\ref{sliceregionsdefinition} says that given any region $X$ such that $\partial X$ decomposes as $\Sigma$ glued to some other hypersurface $\Gamma$, we can glue $S$ to $X$ along the common copies of $\Sigma$ in $\partial S$ and $\partial X$, with opposite orientations, in such a way that the glued region is homeomorphic to $X$ in a canonical way. We thus consider slice regions as units with respect to the operation of relative gluing.
\end{rem}

We present specific examples clarifying Definition~\ref{sliceregionsdefinition}. We do this in the case in which $\Sigma$ is closed, i.e., in the case in which $s$ is trivial, and in the case in which $\partial \Sigma\neq\emptyset$ separately.

\begin{ex}\label{exsliceregionsclosed}
In the notation of Definition~\ref{sliceregionsdefinition} let $\Sigma$ be the circle $S^1$ with its usual orientation. Let $S$ be the cylinder $\Sigma\times [0,1]$. Let $s:\Sigma\sqcup \overline{\Sigma}\to\partial S$ be the disjoint union of the inclusion of $\Sigma$ in $S$ as $\Sigma\times\left\{0\right\}$ and as $\Sigma\times \left\{1\right\}$ respectively. Pictorially:
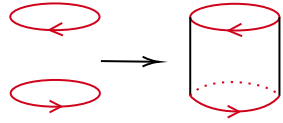
\begin{center}
\tikzset{every picture/.style={line width=0.75pt}} 
\begin{tikzpicture}[x=0.75pt,y=0.75pt,yscale=-0.75,xscale=0.75]

			\draw  [color={rgb, 255:red, 208; green, 2; blue, 27 }  ,draw opacity=1 ] (290.8,80.2) .. controls (290.8,75.23) and (304.23,71.2) .. (320.8,71.2) .. controls (337.37,71.2) and (350.8,75.23) .. (350.8,80.2) .. controls (350.8,85.17) and (337.37,89.2) .. (320.8,89.2) .. controls (304.23,89.2) and (290.8,85.17) .. (290.8,80.2) -- cycle ;
			\draw    (290.8,80.2) -- (290.8,133.15) ;
			\draw    (350.8,80.2) -- (350.8,133.15) ;
			\draw [color={rgb, 255:red, 208; green, 2; blue, 27 }  ,draw opacity=1 ]   (290.8,133.15) .. controls (303.65,146.95) and (340.05,147.35) .. (350.8,133.15) ;
			\draw [color={rgb, 255:red, 208; green, 2; blue, 27 }  ,draw opacity=1 ] [dash pattern={on 0.84pt off 2.51pt}]  (290.8,133.15) .. controls (301.25,118.69) and (342.85,123.09) .. (350.8,133.15) ;
			\draw  [color={rgb, 255:red, 208; green, 2; blue, 27 }  ,draw opacity=1 ] (315.96,139.84) -- (323.69,143.82) -- (315.96,147.79) ;
			\draw  [color={rgb, 255:red, 208; green, 2; blue, 27 }  ,draw opacity=1 ] (325.7,93) -- (316.59,89.29) -- (325.23,84.61) ;
			\draw  [color={rgb, 255:red, 208; green, 2; blue, 27 }  ,draw opacity=1 ] (169.8,79.95) .. controls (169.8,74.98) and (183.23,70.95) .. (199.8,70.95) .. controls (216.37,70.95) and (229.8,74.98) .. (229.8,79.95) .. controls (229.8,84.92) and (216.37,88.95) .. (199.8,88.95) .. controls (183.23,88.95) and (169.8,84.92) .. (169.8,79.95) -- cycle ;
			\draw  [color={rgb, 255:red, 208; green, 2; blue, 27 }  ,draw opacity=1 ] (170.05,131.45) .. controls (170.05,126.48) and (183.48,122.45) .. (200.05,122.45) .. controls (216.62,122.45) and (230.05,126.48) .. (230.05,131.45) .. controls (230.05,136.42) and (216.62,140.45) .. (200.05,140.45) .. controls (183.48,140.45) and (170.05,136.42) .. (170.05,131.45) -- cycle ;
			\draw  [color={rgb, 255:red, 208; green, 2; blue, 27 }  ,draw opacity=1 ] (196.21,136.34) -- (203.94,140.32) -- (196.21,144.29) ;
			\draw  [color={rgb, 255:red, 208; green, 2; blue, 27 }  ,draw opacity=1 ] (205.2,92.5) -- (196.09,88.79) -- (204.73,84.11) ;
			\draw    (230.75,109.75) -- (267.69,110.22) ;
			\draw [shift={(269.69,110.25)}, rotate = 180.74] [color={rgb, 255:red, 0; green, 0; blue, 0 }  ][line width=0.75]    (10.93,-3.29) .. controls (6.95,-1.4) and (3.31,-0.3) .. (0,0) .. controls (3.31,0.3) and (6.95,1.4) .. (10.93,3.29)   ;
			
	\end{tikzpicture}
\end{center}
Thus defined, $S$ is a slice region on $\Sigma$ with slice boundary decomposition $s$. We illustrate the pieces of data involved in the gluing triviality condition for $S$ with a specific example. Let $X$ be a three-punctured torus. We orient $X$ in such a way that $X$ admits a pictorial representation as:
	\begin{center}
		\tikzset{every picture/.style={line width=0.75pt}} 
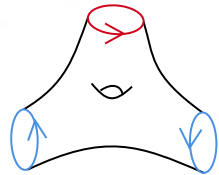
		\begin{tikzpicture}[x=0.75pt,y=0.75pt,yscale=-0.75,xscale=0.75]

			\draw  [color={rgb, 255:red, 254; green, 254; blue, 254 }  ,draw opacity=1 ][line width=3] [line join = round][line cap = round] (311.5,79.17) .. controls (312.5,79.17) and (313.5,79.17) .. (314.5,79.17) ;
			\draw  [color={rgb, 255:red, 254; green, 254; blue, 254 }  ,draw opacity=1 ][line width=3] [line join = round][line cap = round] (311.5,80.17) .. controls (312.55,80.17) and (313.75,79.91) .. (314.5,79.17) ;
			\draw  [color={rgb, 255:red, 254; green, 254; blue, 254 }  ,draw opacity=1 ][line width=3] [line join = round][line cap = round] (314.5,72.17) .. controls (314.5,72.91) and (315.5,73.42) .. (315.5,74.17) ;
			\draw  [color={rgb, 255:red, 254; green, 254; blue, 254 }  ,draw opacity=1 ][line width=3] [line join = round][line cap = round] (314.5,72.17) .. controls (314.5,73.54) and (315.5,74.79) .. (315.5,76.17) ;
			\draw  [color={rgb, 255:red, 254; green, 254; blue, 254 }  ,draw opacity=1 ][line width=3] [line join = round][line cap = round] (315.5,70.17) .. controls (315.5,72.92) and (314.69,77.17) .. (317.5,77.17) ;
			\draw  [color={rgb, 255:red, 254; green, 254; blue, 254 }  ,draw opacity=1 ][line width=3] [line join = round][line cap = round] (310.5,80.17) .. controls (315.7,80.17) and (315.72,77.17) .. (319.5,77.17) ;
			\draw  [color={rgb, 255:red, 254; green, 254; blue, 254 }  ,draw opacity=1 ][line width=3] [line join = round][line cap = round] (312.5,79.17) .. controls (320.64,79.17) and (319.34,77.17) .. (312.5,77.17) ;
			\draw  [color={rgb, 255:red, 254; green, 254; blue, 254 }  ,draw opacity=1 ][line width=3] [line join = round][line cap = round] (311.5,81.17) .. controls (312.04,81.17) and (316.5,80.17) .. (316.5,80.17) .. controls (316.5,80.17) and (314.21,80.87) .. (313.5,80.17) .. controls (308.83,75.5) and (314.5,77.17) .. (314.5,77.17) .. controls (314.5,77.17) and (313.5,76.83) .. (313.5,80.17) ;
			\draw  [color={rgb, 255:red, 254; green, 254; blue, 254 }  ,draw opacity=1 ][line width=1.5] [line join = round][line cap = round] (261.08,94.42) .. controls (257.86,94.42) and (264.08,91.42) .. (264.08,91.42) .. controls (264.08,91.42) and (261.39,93.42) .. (260.08,93.42) ;
			\draw  [color={rgb, 255:red, 254; green, 254; blue, 254 }  ,draw opacity=1 ][line width=1.5] [line join = round][line cap = round] (259.08,96.42) .. controls (260.29,96.42) and (261.23,95.27) .. (262.08,94.42) .. controls (262.61,93.89) and (264.83,93.42) .. (264.08,93.42) .. controls (252.79,93.42) and (260.58,90.42) .. (265.08,90.42) .. controls (266.57,90.42) and (262.42,93.08) .. (261.08,92.42) .. controls (260.26,92) and (264.08,89.42) .. (264.08,89.42) .. controls (264.08,89.42) and (260.02,91.42) .. (260.08,91.42) .. controls (261.46,91.42) and (263.47,89.19) .. (264.08,90.42) .. controls (265.74,93.74) and (258.66,95.96) .. (258.08,96.42) .. controls (257.5,96.88) and (259.46,95.83) .. (260.08,95.42) .. controls (261.58,94.42) and (263.29,93.42) .. (265.08,93.42) .. controls (265.55,93.42) and (263.66,92.63) .. (264.08,92.42) .. controls (265.46,91.73) and (269.76,91.09) .. (271.08,91.42) .. controls (272.91,91.87) and (268.77,94.57) .. (267.08,95.42) .. controls (260.98,98.47) and (254.57,101.14) .. (254.08,100.42) .. controls (253.53,99.58) and (253.38,98.12) .. (254.08,97.42) .. controls (255.29,96.21) and (258.55,98.03) .. (259.08,96.42) .. controls (259.74,94.44) and (261.25,91.42) .. (264.08,91.42) ;
			\draw    (325.32,89.13) .. controls (315.66,115.91) and (310.84,133.66) .. (282.15,150.8) ;
			\draw    (284.65,191.41) .. controls (337.29,166.4) and (359.33,174.35) .. (403.35,193.5) ;
			\draw    (362.79,91.78) .. controls (369.71,115.35) and (366.75,127.14) .. (402.35,152.85) ;
			\draw  [color={rgb, 255:red, 74; green, 144; blue, 226 }  ,draw opacity=1 ] (282.16,150.8) .. controls (287.07,150.66) and (291.27,159.65) .. (291.54,170.87) .. controls (291.81,182.1) and (288.05,191.31) .. (283.14,191.46) .. controls (278.23,191.6) and (274.03,182.61) .. (273.76,171.39) .. controls (273.49,160.16) and (277.25,150.95) .. (282.16,150.8) -- cycle ;
			\draw  [color={rgb, 255:red, 74; green, 144; blue, 226 }  ,draw opacity=1 ] (402.36,152.85) .. controls (407.27,152.71) and (411.47,161.69) .. (411.74,172.92) .. controls (412.01,184.14) and (408.25,193.36) .. (403.34,193.5) .. controls (398.43,193.64) and (394.23,184.65) .. (393.96,173.43) .. controls (393.69,162.2) and (397.45,152.99) .. (402.36,152.85) -- cycle ;
			\draw  [color={rgb, 255:red, 208; green, 2; blue, 27 }  ,draw opacity=1 ] (325.32,89.13) .. controls (325.64,83.82) and (334.29,80.11) .. (344.64,80.84) .. controls (354.98,81.58) and (363.11,86.48) .. (362.79,91.79) .. controls (362.47,97.1) and (353.82,100.81) .. (343.48,100.07) .. controls (333.13,99.34) and (325,94.44) .. (325.32,89.13) -- cycle ;
			\draw    (327.8,133.3) .. controls (337.69,145.08) and (342.64,147.22) .. (355,135.44) ;
			\draw    (333.74,138.12) .. controls (342.14,132.22) and (347.58,135.44) .. (348.57,140.26) ;
			\draw  [color={rgb, 255:red, 208; green, 2; blue, 27 }  ,draw opacity=1 ] (336.88,92.97) -- (349.16,99.71) -- (336.88,106.46) ;
			\draw  [color={rgb, 255:red, 74; green, 144; blue, 226 }  ,draw opacity=1 ] (401.78,164.37) -- (393.68,176.65) -- (387.13,163.33) ;
			\draw  [color={rgb, 255:red, 74; green, 144; blue, 226 }  ,draw opacity=1 ] (285.11,171.76) -- (289.82,158.33) -- (297.47,170.06) ;
		\end{tikzpicture}
	\end{center}
	We let $\Gamma$ be the blue component of $\partial X$, and we take $\eta:\Sigma\sqcup \Gamma\to \partial X$ to be the obvious inclusion. Gluing triviality on $(X,\eta)$ is provided by the function $f$ from the left-hand side to the right-hand side of the following picture:
	\begin{center}
		\tikzset{every picture/.style={line width=0.75pt}} 
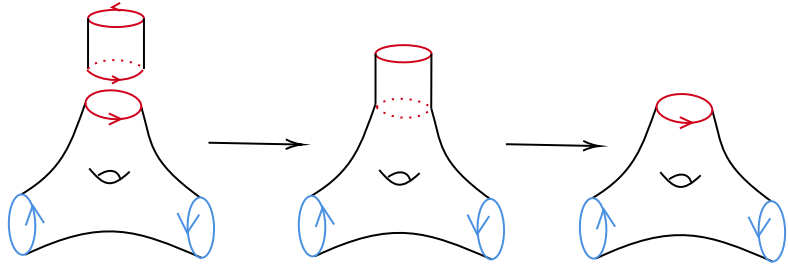
		\begin{tikzpicture}[x=0.75pt,y=0.75pt,yscale=-0.75,xscale=0.75]
			
			\draw    (208.32,145.13) .. controls (198.66,171.91) and (193.84,189.66) .. (165.15,206.8) ;
			\draw    (167.65,247.41) .. controls (220.29,222.4) and (242.33,230.35) .. (286.35,249.5) ;
			\draw    (245.79,147.78) .. controls (252.71,171.35) and (249.75,183.14) .. (285.35,208.85) ;
			\draw  [color={rgb, 255:red, 74; green, 144; blue, 226 }  ,draw opacity=1 ] (165.16,206.8) .. controls (170.07,206.66) and (174.27,215.65) .. (174.54,226.87) .. controls (174.81,238.1) and (171.05,247.31) .. (166.14,247.46) .. controls (161.23,247.6) and (157.03,238.61) .. (156.76,227.39) .. controls (156.49,216.16) and (160.25,206.95) .. (165.16,206.8) -- cycle ;
			\draw  [color={rgb, 255:red, 74; green, 144; blue, 226 }  ,draw opacity=1 ] (285.36,208.85) .. controls (290.27,208.71) and (294.47,217.69) .. (294.74,228.92) .. controls (295.01,240.14) and (291.25,249.36) .. (286.34,249.5) .. controls (281.43,249.64) and (277.23,240.65) .. (276.96,229.43) .. controls (276.69,218.2) and (280.45,208.99) .. (285.36,208.85) -- cycle ;
			\draw  [color={rgb, 255:red, 208; green, 2; blue, 27 }  ,draw opacity=1 ] (208.32,145.13) .. controls (208.64,139.82) and (217.29,136.11) .. (227.64,136.84) .. controls (237.98,137.58) and (246.11,142.48) .. (245.79,147.79) .. controls (245.47,153.1) and (236.82,156.81) .. (226.48,156.07) .. controls (216.13,155.34) and (208,150.44) .. (208.32,145.13) -- cycle ;
			\draw    (210.8,189.3) .. controls (220.69,201.08) and (225.64,203.22) .. (238,191.44) ;
			\draw    (216.74,194.12) .. controls (225.14,188.22) and (230.58,191.44) .. (231.57,196.26) ;
			\draw  [color={rgb, 255:red, 208; green, 2; blue, 27 }  ,draw opacity=1 ] (223.95,152.25) -- (232.16,156.05) -- (223.95,159.85) ;
			\draw  [color={rgb, 255:red, 74; green, 144; blue, 226 }  ,draw opacity=1 ] (284.78,220.37) -- (276.68,232.65) -- (270.13,219.33) ;
			\draw  [color={rgb, 255:red, 74; green, 144; blue, 226 }  ,draw opacity=1 ] (168.11,227.76) -- (172.82,214.33) -- (180.47,226.06) ;
			\draw  [color={rgb, 255:red, 208; green, 2; blue, 27 }  ,draw opacity=1 ] (210.01,88.36) .. controls (210.01,85.17) and (218.43,82.58) .. (228.81,82.58) .. controls (239.2,82.58) and (247.61,85.17) .. (247.61,88.36) .. controls (247.61,91.56) and (239.2,94.14) .. (228.81,94.14) .. controls (218.43,94.14) and (210.01,91.56) .. (210.01,88.36) -- cycle ;
			\draw    (210.01,88.36) -- (210.01,122.37) ;
			\draw    (247.61,88.36) -- (247.61,122.37) ;
			\draw [color={rgb, 255:red, 208; green, 2; blue, 27 }  ,draw opacity=1 ]   (209.35,123.04) .. controls (217.4,131.91) and (240.21,132.16) .. (246.95,123.04) ;
			\draw [color={rgb, 255:red, 208; green, 2; blue, 27 }  ,draw opacity=1 ] [dash pattern={on 0.84pt off 2.51pt}]  (210.01,122.37) .. controls (216.56,113.09) and (242.63,115.92) .. (247.61,122.37) ;
			\draw  [color={rgb, 255:red, 208; green, 2; blue, 27 }  ,draw opacity=1 ] (226.11,127.09) -- (230.96,129.64) -- (226.11,132.19) ;
			\draw  [color={rgb, 255:red, 208; green, 2; blue, 27 }  ,draw opacity=1 ] (231.88,83.28) -- (226.17,80.9) -- (231.59,77.89) ;
			\draw    (403.32,146.13) .. controls (393.66,172.91) and (388.84,190.66) .. (360.15,207.8) ;
			\draw    (362.65,248.41) .. controls (415.29,223.4) and (437.33,231.35) .. (481.35,250.5) ;
			\draw    (440.79,148.78) .. controls (447.71,172.35) and (444.75,184.14) .. (480.35,209.85) ;
			\draw  [color={rgb, 255:red, 74; green, 144; blue, 226 }  ,draw opacity=1 ] (360.16,207.8) .. controls (365.07,207.66) and (369.27,216.65) .. (369.54,227.87) .. controls (369.81,239.1) and (366.05,248.31) .. (361.14,248.46) .. controls (356.23,248.6) and (352.03,239.61) .. (351.76,228.39) .. controls (351.49,217.16) and (355.25,207.95) .. (360.16,207.8) -- cycle ;
			\draw  [color={rgb, 255:red, 74; green, 144; blue, 226 }  ,draw opacity=1 ] (480.36,209.85) .. controls (485.27,209.71) and (489.47,218.69) .. (489.74,229.92) .. controls (490.01,241.14) and (486.25,250.36) .. (481.34,250.5) .. controls (476.43,250.64) and (472.23,241.65) .. (471.96,230.43) .. controls (471.69,219.2) and (475.45,209.99) .. (480.36,209.85) -- cycle ;
			\draw    (405.8,190.3) .. controls (415.69,202.08) and (420.64,204.22) .. (433,192.44) ;
			\draw    (411.74,195.12) .. controls (420.14,189.22) and (425.58,192.44) .. (426.57,197.26) ;
			\draw  [color={rgb, 255:red, 74; green, 144; blue, 226 }  ,draw opacity=1 ] (479.78,221.37) -- (471.68,233.65) -- (465.13,220.33) ;
			\draw  [color={rgb, 255:red, 74; green, 144; blue, 226 }  ,draw opacity=1 ] (363.11,228.76) -- (367.82,215.33) -- (375.47,227.06) ;
			\draw    (291,172) -- (351.88,173.21) ;
			\draw [shift={(353.88,173.25)}, rotate = 181.14] [color={rgb, 255:red, 0; green, 0; blue, 0 }  ][line width=0.75]    (10.93,-3.29) .. controls (6.95,-1.4) and (3.31,-0.3) .. (0,0) .. controls (3.31,0.3) and (6.95,1.4) .. (10.93,3.29)   ;
			\draw  [color={rgb, 255:red, 208; green, 2; blue, 27 }  ,draw opacity=1 ] (403.32,112.12) .. controls (403.32,108.93) and (411.74,106.34) .. (422.12,106.34) .. controls (432.51,106.34) and (440.92,108.93) .. (440.92,112.12) .. controls (440.92,115.31) and (432.51,117.9) .. (422.12,117.9) .. controls (411.74,117.9) and (403.32,115.31) .. (403.32,112.12) -- cycle ;
			\draw    (403.32,112.12) -- (403.32,146.13) ;
			\draw    (440.92,112.12) -- (440.79,148.78) ;
			\draw [color={rgb, 255:red, 208; green, 2; blue, 27 }  ,draw opacity=1 ] [dash pattern={on 0.84pt off 2.51pt}]  (403.95,148.05) .. controls (412,156.92) and (434.05,157.91) .. (440.79,148.78) ;
			\draw [color={rgb, 255:red, 208; green, 2; blue, 27 }  ,draw opacity=1 ] [dash pattern={on 0.84pt off 2.51pt}]  (403.95,148.05) .. controls (410.5,138.77) and (435.81,142.33) .. (440.79,148.78) ;
			\draw    (592.32,147.63) .. controls (582.66,174.41) and (577.84,192.16) .. (549.15,209.3) ;
			\draw    (551.65,249.91) .. controls (604.29,224.9) and (626.33,232.85) .. (670.35,252) ;
			\draw    (629.79,150.28) .. controls (636.71,173.85) and (633.75,185.64) .. (669.35,211.35) ;
			\draw  [color={rgb, 255:red, 74; green, 144; blue, 226 }  ,draw opacity=1 ] (549.16,209.3) .. controls (554.07,209.16) and (558.27,218.15) .. (558.54,229.37) .. controls (558.81,240.6) and (555.05,249.81) .. (550.14,249.96) .. controls (545.23,250.1) and (541.03,241.11) .. (540.76,229.89) .. controls (540.49,218.66) and (544.25,209.45) .. (549.16,209.3) -- cycle ;
			\draw  [color={rgb, 255:red, 74; green, 144; blue, 226 }  ,draw opacity=1 ] (669.36,211.35) .. controls (674.27,211.21) and (678.47,220.19) .. (678.74,231.42) .. controls (679.01,242.64) and (675.25,251.86) .. (670.34,252) .. controls (665.43,252.14) and (661.23,243.15) .. (660.96,231.93) .. controls (660.69,220.7) and (664.45,211.49) .. (669.36,211.35) -- cycle ;
			\draw  [color={rgb, 255:red, 208; green, 2; blue, 27 }  ,draw opacity=1 ] (592.32,147.63) .. controls (592.64,142.32) and (601.29,138.61) .. (611.64,139.34) .. controls (621.98,140.08) and (630.11,144.98) .. (629.79,150.29) .. controls (629.47,155.6) and (620.82,159.31) .. (610.48,158.57) .. controls (600.13,157.84) and (592,152.94) .. (592.32,147.63) -- cycle ;
			\draw    (594.8,191.8) .. controls (604.69,203.58) and (609.64,205.72) .. (622,193.94) ;
			\draw    (600.74,196.62) .. controls (609.14,190.72) and (614.58,193.94) .. (615.57,198.76) ;
			\draw  [color={rgb, 255:red, 208; green, 2; blue, 27 }  ,draw opacity=1 ] (607.95,154.75) -- (616.16,158.55) -- (607.95,162.35) ;
			\draw  [color={rgb, 255:red, 74; green, 144; blue, 226 }  ,draw opacity=1 ] (668.78,222.87) -- (660.68,235.15) -- (654.13,221.83) ;
			\draw  [color={rgb, 255:red, 74; green, 144; blue, 226 }  ,draw opacity=1 ] (552.11,230.26) -- (556.82,216.83) -- (564.47,228.56) ;
			\draw    (491,173) -- (551.88,174.21) ;
			\draw [shift={(553.88,174.25)}, rotate = 181.14] [color={rgb, 255:red, 0; green, 0; blue, 0 }  ][line width=0.75]    (10.93,-3.29) .. controls (6.95,-1.4) and (3.31,-0.3) .. (0,0) .. controls (3.31,0.3) and (6.95,1.4) .. (10.93,3.29)   ;
		\end{tikzpicture}
	\end{center}
	The upper left corner of the diagram in Definition~\ref{sliceregionsdefinition} in this case is equal to the disjoint union of the top closed string of $S$ and the closed strings on $\partial X$ (blue).
\end{ex}

\begin{rem}
\label{remarkclosedsliceregions}
The construction presented in Example~\ref{exsliceregionsclosed} is generic, in the sense that for every closed hypersurface $\Sigma$, the cylinder $\Sigma\times [0,1]$, with boundary decomposition $\Sigma\sqcup\overline{\Sigma}\to\partial (\Sigma\times [0,1])$ provided by the inclusions of $\Sigma$ and $\overline{\Sigma}$ in $\Sigma\times [0,1]$ as $\Sigma\times\left\{1\right\}$ and $\Sigma\times \left\{0\right\}$ is a slice region on $\Sigma$. Every closed hypersurface thus admits a slice region.
\end{rem}

We now illustrate the pieces of data involved in Definition~\ref{sliceregionsdefinition} for hypersurfaces with non-empty boundary with a specific example.

\begin{ex}
	\label{exslicenonclosed}
	In the notation of Definition~\ref{sliceregionsdefinition} let $\Sigma$ be a closed string. Let $S$ be the pinched cylinder on $\Sigma$, i.e., let $S$ be the surface obtained from $\Sigma\times [0,1]$ by identifying the boundary component $\partial \Sigma\times [0,1]$ with $\partial\Sigma$. Let $s:\Sigma\sqcup\overline{\Sigma}\to \partial S$ be the composition of the inclusion of $\Sigma$ and $\overline{\Sigma}$ in $\Sigma\times [0,1]$ analogous to Example~\ref{exsliceregionsclosed} and the projection of $\Sigma\times [0,1]$ onto $S$. We picture $(S,s)$ as:
	\begin{center}
		\tikzset{every picture/.style={line width=0.75pt}} 
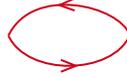
		\begin{tikzpicture}[x=0.75pt,y=0.75pt,yscale=-0.75,xscale=0.75]
			
			\draw [color={rgb, 255:red, 208; green, 2; blue, 27 }  ,draw opacity=1 ]   (258,118.92) .. controls (270.67,146.58) and (328,145.58) .. (339.58,118.58) ;
			\draw [color={rgb, 255:red, 208; green, 2; blue, 27 }  ,draw opacity=1 ]   (258,118.92) .. controls (274.67,84.58) and (332.67,98.92) .. (339.58,118.58) ;
			\draw  [color={rgb, 255:red, 208; green, 2; blue, 27 }  ,draw opacity=1 ] (293.25,135.25) -- (303.33,139.25) -- (293.25,143.25) ;
			\draw  [color={rgb, 255:red, 208; green, 2; blue, 27 }  ,draw opacity=1 ] (302.48,102.06) -- (292.25,98.44) -- (302.18,94.06) ;
		\end{tikzpicture}
	\end{center}
The pair $(S,s)$ is a slice region on $\Sigma$. We illustrate the pieces of data involved in the gluing triviality condition for $S$ with a specific example. Let $X$ be the cylinder $S^1\times [0,1]$, which for our convenience we picture as:
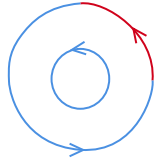
\begin{center}

		\tikzset{every picture/.style={line width=0.75pt}} 
		
		\begin{tikzpicture}[x=0.75pt,y=0.75pt,yscale=-0.75,xscale=0.75]
			
			\draw [color={rgb, 255:red, 208; green, 2; blue, 27 }  ,draw opacity=1 ]   (375.25,104.63) .. controls (377.75,73.25) and (343.75,50.63) .. (326.75,52.75) ;
			\draw [color={rgb, 255:red, 74; green, 144; blue, 226 }  ,draw opacity=1 ]   (326.75,52.75) .. controls (304.75,54.44) and (294,64.53) .. (286.06,75.58) .. controls (278.13,86.63) and (278.38,95.38) .. (278.75,105.25) .. controls (279.13,115.13) and (288.25,148.38) .. (326.75,151.25) .. controls (365.25,154.13) and (376.75,112.13) .. (375.25,104.63) ;
			\draw  [color={rgb, 255:red, 74; green, 144; blue, 226 }  ,draw opacity=1 ] (307.35,103.5) .. controls (307.35,92.37) and (316.04,83.35) .. (326.75,83.35) .. controls (337.46,83.35) and (346.15,92.37) .. (346.15,103.5) .. controls (346.15,114.63) and (337.46,123.65) .. (326.75,123.65) .. controls (316.04,123.65) and (307.35,114.63) .. (307.35,103.5) -- cycle ;
			\draw  [color={rgb, 255:red, 208; green, 2; blue, 27 }  ,draw opacity=1 ] (364.33,79.57) -- (362.36,69.7) -- (371.19,74.54) ;
			\draw  [color={rgb, 255:red, 74; green, 144; blue, 226 }  ,draw opacity=1 ] (319.5,147.13) -- (328.63,151.38) -- (319.5,155.63) ;
			\draw  [color={rgb, 255:red, 74; green, 144; blue, 226 }  ,draw opacity=1 ] (330.63,87) -- (321.04,83.96) -- (329.54,78.57) ;
			
		\end{tikzpicture}
	\end{center}
We consider the boundary decomposition of $X$ described by red and blue components as above. In that case $\Gamma$ is the disjoint union of the inner closed string pictured in blue and the open string boundary component pictured in blue, and $\eta:\Gamma\sqcup\Sigma\to \partial X$ is described above. A gluing trivialization diagram for $S$ and $X$ is described by the picture:
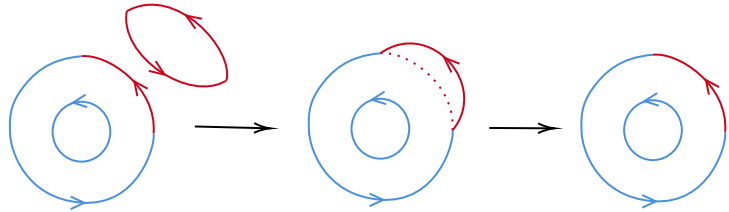
\begin{center}
		\tikzset{every picture/.style={line width=0.75pt}} 
		\begin{tikzpicture}[x=0.75pt,y=0.75pt,yscale=-0.75,xscale=0.75]
			
			\draw [color={rgb, 255:red, 208; green, 2; blue, 27 }  ,draw opacity=1 ]   (348.75,106.75) .. controls (382.25,83) and (422.38,128.25) .. (397.25,158.63) ;
			\draw  [color={rgb, 255:red, 208; green, 2; blue, 27 }  ,draw opacity=1 ] (396.51,119.58) -- (391.84,109.78) -- (401.95,113.71) ;
			\draw [color={rgb, 255:red, 208; green, 2; blue, 27 }  ,draw opacity=1 ] [dash pattern={on 0.84pt off 2.51pt}]  (397.25,158.63) .. controls (399.75,127.25) and (365.75,104.63) .. (348.75,106.75) ;
			\draw [color={rgb, 255:red, 74; green, 144; blue, 226 }  ,draw opacity=1 ]   (348.75,106.75) .. controls (326.75,108.44) and (316,118.53) .. (308.06,129.58) .. controls (300.13,140.63) and (300.38,149.38) .. (300.75,159.25) .. controls (301.13,169.13) and (310.25,202.38) .. (348.75,205.25) .. controls (387.25,208.13) and (398.75,166.13) .. (397.25,158.63) ;
			\draw  [color={rgb, 255:red, 74; green, 144; blue, 226 }  ,draw opacity=1 ] (329.35,157.5) .. controls (329.35,146.37) and (338.04,137.35) .. (348.75,137.35) .. controls (359.46,137.35) and (368.15,146.37) .. (368.15,157.5) .. controls (368.15,168.63) and (359.46,177.65) .. (348.75,177.65) .. controls (338.04,177.65) and (329.35,168.63) .. (329.35,157.5) -- cycle ;
			\draw  [color={rgb, 255:red, 74; green, 144; blue, 226 }  ,draw opacity=1 ] (341.5,201.13) -- (350.63,205.38) -- (341.5,209.63) ;
			\draw  [color={rgb, 255:red, 74; green, 144; blue, 226 }  ,draw opacity=1 ] (352.63,141) -- (343.04,137.96) -- (351.54,132.57) ;
			\draw [color={rgb, 255:red, 208; green, 2; blue, 27 }  ,draw opacity=1 ]   (580.58,159.63) .. controls (583.08,128.25) and (549.08,105.63) .. (532.08,107.75) ;
			\draw [color={rgb, 255:red, 74; green, 144; blue, 226 }  ,draw opacity=1 ]   (532.08,107.75) .. controls (510.08,109.44) and (499.33,119.53) .. (491.4,130.58) .. controls (483.46,141.63) and (483.71,150.38) .. (484.08,160.25) .. controls (484.46,170.13) and (493.58,203.38) .. (532.08,206.25) .. controls (570.58,209.13) and (582.08,167.13) .. (580.58,159.63) ;
			\draw  [color={rgb, 255:red, 74; green, 144; blue, 226 }  ,draw opacity=1 ] (512.68,158.5) .. controls (512.68,147.37) and (521.37,138.35) .. (532.08,138.35) .. controls (542.8,138.35) and (551.48,147.37) .. (551.48,158.5) .. controls (551.48,169.63) and (542.8,178.65) .. (532.08,178.65) .. controls (521.37,178.65) and (512.68,169.63) .. (512.68,158.5) -- cycle ;
			\draw  [color={rgb, 255:red, 208; green, 2; blue, 27 }  ,draw opacity=1 ] (569.67,134.57) -- (567.7,124.7) -- (576.52,129.54) ;
			\draw  [color={rgb, 255:red, 74; green, 144; blue, 226 }  ,draw opacity=1 ] (524.83,202.13) -- (533.96,206.38) -- (524.83,210.63) ;
			\draw  [color={rgb, 255:red, 74; green, 144; blue, 226 }  ,draw opacity=1 ] (535.97,142) -- (526.37,138.96) -- (534.87,133.57) ;
			\draw    (422,157) -- (463.75,157.24) ;
			\draw [shift={(465.75,157.25)}, rotate = 180.33] [color={rgb, 255:red, 0; green, 0; blue, 0 }  ][line width=0.75]    (10.93,-3.29) .. controls (6.95,-1.4) and (3.31,-0.3) .. (0,0) .. controls (3.31,0.3) and (6.95,1.4) .. (10.93,3.29)   ;
			\draw [color={rgb, 255:red, 208; green, 2; blue, 27 }  ,draw opacity=1 ]   (178.27,78.41) .. controls (172.72,108.33) and (220.2,140.48) .. (245.2,125.06) ;
			\draw [color={rgb, 255:red, 208; green, 2; blue, 27 }  ,draw opacity=1 ]   (178.27,78.41) .. controls (211.65,59.91) and (250.86,104.99) .. (245.2,125.06) ;
			\draw  [color={rgb, 255:red, 208; green, 2; blue, 27 }  ,draw opacity=1 ] (197.71,112.04) -- (203.66,121.11) -- (193.11,118.59) ;
			\draw  [color={rgb, 255:red, 208; green, 2; blue, 27 }  ,draw opacity=1 ] (224.36,90.2) -- (218.07,81.36) -- (228.71,83.48) ;
			\draw [color={rgb, 255:red, 208; green, 2; blue, 27 }  ,draw opacity=1 ]   (196.25,160.63) .. controls (198.75,129.25) and (164.75,106.63) .. (147.75,108.75) ;
			\draw [color={rgb, 255:red, 74; green, 144; blue, 226 }  ,draw opacity=1 ]   (147.75,108.75) .. controls (125.75,110.44) and (115,120.53) .. (107.06,131.58) .. controls (99.13,142.63) and (99.38,151.38) .. (99.75,161.25) .. controls (100.13,171.13) and (109.25,204.38) .. (147.75,207.25) .. controls (186.25,210.13) and (197.75,168.13) .. (196.25,160.63) ;
			\draw  [color={rgb, 255:red, 74; green, 144; blue, 226 }  ,draw opacity=1 ] (128.35,159.5) .. controls (128.35,148.37) and (137.04,139.35) .. (147.75,139.35) .. controls (158.46,139.35) and (167.15,148.37) .. (167.15,159.5) .. controls (167.15,170.63) and (158.46,179.65) .. (147.75,179.65) .. controls (137.04,179.65) and (128.35,170.63) .. (128.35,159.5) -- cycle ;
			\draw  [color={rgb, 255:red, 208; green, 2; blue, 27 }  ,draw opacity=1 ] (185.33,135.57) -- (183.36,125.7) -- (192.19,130.54) ;
			\draw  [color={rgb, 255:red, 74; green, 144; blue, 226 }  ,draw opacity=1 ] (140.5,203.13) -- (149.63,207.38) -- (140.5,211.63) ;
			\draw  [color={rgb, 255:red, 74; green, 144; blue, 226 }  ,draw opacity=1 ] (151.63,143) -- (142.04,139.96) -- (150.54,134.57) ;
			\draw    (224,156) -- (272.75,157.21) ;
			\draw [shift={(274.74,157.26)}, rotate = 181.42] [color={rgb, 255:red, 0; green, 0; blue, 0 }  ][line width=0.75]    (10.93,-3.29) .. controls (6.95,-1.4) and (3.31,-0.3) .. (0,0) .. controls (3.31,0.3) and (6.95,1.4) .. (10.93,3.29)   ;
		\end{tikzpicture}
	\end{center}
	In this case the upper left corner of the diagram in Definition~\ref{sliceregionsdefinition} is provided by the disjoint union of the upper copy of $\Sigma$ in $\partial S$ and the blue components of $\partial X$.  
\end{ex}

\begin{rem}
\label{remarksliceregions}
The construction presented in Example~\ref{exslicenonclosed} is again generic. Given a hypersurface $\Sigma$ such that $\partial\Sigma\neq\emptyset$, the pinched cylinder $S$ on $\Sigma$, see \cite{Walker}, is a slice region on $\Sigma$. Every non-closed hypersurface thus admits a slice region. From this and from Remark~\ref{remarkclosedsliceregions} it follows that every hypersurface admits a slice region.
\end{rem}

The following lemma proves that the collection of slice regions on hypersurfaces is closed under taking disjoint unions and reversing orientations. 

\begin{lem}\label{ClosureofSliceRegions}
Let $\Sigma$ be a hypersurface. Let $(S,s)$ be a slice region on $\Sigma$. $(S,s)$ is also a slice region on $\overline{\Sigma}$ and $(\overline{S},\overline{s})$ is also a slice region on $\Sigma$ and $\overline{\Sigma}$. Let $\Sigma_i$, $i=1,2$ be hypersurfaces. Let $(S_i,s_i)$ be a slice region for $\Sigma_i$, $i=1,2$. The pair $(S_1\sqcup S_2,s_1\sqcup s_2)$ is a slice region on $\Sigma_1\sqcup \Sigma_2$. 
\end{lem}

\begin{lem}\label{sliceregionshomeomorphic}
Let $\Sigma$ be a hypersurface. If $(S,s)$ and $(S',s')$ are slice regions on $\Sigma$ then $S$ and $S'$ are homeomorphic.
\end{lem}
\begin{proof}
	Let $\Sigma$ be a hypersurface. Let $(S,s)$ and $(S',s')$ be slice regions on $\Sigma$. If we make $X$ be equal to $S$ or $S'$ in Definition~\ref{sliceregionsdefinition} we obtain relative gluing diagrams $f:S\sqcup S'\to S$ and $f':S\sqcup S'\to S'$, by Definition~\ref{relativegluingdefinition} both $S$ and $S'$ are coequalizers of the diagram:
	\[\Sigma\rightrightarrows (\Sigma\sqcup\overline{\Sigma})\sqcup (\Sigma\sqcup \overline{\Sigma})\xrightarrow[]{s\sqcup s'} \partial S\sqcup \partial S'\hookrightarrow S\sqcup S'\]
	and are thus homeomorphic. 
\end{proof}

Remark~\ref{remarkclosedsliceregions}, Remark~\ref{remarksliceregions} and Lemma~\ref{sliceregionshomeomorphic} imply that every slice region $(S,s)$ on a hypersurface $\Sigma$ is homeomorphic to the standard cylinder $\Sigma\times [0,1]$ on $\Sigma$ in the case in which $\Sigma$ is closed and to the pinched cylinder on $\Sigma$ in the case in which $\partial\Sigma\neq\emptyset$. The reader might be wondering why we have chosen to define slice regions through the formal properties appearing in Definition~\ref{sliceregionsdefinition} and not directly as cylinders in the case in which the hypersurface $\Sigma$ is closed or as pinched cylinders in the case in which $\partial\Sigma\neq\emptyset$. We have done this in order to lay the formal foundations of slice regions in such a way that Definition~\ref{sliceregionsdefinition} may be adapted to contexts where spacetime systems carry more complicated structures, e.g., area functions \cite{RuSz:areaqft}, defects \cite{RunkelCarquevilleDefects}, metrics \cite{StolzTeichnerSusy}, etc. We urge the reader to think of slice regions as regions that "behave like topological cylinders" and that in the particular context of this paper happen to be topological cylinders. We end this section with the following observation.

\begin{obs}
	Every slice region can be interpreted as a relative gluing diagram in the sense of Definition~\ref{relativegluingdefinition}. To see this let $\Sigma$ be a hypersurface and let $(S,s)$ be a slice region on $\Sigma$. $S$ and $s$ can be identified with the following relative gluing diagram:
	\begin{center}
		\begin{tikzpicture}
			\matrix (m) [matrix of math nodes,row sep=3em,column sep=3em,minimum width=2em]
			{\Sigma\sqcup \overline{\Sigma}&&&\partial S \\
				\Sigma\sqcup \overline{\Sigma} & \partial S &S&S\\};
			\path[-stealth]
			(m-1-1) edge node [above] {$s$} (m-1-4)
			edge node [left]  {$id$} (m-2-1)
			
			(m-2-1) edge node [below] {$s$} (m-2-2)
			(m-2-2) edge node [below]  {} (m-2-3)
			(m-2-3) edge node [below] {$id$} (m-2-4)
			(m-1-4) edge node [right] {} (m-2-4)
			
			;
		\end{tikzpicture}
	\end{center}
\end{obs}


\section{Implementing gluing functions: CQFT}
\label{sec:axiomatics}

In this section we combine the gluing formalism introduced in Section~\ref{usualgluings} with GBQFT. The axiomatic system of Definition~\ref{AFdefinitionClassic} requires an underlying notion of gluing, but is rather agnostic about how such a notion is defined. Taking the notion of gluing to be given by the formalism of Section~\ref{usualgluings} allows us to extract a tighter axiomatic system that (as we shall see) is more amenable to categorification. The new axiomatic system, to be presented in this section, cannot be equivalent to the one of Definition~\ref{AFdefinitionClassic} as it has a different scope, presupposing the formalism of Section~\ref{usualgluings}. Moreover, it integrates notions of symmetry of hypersurfaces and regions inherent in our novel gluing formalism, thus exhibiting equivariance. To emphasize the key feature of categorical compositionality (which subsumes in our mind both locality and functoriality) we shall refer to the new axiomatic system as \emph{Compositional Quantum Field Theory (CQFT)}. At the same time, in laying down the new axiomatic system we take the opportunity to be more precise as compared to GBQFT, especially concerning coherence conditions. What is more, we present two versions of the axiomatic system, differing only in including or not one specific axiom, corresponding to Axiom~\textbf{(T2a)} of Definition~\ref{AFdefinitionClassic}. We also provide missing explanations of certain notions used in Definition~\ref{AFdefinitionClassic}, but omitted in the introduction for brevity.

\subsection{Axiomatic presentation}
\label{sec:axioms}

We begin with a few technical remarks on Definition~\ref{AFdefinitionClassic}. In Axiom~\textbf{(T2)} decompositions of hypersurfaces $\Sigma$ of the form $\Sigma=\Sigma_1\cup...\cup\Sigma_n$ represent inductive decompositions as unions of hypersurfaces $\Lambda\cup\Sigma_i$ where $\Lambda$ is the union $\Sigma_1\cup...\cup\Sigma_{i-1}$ and where the symbol $\cup$ is taken to mean either the disjoint union $\Lambda\sqcup\Sigma_i$ or a gluing $\Lambda\cup_{B}\Sigma_i$ along a common closed codimension 0 sub-hypersurface $B$ of $\partial\Lambda$ and $\partial\Sigma_i$. An example of this in dimension $d=2$ is the decomposition of
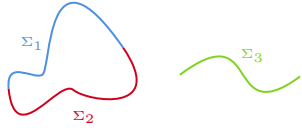
\begin{center}

\tikzset{every picture/.style={line width=0.75pt}} 
\begin{tikzpicture}[x=0.75pt,y=0.75pt,yscale=-1,xscale=1]
	
	\draw [color={rgb, 255:red, 126; green, 211; blue, 33 }  ,draw opacity=1 ]   (322.73,179.53) .. controls (341.79,165.24) and (348.06,169.88) .. (353,177.17) .. controls (357.94,184.46) and (362.39,196.04) .. (383.33,180.33) ;
	\draw [color={rgb, 255:red, 74; green, 144; blue, 226 }  ,draw opacity=1 ]   (236.33,186.56) .. controls (235,171.17) and (248.67,182.17) .. (253.33,179.5) .. controls (258,176.83) and (254.67,108.83) .. (294.11,165.87) ;
	\draw [color={rgb, 255:red, 208; green, 2; blue, 27 }  ,draw opacity=1 ]   (236.33,186.56) .. controls (239.67,218.5) and (261.67,181.17) .. (269,187.17) .. controls (276.33,193.17) and (316.67,199.17) .. (294.11,165.87) ;
	
	\draw (241,159.07) node [anchor=north west][inner sep=0.75pt]  [font=\tiny,color={rgb, 255:red, 74; green, 144; blue, 226 }  ,opacity=1 ]  {$\Sigma _{1}$};
	\draw (267,196.07) node [anchor=north west][inner sep=0.75pt]  [font=\tiny,color={rgb, 255:red, 208; green, 2; blue, 27 }  ,opacity=1 ]  {$\Sigma _{2}$};
	\draw (352,165.07) node [anchor=north west][inner sep=0.75pt]  [font=\tiny,color={rgb, 255:red, 126; green, 211; blue, 33 }  ,opacity=1 ]  {$\Sigma _{3}$};

\end{tikzpicture}
\end{center}
where all the hypersurfaces $\Sigma_i$ are open strings, and the union $\Sigma_1\cup\Sigma_2$ represents the closing $\Lambda$ of the blue and red open strings $\Sigma_1,\Sigma_2$ along their common boundary. The union $\Lambda\cup\Sigma_3$ represents the disjoint union of the closed string $\Sigma_1\cup\Sigma_2$ and the green open string $\Sigma_3$. The associativity condition on the morphism $\tau$ in Axiom~\textbf{(T2)} should be taken to mean that the value of $\tau$ does not depend on the order in which the inductive union procedure described above is performed. Precisely, the decomposition $\Sigma=\Sigma_1\cup\Sigma_2\cup\Sigma_3$ above can be interpreted as $\Sigma=(\Sigma_1\cup\Sigma_2)\cup\Sigma_3$ or as $\Sigma=\Sigma_1\cup(\Sigma_2\cup\Sigma_3)$. In both cases we obtain a composition of morphisms $\tau_{\Sigma_1\cup\Sigma_2,\Sigma_3}\circ\tau_{\Sigma_1,\Sigma_2}$ and $\tau_{\Sigma_1,\Sigma_2\cup\Sigma_3}\circ\tau_{\Sigma_2\cup\Sigma_3}$. Associativity requires that the equation $\tau_{\Sigma_1\cup\Sigma_2,\Sigma_3}\circ\tau_{\Sigma_1,\Sigma_2}=\tau_{\Sigma_1,\Sigma_2\cup\Sigma_3}\circ\tau_{\Sigma_2\cup\Sigma_3}$ should always hold, i.e., that the following diagram be commutative:
\begin{center}
\begin{tikzpicture}
  \matrix (m) [matrix of math nodes,row sep=4em,column sep=4em,minimum width=2em]
  {
     \cH_{\Sigma_1}\otimes\cH_{\Sigma_2}\otimes\cH_{\Sigma_3}&\cH_{\Sigma_1\sqcup\Sigma_2}\otimes\cH_{\Sigma_3} \\
		 \cH_{\Sigma_1}\otimes\cH_{\Sigma_2\sqcup\Sigma_3}&\cH_{\Sigma_1\sqcup\Sigma_2\sqcup\Sigma_3} \\};
  \path[-stealth]
    (m-1-1) edge node [left] {$id_{\cH_{\Sigma_1}}\otimes\tau$} (m-2-1)
            edge node [above] {$\tau\otimes id_{\cH_{\Sigma3}}$} (m-1-2)
    (m-1-2) edge node [right] {$\tau$} (m-2-2)
    (m-2-1) edge node [below] {$\tau$}(m-2-2);
\end{tikzpicture}
\end{center}
We sometimes omit the subscript in this notation and simply write $\tau$ for $\tau_{\Sigma_1,...,\Sigma_n}$. The above diagram will be part of the new axiomatic system to be presented.

The two interpretations of the union symbol $\cup$ in Axiom~\textbf{(T2)} are of a different nature: While the disjoint union $\Sigma\sqcup \Lambda$ of hypersurfaces is a parallel operation depending only on the components $\Sigma$ and $\Lambda$, the gluing $\Sigma\cup_B\Lambda$ of $\Sigma$ and $\Lambda$ along a common closed sub-hypersurface of $\partial\Sigma$ and $\partial\Lambda$ is a compositional operation, depending on extra data provided by $B$. Gluing is closer to cospans, see\cite{BaezCourser}, while disjoint union is closer to tensor products. It is thus natural to consider gluings and disjoint unions as separate associative operations as in Axioms~\textbf{(T5a)} and \textbf{(T5b)} above. Doing this requires separate versions of Axioms~\textbf{(T2)} and \textbf{(T2b)} providing isometric isomorphisms $\tau$ and  epimorphisms $\gamma$ associated to disjoint union and gluing respectively, satisfying associativity diagrams as described above, together with an extra axiom requiring compatibility of $\tau$ and $\gamma$ allowing us to combine the two operations. We do this in Definition~\ref{QFTAFfin}. Finally, when interpreting the symbol $\cup$ in expressions $\Sigma=\Sigma_1\cup...\cup\Sigma_n$ appearing in Axiom~\textbf{(T2)} as disjoint union $\sqcup$, it is implicitly assumed that the isomorphism $\tau$ commutes with any possible rearrangement of the factors appearing in the tensor product $\cH_{\Sigma_1}\otimes...\otimes \cH_{\Sigma_n}$. We again include this consideration in Definition~\ref{QFTAFfin}.

We divide the axioms of Definition~\ref{QFTAFfin} into axioms providing structure data, numbered \ref{Statespacesaxiom}--\ref{AmplitudeMapsAxiom} below, and axioms providing coherence equations for this data, numbered \ref{AssociativityTensorIsometryAxiom}--\ref{AmplitudeRelativeGluingAxiom} below.
We also include the coherence Axiom~\ref{TensorIsometryandGradedSymmetryAxiom} as an optional separate axiom. The numbering scheme and general structure of the new axiomatic system differs from that of Definition~\ref{AFdefinitionClassic}. In Section~\ref{sec:compareaxioms} below we explain the differences in detail. At first, we present the axioms for the case that all superspaces involved are finite-dimensional. In this case we call the CQFT \emph{rigid}.

\begin{definition}
\label{QFTAFfin}
A \emph{rigid CQFT} consists of the following data:
\begin{enumerate}
\item \label{Statespacesaxiom}\textbf{State spaces:}
There is a complex superspace $\cH_\Sigma$ for every hypersurface $\Sigma$. We call $\cH_\Sigma$ the \textit{state space} associated to $\Sigma$. The equation $\cH_\emptyset=\C$ holds.

\item \textbf{Graded involution:} \label{ConjugationIsometryAxiom}
There is an anti-linear graded isomorphism $\iota_\Sigma:\cH_\Sigma\to\cH_{\overline{\Sigma}}$ for every hypersurface $\Sigma$. The equation $\iota_{\overline{\Sigma}}\iota_\Sigma=id_{\cH_\Sigma}$ holds for every $\Sigma$. We call $\iota_\Sigma$ the \textit{graded involution} associated to $\Sigma$.

\item \textbf{Graded tensor isomorphism:} \label{TensorIsometryAxiom}
There is a graded isomorphism $\tau$ from $\cH_{\Sigma_1}\otimes\cH_{\Sigma_2}$ to $\cH_{\Sigma_1\sqcup\Sigma_2}$ for every pair of hypersurfaces $\Sigma_1,\Sigma_2$. We call $\tau$ the \textit{graded tensor isomorphism} associated to $\Sigma_1,\Sigma_2$.

\item \textbf{Graded gluing epimorphism:} \label{GluingIsometryAxiom}
There is a graded epimorphism $\gamma_f$ from $\cH_\Sigma$ to $\cH_{\Sigma'}$ for every hypersurface gluing function $f:\Sigma\to\Sigma'$. We call $\gamma$ the \textit{graded gluing epimorphism} associated to $f$. If $f$ is a homeomorphism, then $\gamma_f$ is a graded isomorphism.  Moreover, if $f$ is the identity then so is $\gamma_f$.

\item \textbf{Graded amplitude maps:} \label{AmplitudeMapsAxiom}
There is a graded linear map $\rho_X:\cH_{\partial X}\to\C$ for every region $X$. We call $\rho_X$ the \textit{graded amplitude map} associated to $X$. We require that $\rho_\emptyset=\id_\C$.
\end{enumerate}
The data presented above satisfies the following conditions:
\begin{enumerate}[label=(\Alph*)]

 \item \textbf{Associativity of tensor isomorphism:}
 \label{AssociativityTensorIsometryAxiom}
  Let $\Sigma_1,\Sigma_2,\Sigma_3$ be hypersurfaces. Then, the following square commutes: 
 \begin{center}
 \begin{tikzpicture}
  \matrix (m) [matrix of math nodes,row sep=4em,column sep=4em,minimum width=2em]
  {
     \cH_{\Sigma_1}\otimes\cH_{\Sigma_2}\otimes\cH_{\Sigma_3}&\cH_{\Sigma_1\sqcup\Sigma_2}\otimes\cH_{\Sigma_3} \\
		 \cH_{\Sigma_1}\otimes\cH_{\Sigma_2\sqcup\Sigma_3}&\cH_{\Sigma_1\sqcup\Sigma_2\sqcup\Sigma_3} \\};
  \path[-stealth]
    (m-1-1) edge node [left] {$id_{\cH_{\Sigma_1}}\otimes\tau$} (m-2-1)
            edge node [above] {$\tau\otimes id_{\cH_{\Sigma3}}$} (m-1-2)
    (m-1-2) edge node [right] {$\tau$} (m-2-2)
    (m-2-1) edge node [below] {$\tau$}(m-2-2);
 \end{tikzpicture}
 \end{center}
 
 \item \textbf{Associativity of gluing epimorphism:}
 \label{AssociativityofGluingIsometryAxiom}
 Let $\Sigma$ be a hypersurface. Let
 \begin{center}
 \begin{tikzpicture}
  \matrix (m) [matrix of math nodes,row sep=4em,column sep=4em,minimum width=2em]
  {\Sigma&\Sigma_{ 2} \\
	\Sigma_{ 1}&\Sigma_{ 1,2}\\};
  \path[-stealth]
    (m-1-1) edge node [left] {$f_1$} (m-2-1)
            edge node [above] {$f_2$} (m-1-2)
    (m-1-2) edge node [right] {$f_{2,1}$} (m-2-2)
    (m-2-1) edge node [below] {$f_{1,2}$}(m-2-2);
 \end{tikzpicture}
 \end{center}
 be a commuting square of hypersurface gluing functions. Then, the following square commutes:
 \begin{center}
 \begin{tikzpicture}
  \matrix (m) [matrix of math nodes,row sep=4em,column sep=4em,minimum width=2em]
  {\cH_\Sigma&\cH_{\Sigma_{ 2}} \\
	\cH_{\Sigma_{ 1}}&\cH_{\Sigma_{ 1,2}}\\};
  \path[-stealth]
    (m-1-1) edge node [left] {$\gamma_{f_1}$} (m-2-1)
            edge node [above] {$\gamma_{f_2}$} (m-1-2)
    (m-1-2) edge node [right] {$\gamma_{f_{2.1}}$} (m-2-2)
    (m-2-1) edge node [below] {$\gamma_{f_{1,2}}$}(m-2-2);
 \end{tikzpicture}
 \end{center}
 
 \item \textbf{Tensor and involution isomorphisms:}
 \label{ConjugationTensorIsometryAxiom}
 Let $\Sigma_1,\Sigma_2$ be hypersurfaces. Then, the following diagram commutes:
 \begin{center}
	\begin{tikzpicture}
		\matrix (m) [matrix of math nodes,row sep=3em,column sep=3em,minimum width=2em]
		{\cH_{\Sigma_1}\times\cH_{\Sigma_2}&&\cH_{\Sigma_2}\times\cH_{\Sigma_1} \\
		 \cH_{\overline{\Sigma}_1}\times\cH_{\overline{\Sigma}_2}&&\cH_{\Sigma_2}\otimes\cH_{\Sigma_1}\\
		 \cH_{\overline{\Sigma}_1}\otimes\cH_{\overline{\Sigma}_2}&&\cH_{\Sigma_2\sqcup\Sigma_1}\\
		 &\cH_{\overline{\Sigma}_1\sqcup\overline{\Sigma}_2}&\\};
		\path[-stealth]
		(m-1-1) edge node [left] {$\iota\times\iota$} (m-2-1)
		edge node [above]  {$\sigma$} (m-1-3)
		(m-1-3) edge node [right] {} (m-2-3)
		(m-2-1) edge node [left] {} (m-3-1)
		(m-3-1) edge node [below]  {$\tau$} (m-4-2)
		(m-2-3) edge node [right] {$\tau$} (m-3-3)
		(m-3-3) edge node [below] {$\iota$} (m-4-2)
		
		;
	\end{tikzpicture}
 \end{center}
 Here, $\sigma$ is the flip map and the maps from the cartesian products to the tensor products are the usual projections.

 \item \textbf{Tensor isomorphism and gluing epimorphism:}
 \label{TensorandGluingIsometryAxiom}
 Let $f_1:\Sigma_1\to\Sigma_1'$, $f_2:\Sigma_2\to\Sigma_2'$ be hypersurface gluing functions. The following square commutes: 
 \begin{center}
 \begin{tikzpicture}
  \matrix (m) [matrix of math nodes,row sep=4em,column sep=4em,minimum width=2em]
  {\cH_{\Sigma_1}\otimes\cH_{\Sigma_2}&\cH_{\Sigma_{1}'}\otimes\cH_{\Sigma_{2}'} \\
	\cH_{\Sigma_1\sqcup\Sigma_2}&\cH_{\Sigma_1'\sqcup\Sigma_2'} \\};
  \path[-stealth]
    (m-1-1) edge node [left] {$\tau_{\Sigma_1,\Sigma_2}$} (m-2-1)
            edge node [above] {$\gamma_{f_1}\otimes\gamma_{f_2}$} (m-1-2)
    (m-1-2) edge node [right] {$\tau_{\Sigma_1',\Sigma_2'}$} (m-2-2)
    (m-2-1) edge node [below] {$\gamma_{f_1\sqcup f_2}$}(m-2-2);
 \end{tikzpicture}
 \end{center}

 \item \textbf{Gluing homomorphism and involution:}
 \label{ConjugationGluingIsometryAxiom}
 Let $f:\Sigma\to\Sigma'$ be a gluing function of hypersurfaces. The following square commutes:
 \begin{center}
 \begin{tikzpicture}
  \matrix (m) [matrix of math nodes,row sep=4em,column sep=4em,minimum width=2em]
  {
   \cH_\Sigma&\cH_{\Sigma'} \\
   \cH_{\overline{\Sigma}}&\cH_{\overline{\Sigma'}} \\};
  \path[-stealth]
    (m-1-1) edge node [left] {$\iota_{\Sigma}$} (m-2-1)
            edge node [above] {$\gamma_f$} (m-1-2)
    (m-1-2) edge node [right] {$\iota_{\Sigma'}$} (m-2-2)
    (m-2-1) edge node [below] {$\gamma_{\overline{f}}$}(m-2-2);
 \end{tikzpicture}
 \end{center}

 \item \textbf{Amplitude and tensor isomorphism:}\label{AmplitudeandTensorIsometryAxiom}
 Let $X_1,X_2$ be regions. The following diagram commutes:
 \begin{center}
	\begin{tikzpicture}
	 \matrix (m) [matrix of math nodes,row sep=4em,column sep=4em,minimum width=2em]
	 {\cH_{\partial X_1}\tens\cH_{\partial X_2}&\C\tens\C \\
	  \cH_{\partial X_1 \sqcup\partial X_2}&\C \\};
	 \path[-stealth]
	   (m-1-1) edge node [above] {$\rho_{X_1}\tens\rho_{X_2}$} (m-1-2)
	   (m-1-1) edge node [left]  {$\tau$} (m-2-1)
	   (m-2-1) edge node [below] {$\rho_{X_1\sqcup X_2}$} (m-2-2)
	   (m-1-2) edge node [right]  {$\id$} (m-2-2)
	   ;
    \end{tikzpicture}
 \end{center}
   
 \item \textbf{Amplitude and involution:}
 \label{AmplitudeConjugationIsometryAxiom}
 Let $X$ be a region. The following square commutes:
 \begin{center}
 \begin{tikzpicture}
  \matrix (m) [matrix of math nodes,row sep=4em,column sep=4em,minimum width=2em]
  {\cH_{\partial X}&\cH_{\overline{\partial X}} \\
   \C&\C\\};
  \path[-stealth]
    (m-1-1) edge node [above] {$\iota_{\partial X}$} (m-1-2)
    (m-1-1) edge node [left]  {$\rho_X$} (m-2-1)
    (m-2-1) edge node [below] {$\overline{(\cdot)}$} (m-2-2)
    (m-1-2) edge node [right]  {$\rho_{\overline{X}}$} (m-2-2)
    ;
 \end{tikzpicture}
 \end{center}
 Here, $\overline{(\cdot)}$ denotes the complex conjugation as an anti-linear automorphism of $\C$.

\item \textbf{Amplitude and homeomorphism:}
\label{AmplitudeEquivarianceAxiom}
Let $X$ and $X'$ be regions. Let $f:X\to X'$ be a homeomorphism. The following diagram commutes:
\begin{center}
	\begin{tikzpicture}
		\matrix (m) [matrix of math nodes,row sep=4em,column sep=3em,minimum width=2em]
		{\cH_{\partial X}&&\cH_{\partial X'} \\
			&\C&\\};
		\path[-stealth]
		(m-1-1) edge node [left] {$\rho_{X}$} (m-2-2)
		edge node [above] {$\gamma_{f|_{\partial X}}$} (m-1-3)
		(m-1-3) edge node [right] {$\rho_{X'}$} (m-2-2);
	\end{tikzpicture}
\end{center}

 \item \textbf{Evaluation and slice regions:}
 \label{SliceRegionsAxiom}
 Given a hypersurface $\Sigma$ and associated slice region $(S,\alpha)$, we define the following map:
 \begin{center}
	\begin{tikzpicture}
		\matrix (m) [matrix of math nodes,row sep=4em,column sep=4em,minimum width=2em]
		{
		\ev_{\Sigma}:\cH_{\overline{\Sigma}}\otimes\cH_\Sigma&\cH_{\overline{\Sigma}\sqcup\Sigma}&\cH_{\partial S}&\C \\};
		\path[-stealth]
		(m-1-1) edge node [above] {$\tau_{\overline{\Sigma},\Sigma}$} (m-1-2)
    	(m-1-2) edge node [above] {$\gamma_\alpha$}(m-1-3)
    	(m-1-3) edge node [above] {$\rho_S$}(m-1-4)
		;
	\end{tikzpicture}
 \end{center}
 (By Axiom~\ref{AmplitudeEquivarianceAxiom} combined with Lemma~\ref{examplecomputationhomeomorphisms}, this is independent of the choice of $(S,\alpha)$.) We require $\ev_{\Sigma}$ to be non-degenerate in the sense that there exists a map $\coev_{\Sigma}:\C\to\cH_{\Sigma}\tens\cH_{\overline{\Sigma}}$ such that the following diagram commutes:
 \begin{center}
	\begin{tikzpicture}
	 \matrix (m) [matrix of math nodes,row sep=4em,column sep=4em,minimum width=2em]
	 {\C\tens\cH_{\Sigma}& \cH_{\Sigma} & \\
	  \cH_{\Sigma}\tens\cH_{\overline{\Sigma}}\tens\cH_{\Sigma}&\cH_{\Sigma}\tens\C \\};
	 \path[-stealth]
	   (m-1-2) edge node [above] {$\id$} (m-1-1)
	   (m-1-1) edge node [left]  {$\coev_{\Sigma}\tens\id$} (m-2-1)
	   (m-2-1) edge node [below] {$\id\tens\ev_{\Sigma}$} (m-2-2)
	   (m-2-2) edge node [right] {$\id$} (m-1-2)
	   ;
    \end{tikzpicture}
 \end{center}

 \item \textbf{Amplitude and relative gluing diagrams:}
 \label{AmplitudeRelativeGluingAxiom}
 Let
 \begin{center}
 \begin{tikzpicture}
  \matrix (m) [matrix of math nodes,row sep=3em,column sep=3em,minimum width=2em]
  {\Lambda&&&\partial X' \\
   \Sigma\sqcup \overline{\Sigma}\sqcup\Lambda&\partial X&X&X'\\};
  \path[-stealth]
    (m-1-1) edge node [above] {$\beta$} (m-1-4)
            edge node [left]  {} (m-2-1)
    (m-2-1) edge node [below] {$\alpha$} (m-2-2)
    (m-2-2) edge node [below]  {} (m-2-3)
    (m-2-3) edge node [below] {$f$} (m-2-4)
    (m-1-4) edge node [right] {} (m-2-4)
    ;
 \end{tikzpicture}
 \end{center}
 be a relative gluing diagram. Then, the following diagram commutes:
 \begin{center}
	\begin{tikzpicture}
	 \matrix (m) [matrix of math nodes,row sep=3em,column sep=3em,minimum width=2em]
	 {\C\tens\cH_{\Lambda}&\cH_{\Lambda}&&&\cH_{\partial X'} \\
	  \cH_{\Sigma}\tens\cH_{\overline{\Sigma}}\tens\cH_{\Lambda}&\cH_{\Sigma\sqcup \overline{\Sigma}\sqcup\Lambda}&\cH_{\partial X}&\C &\C \\};
	 \path[-stealth]
	   (m-1-1) edge node [left] {$\coev_{\Sigma}\tens\id$} (m-2-1)
	   (m-1-2) edge node [above] {$\gamma_\beta$} (m-1-5)
			   edge node [above]  {$\id$} (m-1-1)
	   (m-2-1) edge node [below] {$\tau$} (m-2-2)
	   (m-2-2) edge node [below] {$\gamma_\alpha$} (m-2-3)
	   (m-2-3) edge node [below]  {$\rho_X$} (m-2-4)
	   (m-2-4) edge node [below] {$c_{f,\beta}$} (m-2-5)
	   (m-1-5) edge node [right] {$\rho_{X'}$} (m-2-5)
	   ;
	\end{tikzpicture}
 \end{center}
Here, $c_{f,\beta}\in\C\setminus\left\{0\right\}$ depends on the gluing diagram $(f,\beta)$. We call the number $c_{f,\beta}$ the gluing anomaly associated to $(f,\beta)$. The arrow labeled by $c_{f,\beta}$ represents multiplication by this number. We require $c_{f,\beta}=1$ for trivial gluing diagrams in the sense of Observation~\ref{ObsHomeomorphismsRelative}, i.e., for any homeomorphism $f:X\to X'$ the gluing anomaly $c_{f,f|_{\partial X}}$ is equal to 1.
\end{enumerate}
We call a CQFT \emph{symmetric} if in addition it satisfies the following axiom:
\begin{enumerate}[label=(\Alph*)]
\setcounter{enumi}{\numexpr18\relax}
\item \textbf{Tensor isomorphism and graded symmetry:}
\label{TensorIsometryandGradedSymmetryAxiom}
Let $\Sigma_1,\Sigma_2$ be hypersurfaces. Then, the following square commutes:
\begin{center}
\begin{tikzpicture}
 \matrix (m) [matrix of math nodes,row sep=4em,column sep=4em,minimum width=2em]
 {\cH_{\Sigma_1}\otimes\cH_{\Sigma_2}&\cH_{\Sigma_{2}}\otimes\cH_{\Sigma_{1}} \\
 \cH_{\Sigma_1\sqcup\Sigma_2}&\cH_{\Sigma_2\sqcup\Sigma_1}\\};
 \path[-stealth]
   (m-1-1) edge node [left] {$\tau_{\Sigma_1,\Sigma_2}$} (m-2-1)
		   edge node [above] {$s_{\Sigma_1,\Sigma_2}$} (m-1-2)
   (m-1-2) edge node [right] {$\tau_{\Sigma_2,\Sigma_1}$} (m-2-2)
   (m-2-1) edge node [below] {$\id$} (m-2-2);
\end{tikzpicture}
\end{center}
Here, $s_{\Sigma_1,\Sigma_2}$ is the isomorphism defined as $s_{\Sigma_1,\Sigma_2}(v_1\otimes v_2)=(-1)^{|v_1||v_2|}(v_2\otimes v_1)$ for homogeneous $v_1\in \cH_{\Sigma_1}$ and $v_2\in \cH_{\Sigma_2}$ and where the bottom-most arrow is the identity of the state space associated to $\Sigma_1\sqcup\Sigma_2=\Sigma_2\sqcup \Sigma_1$.
\end{enumerate}
\end{definition}

\subsection{First consequences and an inner product}

The first implication of the axioms that we remark on concerns the fact that Axiom~\ref{AmplitudeEquivarianceAxiom} arises as a special case of Axiom~\ref{AmplitudeRelativeGluingAxiom}. 
\begin{lem}\label{examplecomputationhomeomorphisms}
	Axiom~\ref{AmplitudeRelativeGluingAxiom} implies Axiom~\ref{AmplitudeEquivarianceAxiom}.
\end{lem}
\begin{proof}
	Let $X$ and $X'$ be regions. Let $f:X\to X'$ be a homeomorphism. Let $(f,f|_{\partial X})$ be the trivial relative gluing diagram associated to $f$ as in Observation~\ref{ObsHomeomorphismsRelative}:
	\begin{center}
		 
		 \begin{tikzpicture}
			 \matrix (m) [matrix of math nodes,row sep=3em,column sep=3em,minimum width=2em]
			 {\partial X&&&\partial X' \\
				 \partial X&\partial X&X&X'\\};
			 \path[-stealth]
			 (m-1-1) edge node [above] {$f|_{\partial X}$} (m-1-4)
			 edge node [left]  {$id_{\partial X}$} (m-2-1)
			 
			 (m-2-1) edge node [below] {$id_{\partial X}$} (m-2-2)
			 (m-2-2) edge node [below]  {} (m-2-3)
			 (m-2-3) edge node [below] {$f$} (m-2-4)
			 (m-1-4) edge node [right] {} (m-2-4)
			 
			 ;
		 \end{tikzpicture}
	\end{center}
	In the notation of Axiom~\ref{AmplitudeRelativeGluingAxiom}, $\Lambda=\partial X$, $\Sigma=\emptyset$ (and thus $\cH_\Sigma=\C$), $\alpha=\id_{\partial X}$ and $\beta=f|_{\partial X}$. The second diagram of Axiom~\ref{AmplitudeRelativeGluingAxiom} then collapses to the equation,
	\begin{equation}\label{examplehomeomorphismsformula2}
	\rho_{X'}	\gamma_{f|_{\partial X}}(w)=c_{f,f|_{\partial X}}\rho_X(w)
	\end{equation}
	for all $w\in\cH_{\partial X}$. Since $f$ is a homeomorphism, Axiom~\ref{AmplitudeRelativeGluingAxiom} requires $c_{f,f|_{\partial X}}=1$. We conclude that the triangle in the statement of Axiom~\ref{AmplitudeEquivarianceAxiom} commutes.
\end{proof}

In spite of this result, it might appear as if Axiom~\ref{AmplitudeEquivarianceAxiom} was not in fact redundant, for the following reason. Axiom~\ref{SliceRegionsAxiom} depends on Axiom~\ref{AmplitudeEquivarianceAxiom}. In turn Axiom~\ref{AmplitudeRelativeGluingAxiom} depends on Axiom~\ref{SliceRegionsAxiom}. So replacing the dependency on Axiom~\ref{AmplitudeEquivarianceAxiom} in Axiom~\ref{SliceRegionsAxiom} with a dependency on Axiom~\ref{AmplitudeRelativeGluingAxiom} would lead to an apparent dependency loop. However, as we have seen with Lemma~\ref{examplecomputationhomeomorphisms}, Axiom~\ref{AmplitudeEquivarianceAxiom} arises as a special case of Axiom~\ref{AmplitudeRelativeGluingAxiom}. Crucially, in this special case the coevaluation map is trivial and the invariance property of the corresponding evaluation map of Axiom~\ref{SliceRegionsAxiom} is trivially satisfied, without recurrence to Axiom~\ref{AmplitudeEquivarianceAxiom}. Thus, Axiom~\ref{AmplitudeEquivarianceAxiom} may be considered redundant after all, and we will often omit it when considering the axiomatic system of Definition~\ref{examplecomputationhomeomorphisms}. The main reason we have included it, is that we sometimes wish to consider a subset of the axioms. In particular, we will be interested at some point in considering a subset that does not include Axiom~\ref{AmplitudeRelativeGluingAxiom}, but does include Axiom~\ref{AmplitudeEquivarianceAxiom}.

We continue with the observation that from the evaluation map defined in Axiom~\ref{SliceRegionsAxiom} we may derive a sesquilinear form $\langle\cdot,\cdot\rangle_\Sigma$ on $\cH_{\Sigma}$ given by,
\begin{equation}
  \langle v,w\rangle_{\Sigma}\defeq \ev_{\Sigma}(\iota_{\Sigma}(v)\tens w) .
  \label{eq:propip}
\end{equation}

\begin{thm}\label{thm:propip}
	The sesquilinear form $\langle\cdot,\cdot\rangle_\Sigma$ has the following properties:
	\begin{enumerate}
		\item It is hermitian, non-degenerate and graded and thus a graded inner product.
		\item The graded gluing epimorphisms of Axiom~\ref{GluingIsometryAxiom} that originate from gluing functions that are homeomorphisms, are isomorphic isometries.
	\end{enumerate}
	If Axiom~\ref{TensorIsometryandGradedSymmetryAxiom} is satisfied, the following additional properties hold: 
	\begin{enumerate}
		\setcounter{enumi}{2}
		\item The graded involutions of Axiom~\ref{ConjugationIsometryAxiom} are anti-linear graded isometries.
		\item The graded tensor isomorphisms of Axiom~\ref{TensorIsometryAxiom} are isometries.
	\end{enumerate}
\end{thm}
\begin{proof}
	The proof is provided point by point.
	\begin{enumerate}
	\item The non-degeneracy of the sesquilinear form is mandated by Axiom~\ref{SliceRegionsAxiom}. It is graded since it is a composition of graded maps. The hermiticity follows with the Axioms~\ref{ConjugationTensorIsometryAxiom}, \ref{ConjugationGluingIsometryAxiom}, and \ref{AmplitudeConjugationIsometryAxiom}. To see this consider a slice region $(S,\alpha)$ for the hypersurface $\Sigma$ with boundary gluing map $\alpha:\overline{\Sigma}_1\sqcup\Sigma_2\to\partial S$. (For greater clarity, we distinguish the two copies of $\Sigma$ with indices.) The following diagram commutes.
	\end{enumerate}%
	\begin{center}
		\begin{tikzpicture}
			\matrix (m) [matrix of math nodes,row sep=4em,column sep=4em,minimum width=2em]
			{
				\cH_{\Sigma_1}\times\cH_{\Sigma_2}&\cH_{\overline{\Sigma}_1}\times\cH_{\Sigma_2} & \cH_{\overline{\Sigma}_1}\otimes\cH_{\Sigma_2}& \cH_{\overline{\Sigma}_1\sqcup \Sigma_2}&\cH_{\partial S}&\mathbb{C} \\
				\cH_{\Sigma_2}\times\cH_{\Sigma_1}&\cH_{\overline{\Sigma}_2}\times\cH_{\Sigma_1} & \cH_{\overline{\Sigma}_2}\otimes\cH_{\Sigma_1}& \cH_{\overline{\Sigma}_2\sqcup \Sigma_1}&\cH_{\partial \overline{S}}&\mathbb{C} \\};
			\path[-stealth]
			(m-1-1) edge node [left] {$\sigma$} (m-2-1)
			edge node [above] {$\iota \times id$} (m-1-2)
			(m-1-4) edge node [right] {$\iota$} (m-2-4)
			(m-2-1) edge node [below] {$\iota\times id$}(m-2-2)
			(m-1-2) edge node [above]{$\otimes$} (m-1-3)
			(m-1-3) edge node [above]{$\tau$} (m-1-4)
			(m-1-4) edge node [above]{$\gamma_\alpha$}(m-1-5)
			(m-1-5) edge node [above] {$\rho_S$}(m-1-6)
			
			(m-2-2) edge node [below]{$\otimes$} (m-2-3)
			(m-2-3) edge node [below]{$\tau$} (m-2-4)
			(m-2-4) edge node [below]{$\gamma_{\overline{\alpha}}$}(m-2-5)
			(m-2-5) edge node [below] {$\rho_{\overline{S}}$}(m-2-6)

			(m-1-5) edge node [right]{$\iota$}(m-2-5)
			(m-1-6) edge node [right]{$\overline{(\cdot)}$} (m-2-6)
			
			(m-1-1) edge [white] node [black] {$(C)$}(m-2-4)
			(m-1-4) edge [white] node [black] {$(E)$}(m-2-5)
			(m-1-5) edge [white] node [black] {$(G)$}(m-2-6);
		\end{tikzpicture}
	\end{center}
	\begin{enumerate}
	\setcounter{enumi}{1}
	\item Let $f:\Sigma\to\widetilde{\Sigma}$ be a homeomorphism between hypersurfaces. Let $(S,\widetilde{\alpha})$ be a slice region for $\widetilde{\Sigma}$. Then, $(S,\alpha)$ is a slice region for $\Sigma$, where $\alpha=\widetilde{\alpha} \circ \overline{f}\sqcup f$. It follows from Axioms~\ref{AssociativityofGluingIsometryAxiom}, \ref{TensorandGluingIsometryAxiom} and \ref{ConjugationGluingIsometryAxiom} that $\gamma_f$ is an isometric isomorphism. This is the following commutative diagram.
	\end{enumerate}%
	\begin{center}
		\begin{tikzpicture}
			\matrix (m) [matrix of math nodes,row sep=4em,column sep=4em,minimum width=2em]
				{
					\cH_{\Sigma_1}\times\cH_{\Sigma_2}&\cH_{\overline{\Sigma}_1}\times\cH_{\Sigma_2} & \cH_{\overline{\Sigma}_1}\otimes\cH_{\Sigma_2}& \cH_{\overline{\Sigma}_1\sqcup \Sigma_2}&\cH_{\partial S}&\C \\
					\cH_{\widetilde{\Sigma}_1}\times\cH_{\widetilde{\Sigma}_2}&\cH_{\overline{\widetilde{\Sigma}}_1}\times\cH_{\widetilde{\Sigma}_2} & \cH_{\overline{\widetilde{\Sigma}}_1}\otimes\cH_{\widetilde{\Sigma}_2}& \cH_{\overline{\widetilde{\Sigma}}_1\sqcup \widetilde{\Sigma}_2}&\cH_{\partial S}&\C \\};
				\path[-stealth]
				(m-1-1) edge node [above] {$\iota \times id$} (m-1-2)
				(m-1-2) edge node [above]{$\tens$} (m-1-3)
				(m-1-3) edge node [above]{$\tau$} (m-1-4)
				(m-1-4) edge node [above]{$\gamma_{\alpha}$}(m-1-5)
				(m-1-5) edge node [above] {$\rho_S$}(m-1-6)
				
				(m-2-1) edge node [below] {$\iota \times id$}(m-2-2)
				(m-2-2) edge node [below]{$\tens$} (m-2-3)
				(m-2-3) edge node [below]{$\tau$} (m-2-4)
				(m-2-4) edge node [below]{$\gamma_{\widetilde{\alpha}}$}(m-2-5)
				(m-2-5) edge node [below] {}(m-2-6)

				(m-1-1) edge node [left] {$\gamma_f\times \gamma_f$} (m-2-1)
				(m-1-2) edge node [left] {$\gamma_{\overline{f}}\times\gamma_f$}(m-2-2)
				(m-1-3) edge node [left] {$\gamma_{\overline{f}}\otimes \gamma_f$}(m-2-3)
				(m-1-4) edge node [right] {$\gamma_{\overline{f}\sqcup f}$} (m-2-4)
				(m-1-5) edge node [right]{$id$}(m-2-5)
				(m-1-6) edge node [right]{$id$} (m-2-6)
				
				(m-1-1) edge [white] node [black] {$(E)$}(m-2-2)
				(m-1-2) edge [white] node [black] {$\otimes$}(m-2-3)
				(m-1-3) edge [white] node [black] {$(D)$}(m-2-4)
				(m-1-4) edge [white] node [black] {$(B)$}(m-2-5);
			\end{tikzpicture}
	\end{center}
	\begin{enumerate}
	\setcounter{enumi}{2}
	\item The graded isometry property of the involutions can be stated as,
	\begin{equation}
		\langle \iota_{\Sigma}(a),\iota_{\Sigma}(b)\rangle_{\overline{\Sigma}}=(-1)^{|a| |b|}\langle b,a\rangle_{\Sigma} .
	\end{equation}
	In terms of the evaluation this is equivalent to the relation,
	\begin{equation}
		\ev_{\Sigma}(a\tens b)=(-1)^{|a| |b|}\ev_{\overline{\Sigma}}(b\tens a) .
	\end{equation}
	Consider again a slice region $(S,\alpha)$ for the hypersurface $\Sigma$ with boundary gluing map $\alpha:\overline{\Sigma}_1\sqcup\Sigma_2\to\partial S$. (Again we distinguish the two copies of $\Sigma$ with indices.) The stated relation then follows from Axiom~\ref{TensorIsometryandGradedSymmetryAxiom} as can be seen via the commutation of the following diagram.
	\begin{center}
		\begin{tikzpicture}
			\matrix (m) [matrix of math nodes,row sep=4em,column sep=4em,minimum width=2em]
			{
				\cH_{\overline{\Sigma}_1}\otimes\cH_{\Sigma_2}&\cH_{\overline{\Sigma}_1\sqcup \Sigma_2} & \cH_{\partial S}&\C \\
				\cH_{\Sigma_2}\otimes\cH_{\overline{\Sigma}_1}&\cH_{\Sigma_2\sqcup \overline{\Sigma}_1} &\cH_{\partial S}&\C \\};
			\path[-stealth]
			(m-1-1) edge node [above] {$\tau$} (m-1-2)
			(m-1-2) edge node [above] {$\gamma_{\alpha}$} (m-1-3)
			(m-1-3) edge node [above]{$\rho_S$} (m-1-4)

			(m-2-1) edge node [below] {$\tau$}(m-2-2)
			(m-2-2) edge node [below]{$\gamma_{\alpha}$} (m-2-3)
			(m-2-3) edge node [below]{$\rho_S$} (m-2-4)

			(m-1-1) edge node [left] {$s_{\overline{\Sigma}_1,\overline{\Sigma}_2}$} (m-2-1)
			(m-1-2) edge node [left] {$id$}(m-2-2)
			(m-1-3) edge node [left] {$id$}(m-2-3)
			(m-1-4) edge node [right] {$id$} (m-2-4)

			(m-1-1) edge [white] node [black] {$(S)$}(m-2-2);
		\end{tikzpicture}
	\end{center}
	\item Let $\Sigma$ and $\Sigma'$ be hypersurfaces. The stated property can be written as,
    \begin{equation}
		\langle \tau_{\Sigma_1,\Sigma_2}(a\tens b),\tau_{\Sigma_1,\Sigma_2}(a'\tens b')\rangle_{\Sigma_1\sqcup\Sigma_2}
		=\langle a,a'\rangle_{\Sigma_1} \langle b,b'\rangle_{\Sigma_2} .
	\end{equation}
	Using the definition of the inner product in terms of the evaluation map and using Axiom~\ref{ConjugationTensorIsometryAxiom}, this can be seen to be equivalent to the following relation,
	\begin{equation}
		\ev_{\Sigma_1\sqcup\Sigma_2}(\tau_{\overline{\Sigma}_2,\overline{\Sigma}_1}(b\tens a)\tens\tau_{\Sigma_1,\Sigma_2}(a'\tens b'))
		=\ev_{\Sigma_1}(a\tens a') \ev_{\Sigma_2}(b\tens b') .
	\end{equation}
	We invite the reader to verify that a similar relation can be derived using Axioms~\ref{AssociativityTensorIsometryAxiom}, \ref{TensorandGluingIsometryAxiom}, \ref{AmplitudeandTensorIsometryAxiom}, as well as \ref{TensorIsometryandGradedSymmetryAxiom}. In fact, the relation obtained is the following,
	\begin{equation}
		\ev_{\Sigma_1\sqcup\Sigma_2}(\tau_{\overline{\Sigma}_2,\overline{\Sigma}_1}(b\tens a)\tens\tau_{\Sigma_1,\Sigma_2}(a'\tens b'))
		=(-1)^{(|a|+|a'|)|b|}\ev_{\Sigma_1}(a\tens a') \ev_{\Sigma_2}(b\tens b') .
	\end{equation}
	Observe, however, that the evaluation map is graded and has image of even degree. Thus, only terms where the degree of $a\tens a'$ is even, contribute. But then, $|a|+|a'|$ is even and $(-1)^{(|a|+|a'|)|b|}=1$, so both relations are equivalent.
    \end{enumerate}
	This completes the proof.
\end{proof}

\begin{obs}
	\label{obs:coev}
	The property of the inner product being graded is equivalent to the decomposition $\cH_{\Sigma}=\cH_{\Sigma}^0 \oplus \cH_{\Sigma}^1$ being orthogonal. We thus have individually inner products on $\cH_{\Sigma}^0$ and on $\cH_{\Sigma}^1$. The inner products are in general indefinite, and we can decompose in terms of orthogonal positive- and negative-definite subspaces as $\cH_{\Sigma}^0=\cH_{\Sigma}^{0,+}\oplus \cH_{\Sigma}^{0,-}$ and $\cH_{\Sigma}^1=\cH_{\Sigma}^{1,+}\oplus \cH_{\Sigma}^{1,-}$ with the obvious notation.
	Consider now an orthonormal basis $\{\zeta_k\}_{k\in I}$ of the space $\cH_{\Sigma}$ consisting of orthonormal basis of the subspaces $\cH_{\Sigma}^{0,+}$, $\cH_{\Sigma}^{0,-}$, $\cH_{\Sigma}^{1,+}$, $\cH_{\Sigma}^{1,-}$. By an orthonormal basis for a negative-definite inner product space we understand here an orthonormal basis for the sign-reverted inner product. We use the notation $[v]=0$ if $v\in \cH_{\Sigma}^{+}$ and $[v]=1$ if $v\in\cH_{\Sigma}^{-}$. For basis elements, the evaluation map $\ev_{\Sigma}:\cH_{\overline{\Sigma}}\tens\cH_{\Sigma}$ takes the form,
	\begin{equation}
		\ev_{\Sigma}(\iota_{\Sigma}(\zeta_k),\zeta_l)=(-1)^{[\zeta_k]}\delta_{k,l} .
	\end{equation}
	The coevaluation map $\coev_{\Sigma}:\C\to\cH_{\Sigma}\tens\cH_{\overline{\Sigma}}$ may be identified with the element $\coev_{\Sigma}(1)\in \cH_{\Sigma}\tens\cH_{\overline{\Sigma}}$. This takes the form,
	\begin{equation}
		\label{eq:coevsum}
		\coev_{\Sigma}(1)=\sum_{k\in I} (-1)^{[\zeta_k]} \zeta_k\tens \iota_{\Sigma}(\zeta_k) .
	\end{equation}
	We may then write the second diagram of Axiom~\ref{AmplitudeRelativeGluingAxiom} as the following identity, to be satisfied for any $v\in \cH_{\Lambda}$,
	\begin{equation}\label{EquationAmplitudeRelativeGluingData}
		\rho_{X'}(\gamma_\beta(v))=c_{f,\beta}\sum_{k\in I}(-1)^{[\zeta_k]}\rho_X\gamma_\alpha\tau( \zeta_k\otimes\iota_\Sigma(\zeta_k)\otimes v) .
	\end{equation}
\end{obs}

We exhibit here one further consequence of the axioms: The invariance of the gluing anomaly factor under deformation by homeomorphisms. 
\begin{lem}\label{LemmaGluingAnomalyInvariant}
Let $(f,\beta)$ be a relative gluing diagram of the form:
\begin{center}
	\begin{tikzpicture}
		\matrix (m) [matrix of math nodes,row sep=3em,column sep=3em,minimum width=2em]
		{\Lambda&&&\partial X' \\
			\Sigma\sqcup \overline{\Sigma}\sqcup\Lambda&\partial X&X&X'\\};
		\path[-stealth]
		(m-1-1) edge node [above] {$\beta$} (m-1-4)
		edge node [left]  {} (m-2-1)
		
		(m-2-1) edge node [below] {$\alpha$} (m-2-2)
		(m-2-2) edge node [below]  {} (m-2-3)
		(m-2-3) edge node [below] {$f$} (m-2-4)
		(m-1-4) edge node [right] {} (m-2-4)
		
		;
	\end{tikzpicture}
\end{center}
Let $g:Y\to X$ and $h:X'\to Y'$ be homeomorphisms. Let $(hf,h\beta)$ and $(fg,\beta g)$ be the relative gluing diagrams associated to $(f,\beta)$, $h$ and $g$ in Observation~\ref{observationRelativeClosedCompositionHomeo}. The gluing anomalies $c_{hf,h\beta}$ and $c_{fg,\beta g}$ associated to $(hf,h\beta)$ and $(fg,\beta g)$ respectively, are both equal to the gluing anomaly $c_{f,\beta}$ associated to $(f,\beta)$.
\end{lem}
\begin{proof}
 Let $(f,\beta)$ be a relative gluing diagram as in the statement of the Lemma. Let $g:Y\to X$ and $h:X'\to Y'$ be homeomorphisms. We first compute the graded amplitude of the "deformed" gluing diagram $(hf,h\beta)$. Let $u\in\cH_{\Lambda}$. Applying Axiom~\ref{AmplitudeRelativeGluingAxiom} to $(hf,h\beta)$ and evaluating in $u$ we obtain the formula:
	\begin{equation}\label{EquationAmplitudeActionHomeosLeft}
		\rho_{Y'}(\gamma_{h\beta}(u))=c_{hf,h\beta}\sum_{k\in I}(-1)^{[\zeta_k]}\rho_X\gamma_\alpha(\zeta_k\otimes\iota_\Sigma(\zeta_k)\otimes u)
	\end{equation}
	Here, $\{\zeta_k\}_{k\in I}$ is any orthonormal basis of $\cH_{\Sigma}$. By Axiom~\ref{AssociativityofGluingIsometryAxiom} we have $\gamma_{h\beta}=\gamma_h \gamma_\beta$ and with Lemma~\ref{examplecomputationhomeomorphisms} we then obtain the equation $\rho_{Y'}(\gamma_{h\beta}(u))=\rho_{X'}(\gamma_{\beta}(u))$. Equation~(\ref{EquationAmplitudeActionHomeosLeft}) transforms into the following equation:
	\begin{equation}\label{EquationAmplitudeActionHomeosLeft2}
		\rho_{X'}(\gamma_\beta(u))=c_{hf,h\beta}\sum_{k\in I}(-1)^{[\zeta_k]}\rho_X\gamma_\alpha(\zeta_k\otimes \iota_\Sigma(\zeta_k)\otimes u) .
	\end{equation}
	Comparing this equation to the equation obtained by applying Axiom~\ref{AmplitudeRelativeGluingAxiom} to the original diagram $(f,\beta)$ we find that $c_{hf,h\beta}=c_{f,\beta}$.

	We now compute the graded amplitude map associated to $(fg,\beta g)$. Let $v\in\cH_{g^{-1}\Sigma}$. Applying Axiom~\ref{AmplitudeRelativeGluingAxiom} to $(fg,\beta g)$ and evaluating on $v$ we obtain the equation:
		\begin{equation}\label{EquationAmplitudeActionHomeosRight1}
		\rho_{X'}(\gamma_{\beta g} (v))=c_{fg,\beta g}\sum_{k\in I}(-1)^{[\xi_k]}\rho_Y\gamma_{g^{-1}\alpha g}\tau(\xi_k\otimes\iota_{g^{-1}\Sigma}(\xi_k)\otimes v)
	\end{equation}
	Here, $\{\xi_k\}_{k\in I}$ is any orthonormal basis of $\cH_{g^{-1}\Sigma}$. Let $u\in \cH_{\Sigma}$ be given by $u=\gamma_g(v)$ and $\zeta_k=\gamma_g(\xi_k)$. By Theorem~\ref{thm:propip}, $\{\zeta_k\}_{k\in I}$ is an orthonormal basis of $\cH_{\Sigma}$ and $[\zeta_k]=[\xi_k]$. Combining Lemma~\ref{examplecomputationhomeomorphisms} with Axioms~\ref{AssociativityofGluingIsometryAxiom}, \ref{TensorandGluingIsometryAxiom} and \ref{ConjugationGluingIsometryAxiom} we obtain the equality,
	\begin{equation}
		(-1)^{[\xi_k]}\rho_Y\gamma_{g^{-1}\alpha g}\tau(\xi_k\otimes\iota_{g^{-1}\Sigma}(\xi_k)\otimes v)
		= (-1)^{[\zeta_k]}\rho_X\gamma_\alpha\tau(\zeta_k\otimes \iota_\Sigma(\zeta_k)\otimes u) .
	\end{equation}
	We also have $\gamma_{\beta g} (v)=\gamma_{\beta}(u)$. This yields,
	\begin{equation}
		\rho_{X'}(\gamma_{\beta}(u))=c_{fg,\beta g}\sum_{k\in I}(-1)^{[\zeta_k]}\rho_X\gamma_\alpha\tau(\zeta_k\otimes \iota_\Sigma(\zeta_k)\otimes u)
	\end{equation}
	Comparing this with the equation obtained by applying Axiom~\ref{AmplitudeRelativeGluingAxiom} to the original diagram $(f,\beta)$ we obtain the equation $c_{f,\beta}=c_{fg,\beta g}$. This concludes the proof.
\end{proof}

\subsection{Infinite dimensions}

An obvious problem when transitioning from the finite-dimensional to an infinite-dimensional setting is the sum appearing in the coevaluation (\ref{eq:coevsum}), see Observation~\ref{obs:coev}. This sum would be infinite and thus not exist in the category of superspaces. We may work around this by replacing the second diagram in Axiom~\ref{AmplitudeRelativeGluingAxiom} with the corresponding explicit equation (\ref{eq:propip}). Of course, we need to require convergence of the sum. As laid out in the introduction, we shall guarantee this by explicitly restricting to admissible manifolds, gluing functions and gluing diagrams. Another aspect of the infinite-dimensional setting is that to make sense of orthonormal basis we need to equip the spaces $\cH_{\Sigma}$ with topologies which are precisely determined by the inner products. In other words, the spaces $\cH_{\Sigma}$ have to be \emph{complete} with respect to their inner product. The suitable notion here is that of \emph{Krein space}. That is, instead of dealing with the category of (finite-dimensional) superpaces we deal with the category of Krein superspaces. This also implies that the tensor product is not that of vector spaces, but its completion. This has wide-ranging consequences. Since the tensor product plays a role in many axioms, it has to be determined before an inner product would emerge from Axiom~\ref{SliceRegionsAxiom}. The most straightforward way to deal with this, is to equip the superspaces $\cH_{\Sigma}$ with the structure of a Krein space from the outset, that is in Axiom~\ref{Statespacesaxiom}. Axiom~\ref{SliceRegionsAxiom} then merely requires that the given inner product coincides with that derived from the evaluation map. This setup is reflected in the GBQFT axioms of Definition~\ref{AFdefinitionClassic}.

There is another issue that has to be taken into account in the infinite-dimensional setting. It turns out that the amplitude map of Axiom~\ref{AmplitudeMapsAxiom} is generically unbounded for realistic quantum field theories \cite{Oe:gbqft}. In particular, given a region $X$ it can be defined only on a dense subspace $\cH_{\partial X}^\ds\subseteq \cH_{\partial X}$ as $\rho_X:\cH_{\partial X}^\ds\to\C$. This has repercussions not only for Axiom~\ref{AmplitudeandTensorIsometryAxiom}, but for all axioms involving the amplitude map. Apart from the mere restriction of the amplitude map, there are also certain compatibility conditions that arise. We enumerate them in the following:
\begin{itemize}
	\item in Axiom~\ref{AmplitudeandTensorIsometryAxiom}: $\tau\pi(\cH_{\partial X_1}^\ds \times \cH_{\partial X_2}^\ds)\subseteq \cH_{\partial X_1\sqcup \partial X_2}^\ds$
	\item in Axiom~\ref{AmplitudeConjugationIsometryAxiom}: $\iota_{\Sigma}(\cH_{\Sigma}^\ds)\subseteq \cH_{\overline{\Sigma}}^\ds$
	\item in Axiom~\ref{SliceRegionsAxiom}: $\gamma_{\alpha}\tau_{\overline{\Sigma},\Sigma}\pi(\cH_{\overline{\Sigma}}\times\cH_{\Sigma})\subseteq \cH_{\partial S}^\ds$
	\item in Axiom~\ref{AmplitudeRelativeGluingAxiom}: $\gamma_{\alpha}\tau(\zeta\tens\iota_{\Sigma}(\zeta)\tens v)\in \cH_{\partial X}^\ds$ for $\zeta\in\cH_{\Sigma}$ and $v\in\cH_{\Lambda}^\ds$.
\end{itemize}

\begin{definition}
	\label{QFTAFinfin}
	 A \emph{CQFT} consists of the axioms of Definition~\ref{QFTAFfin}, modified as laid out in this section.
\end{definition}

\subsection{Comparison of axiomatic systems}
\label{sec:compareaxioms}

In this subsection we explain in more detail the relationship between the axiomatic system of Definition~\ref{QFTAFfin} (or Definition~\ref{QFTAFinfin}) and that of Definition~\ref{AFdefinitionClassic}. In doing so we refer to the axioms in Definition~\ref{AFdefinitionClassic} by the numbering presented there, which differs from the numbering in Definition~\ref{QFTAFinfin} by the combination of numbers and letters, e.g.\ \textbf{(T1)}, \textbf{(T2)}, etc.\ in Definition \ref{AFdefinitionClassic} versus \ref{Statespacesaxiom}, \ref{AssociativityTensorIsometryAxiom}, etc.\ in Definition~\ref{QFTAFinfin}. There exists an approximate correspondence between the axiomatic systems, summarized in Table~\ref{tab:axiomcorr} and explained in detail in the following.

\begin{table}
	\caption{Approximate correspondence between CQFT and GBQFT.}
	\label{tab:axiomcorr}
	\begin{center}
			\begin{tabular}{ |p{6cm}|p{6cm}|  }
				\hline
				GBQFT (Definition~\ref{AFdefinitionClassic}) & CQFT (Definition~\ref{QFTAFfin}) \\
				\hline
				Axiom~\textbf{(T1)} & Axiom~\ref{Statespacesaxiom}\\
				Axiom~\textbf{(T1b)}& Axiom~\ref{ConjugationIsometryAxiom}\\
				Axiom~\textbf{(T2)} & Axioms~\ref{TensorIsometryAxiom}, \ref{GluingIsometryAxiom}, \ref{AssociativityTensorIsometryAxiom}, \ref{AssociativityofGluingIsometryAxiom}, \ref{TensorandGluingIsometryAxiom} \\
				Axiom~\textbf{(T2a)} & Axiom~\ref{TensorIsometryandGradedSymmetryAxiom}\\
				Axiom~\textbf{(T2b)} & Axioms~\ref{ConjugationTensorIsometryAxiom}, \ref{ConjugationGluingIsometryAxiom}\\
				Axiom~\textbf{(T4)} & Axioms~\ref{AmplitudeMapsAxiom} and \ref{AmplitudeConjugationIsometryAxiom}  \\
				Axiom~\textbf{(T5a)} & Axiom~\ref{AmplitudeandTensorIsometryAxiom} \\
				Axiom~\textbf{(T3x)} & Axiom~\ref{SliceRegionsAxiom}\\
				Axiom~\textbf{(T5b)} & Axiom~\ref{AmplitudeRelativeGluingAxiom} \\				
				\hline
			\end{tabular}
	\end{center}
\end{table}

Axioms~\ref{Statespacesaxiom} and \ref{ConjugationIsometryAxiom} are Axioms~\textbf{(T1)} and \textbf{(T1b)}, Axioms~\ref{AmplitudeMapsAxiom} and \ref{AmplitudeConjugationIsometryAxiom} are Axiom~\textbf{(T4)}, Axioms~\ref{TensorIsometryAxiom}, \ref{GluingIsometryAxiom}, \ref{AssociativityTensorIsometryAxiom}, \ref{AssociativityofGluingIsometryAxiom} together form a refined version of Axiom~\textbf{(T2)}. Precisely, given hypersurfaces $\Sigma$, $\Sigma_1$ and $\Sigma_2$, a decomposition of $\Sigma$ as $\Sigma_1\cup\Sigma_2$ as in \textbf{(T2)} means that there exists a gluing function $\alpha:\Sigma_1\sqcup\Sigma_2\to \Sigma$. The composition $\gamma_\alpha\tau_{\Sigma_1,\Sigma_2}$ of the graded epimorphisms provided by Axioms~\ref{TensorIsometryAxiom} and \ref{GluingIsometryAxiom} of Definition~\ref{QFTAFinfin} is a graded epimorphism from $\cH_{\Sigma_1}\otimes\cH_{\Sigma_2}$ to $\cH_\Sigma$, which by Axioms~\ref{AssociativityTensorIsometryAxiom} and \ref{AssociativityofGluingIsometryAxiom} is associative, i.e., the graded epimorphism $\gamma_\alpha\tau_{\Sigma_1,\Sigma_2}$ satisfies the conditions in \textbf{(T2)}. If the graded epimorphism is also compatible with graded symmetry (i.e., if Axiom~\ref{TensorIsometryandGradedSymmetryAxiom} holds) then this means it also satisfies \textbf{(T2a)}.
Axiom~\ref{GluingIsometryAxiom} is more powerful than the self-gluing case of axiom~\textbf{(T2)} as it also allows considering homeomorphisms of hypersurfaces for example. This also means that the new Axiom~\ref{TensorandGluingIsometryAxiom} is required, which states that tensor isomorphisms and gluing epimorphisms should combine in a coherent way. The condition $\gamma_{id_\Sigma}=id_{\cH_\Sigma}$ in Axiom~\ref{GluingIsometryAxiom} is new, and it states that the graded gluing epimorphism of the "trivial" gluing functions, i.e.\ identities, are precisely the "trivial" isomorphisms of state spaces. Similarly, we require homeomorphic gluing functions to induce isomorphisms.
Axiom~\ref{ConjugationTensorIsometryAxiom} is equivalent to Axiom~\textbf{(T2b)}.
Axiom~\ref{SliceRegionsAxiom} is Axiom~\textbf{(T3x)}, where the notion of slice region has been substituted for that introduced in Section~\ref{sliceregionssubsection}. By Theorem~\ref{thm:propip} the isometric properties of maps in Axioms~\textbf{(T1b)} and \textbf{(T2)} follow (supposing Axiom~\ref{TensorIsometryandGradedSymmetryAxiom}).

Axiom~\ref{AmplitudeandTensorIsometryAxiom} is Axiom~\textbf{(T5a)}. Axiom~\ref{AmplitudeRelativeGluingAxiom} is Axiom~\textbf{(T5b)} phrased in the language of relative gluing diagrams, where $\tau$ appearing in equation (\ref{EquationAmplitudeRelativeGluingData}) is the graded tensor isometry associated to $\Sigma\sqcup \overline{\Sigma}\sqcup \Lambda$, i.e., is any of the two corners of the diagram appearing in Axiom~\ref{AssociativityTensorIsometryAxiom} in the case in which $\Sigma_1=\Sigma$, $\Sigma_2=\overline{\Sigma}$ and $\Sigma_3=\Lambda$. The condition $c_{f,f|_{\partial X}}=1$ is new, and it states that the gluing anomaly of "trivial" relative gluing diagrams, i.e., of relative gluing diagrams associated to homeomorphisms should be trivial. This arises due to the novel possibility of encoding homeomorphisms in terms of gluing diagrams.

The paragraphs above are intended to make it clear that the information contained in the axiomatic systems in Definition~\ref{AFdefinitionClassic} and \ref{QFTAFinfin} is largely the same, apart from formal considerations in the concept of gluing and organizational differences. A novelty of the new axiomatic system is that it naturally also encodes spacetime symmetries in the form of homeomorphic gluing functions. We consider Definition~\ref{QFTAFinfin} an ad hoc formalization of Definition~\ref{AFdefinitionClassic}, but shall mostly use the more elegant rigid version, Definition~\ref{QFTAFfin}, when referring to CQFT in the following. In circumstances where the difference between the finite and infinite-dimensional state spaces is important, we shall comment on this explicitly.

\subsection{Admissibility}

As outlined in the introduction, an important ingredient in ensuring that certain models of interest with infinite-dimensional state spaces can be treated in GBQFT is \emph{admissibility}. That is, instead of considering all possible manifolds of dimension $n$ as regions, all possible manifolds of dimension $n-1$ as hypersurfaces etc., we may restrict to certain classes that we shall call \emph{admissible}. Similarly, we may restrict slice regions, gluing functions and gluing diagrams to admissible ones. The axioms thus only apply to admissible structures.

We shall apply precisely the same principle to CQFT. It is rather straightforward to enumerate the conditions that the notion of admissibility has to satisfy to be consistent.
\begin{enumerate}
	\item The empty hypersurface is admissible.
	\item The empty region is admissible.
	\item The boundary of an admissible region is admissible.
	\item The domain and codomain of an admissible gluing function are admissible.
	\item The components of an admissible gluing diagram are admissible.
	\item The components of an admissible slice region are admissible.
	\item For every admissible hypersurface, there is an admissible slice region on it.
	\item The operation of disjoint union preserves admissibility.
	\item The operation of orientation-reversal preserves admissibility.
\end{enumerate}

\subsection{Area-dependent theory}
\label{sec:areaaxioms}

As explained in the introduction, the setting of topological manifolds for the axiomatic framework of CQFT in Section~\ref{sec:axioms} should be considered a bare minimum. For many models of interest we require additional structure that is reflected in a corresponding augmentation of the axioms. For the example to be considered in Section~\ref{sec:2dpqym} we need to augment regions with one additional datum: A non-negative real number, that has the interpretation of an $n$-volume. However, since the main case of interest is in dimension $n=2$, we refer to this as \emph{area}. Thus, a region is now given by a pair $(X,a)$, where $X$ is an $n$-manifold and $a\ge 0$.

We proceed to describe the behavior of this additional datum with respect to all relevant operations:
\begin{description}
    \item[Empty region:] For $(\emptyset,a)$ we have $a=0$.
    \item[Orientation-reversal:] Given a region $(X,a)$ we have $\overline{(X,a)}=(\overline{X},a)$.
    \item[Disjoint union:] Given regions $(X,a)$ and $(X',a')$ we have $(X,a)\sqcup (X',a')=(X\sqcup X',a+a')$.
    \item[Gluing functions:] Given $n$-manifolds $X,X'$ and $a\ge 0$, a gluing function $f:(X,a)\to (X',a)$ is a gluing function $f:X\to X'$ in the sense of Definition~\ref{gluingfunctiondefinition}. There are no gluing functions between regions with different areas.
    \item[Relative gluing diagrams:] No further restrictions are imposed.
    \item[Slice regions:] $((S,a),s)$ is a slice region iff $(S,s)$ is a slice region in the sense of Definition~\ref{sliceregionsdefinition} and if $a=0$. That is, slice regions always have vanishing area.
\end{description}
No further non-trivial changes need to be applied to the axioms of a CQFT.

\begin{definition}
	\label{QFTAFarea}
	 An \emph{area-dependent rigid CQFT} consists of the axioms of Definition~\ref{QFTAFfin}, modified as laid out in this section.
	 An \emph{area-dependent CQFT} consists of the axioms of Definition~\ref{QFTAFinfin}, modified as laid out in this section.
\end{definition}


\section{2-dimensional CQFT}
\label{sec:2dim}

In this section we explore symmetric CQFT in the special case of dimension $n=2$.

\subsection{Elementary objects}

We shall take a generators-and-relations approach. First we identify in the collections of hypersurfaces, regions, and gluing functions a small and finite number of objects such that any other object in these collections can be obtained, up to homeomorphism, through iteration of the operations of orientation-reversal, disjoint union, composition and gluing as laid out in Section~\ref{usualgluings}. We call these \emph{elementary} objects. Axioms~\ref{Statespacesaxiom}--\ref{AmplitudeMapsAxiom} of Definition~\ref{QFTAFfin} assign to these corresponding objects in the category of superspaces. We call these assignments \emph{generators}, because they determine, through the axioms, the assignment for any object in the source category, up to homeomorphism. Finally, Axioms~\ref{AssociativityTensorIsometryAxiom}--\ref{AmplitudeRelativeGluingAxiom} impose constraints on these assignments that we shall refer to as \emph{relations}. At first, we shall not take into account non-trivial automorphisms of manifolds. We treat these separately in Section~\ref{sec:specialhom}.

Any hypersurface can be described as a finite disjoint union of connected hypersurfaces, here, compact connected oriented topological manifolds of dimension 1, possibly with boundary. There are \emph{two} homeomorphism classes of these which we referred to in Section~\ref{gluingfunctionssubsection} as \emph{open} and \emph{closed strings}. In the following we shall prefer to refer to them as \emph{intervals} and \emph{circles}.

\begin{definition}\label{def:standardintervalcircle}
	We fix a standard choice of interval and circle, and we will denote these choices by $I$ and $C$ respectively. For each of the above choices we have an operation of \emph{orientation reversal}, $I\to \overline{I}$ and $C\to\overline{C}$.
\end{definition}

The choice of standard objects brings with it some notational challenges. When we contemplate operations on manifolds all the manifolds involved are to be considered as distinguishable objects. This is no different when some of these manifolds are copies of standard manifolds. In order to distinguish different copies it is sometimes convenient to equip these with extra labels. We shall choose in such cases a numerical subscript. For example, for the disjoint union of two standard intervals $I$ the notation $I\sqcup I$ suggests itself. However, when it is useful to explicitly distinguish the two copies involved we shall write $I_1 \sqcup I_2$ or $I_2\sqcup I_1$. Of course  $I_1 \sqcup I_2$ and $I_2\sqcup I_1$ are identical, but when we want to write this identity explicitly it is $(x,y)\mapsto (y,x)$. On the other hand, the map $(x,y)\mapsto (x,y)$ would represent in this case a non-trivial homeomorphism. When orientation is reversed, it turns out to be convenient to reverse the order in our notation. For example, we set $\overline{I_1\sqcup I_2}=\overline{I_2}\sqcup\overline{I_1}$.

Any gluing function for hypersurfaces can be obtained, up to homeomorphism, by composing just two types of gluing functions: The gluing of two intervals into one as in Example~\ref{exampletwointervals} and the gluing of one interval into a circle as in Example~\ref{ExampleClosedStrings}. We choose a fixed representative for each, acting on the elementary hypersurfaces.

\begin{definition}\label{def:standardproductclosure}
	We fix a standard choice of "product" gluing function as in Example~\ref{exampletwointervals}, and we denote it by $f_{I\sqcup I,I}:I\sqcup I\to I$. We further fix a standard choice of "closure" gluing function as in Example~\ref{ExampleClosedStrings}, and we denote it by $f_{I,C}:I\to C$. We denote the orientation reversals of $f_{I\sqcup I,I}$ and $f_{I,C}$ as $f_{\overline{I\sqcup I},\overline{I}}$ and $f_{\overline{I},\overline{C}}$ respectively. In ordered notation, the 2-inteval gluing is $f_{I_1\sqcup I_2}:I_1\sqcup I_2\to I$. For its orientation reversal we have, $f_{\overline{I_1\sqcup I_2},\overline{I}}: \overline{I_1\sqcup I_2}=\overline{I_2}\sqcup\overline{I_1}\to\overline{I}$.
\end{definition}

We can glue more than two copies of the standard interval $I$ into a single standard interval $I$ by iterating the gluing function $f_{I\sqcup I,I}$ combined with identity maps. In general, there are different ways of doing this that do not coincide. In the simplest case of gluing three intervals together we can relate the two ways of doing so through a homeomorphism $\delta:I\to I$ as follows.

\begin{definition}\label{def:delta}
We fix a standard choice of a homeomorphism $\delta:I\to I$ making the following equation hold,
	\begin{equation}
		f_{I\sqcup I,I}\circ(\id_I\sqcup f_{I\sqcup I,I})=\delta\circ f_{I\sqcup I,I}\circ(f_{I\sqcup I,I}\sqcup \id_I) .
		\label{eq:intglueorder}
	\end{equation}
	We represent $\delta$ diagramatically as filling the square:
	\begin{center}
		\begin{tikzpicture}
			\matrix (m) [matrix of math nodes,row sep=4em,column sep=4em,minimum width=2em]
			{
				I\sqcup I\sqcup I & I\sqcup I \\
				I\sqcup I & I \\};
			\path[-stealth]
			(m-1-1) edge node [left] {$id_I\sqcup f_{I\sqcup I,I}$} (m-2-1)
			edge node [above] {$f_{I\sqcup I,I}\sqcup \id_I$} (m-1-2)
			(m-1-2) edge node [right] {$f_{I\sqcup I,I}$} (m-2-2)
			(m-2-1) edge node [below] {$f_{I\sqcup I,I}$}(m-2-2)
			(m-1-1) edge [color = white] node [color = black] {$\delta$} (m-2-2)
			;
		\end{tikzpicture}
	\end{center}
	The above diagram is not strictly commutative, but is commutative "up to the action of $\delta$". The above notation is suggestive of the fact that $\delta$ should be regarded as the associator isomorphism of a weak monoid in an appropriate monoidal category. In order to handle the gluing of multiple intervals in a canonical way we introduce the following notation that implements a gluing where in each step only the two right-most intervals are glued together,
	\begin{equation}
		f_{I\sqcup \cdots \sqcup I,I}\defeq f_{I\sqcup I,I}\circ (\id_{I}\sqcup f_{I\sqcup I})\circ\cdots\circ (\id_I\sqcup\cdots\sqcup\id_I\sqcup f_{I\sqcup I, I}) .
	\end{equation}
\end{definition}

Any region is an oriented, topological Riemann surface with holes. We pick a standard Riemann surface for each choice of genus $g\in \N_0$ and number of holes $k\in\N_0$.

\begin{definition}\label{def:standardRiemannSurface}
	Let $g\in \N_0$ and $k\in\N_0$. We fix a standard choice of Riemann surface of genus $g$ and with $k$ holes, and denote it by $S_{g,k}$. In particular $S_{0,1}$ will denote the standard disk. We will also write $D$ for $S_{0,1}$.
\end{definition}

We can produce any Riemann surface with holes by gluing a single disk. We thus take the disk $D$ as an elementary region. Moreover, we arrange things in such a way that the boundary of $D$ is precisely the standard circle $C$.

We proceed to consider the assignment of objects in the category of superspaces to the elementary objects via Axioms~\ref{Statespacesaxiom}--\ref{AmplitudeMapsAxiom} of Definition~\ref{QFTAFfin}. By Axiom~\ref{Statespacesaxiom}, we assign to the interval $I$ and the circle $C$ corresponding superspaces $\cH_I$ and $\cH_C$. Since we envisage an infinite-dimensional setting, we consider these as super Krein spaces and require also the dense subspaces $\cH_I^\ds\subseteq\cH_I$ and $\cH_C^\ds\subseteq\cH_C$. Axiom~\ref{ConjugationIsometryAxiom} requires graded involutions $\iota_I:\cH_I\to \cH_{\overline{I}}$ and $\iota_C:\cH_C\to\cH_{\overline{C}}$ associated to orientation-reversal. These must restrict to the corresponding dense subspaces. That, is $\iota_I(\cH_I^\ds)\subseteq \cH_{\overline{I}}^\ds$ and $\iota_C(\cH_C^\ds)\subseteq \cH_{\overline{C}}^\ds$.
On the other hand, we shall consider the graded tensor isomorphisms (Axiom~\ref{TensorIsometryAxiom}) as essentially trivial, implicit structures.
According to Axiom~\ref{GluingIsometryAxiom} we assign graded epimorphisms $\gamma_{I\sqcup I,I}:\cH_I\tens \cH_I \to \cH_I$ and $\gamma_{I,C}:\cH_I\to \cH_C$ to the gluing maps $f_{I\sqcup I,I}:I\sqcup I\to I$ and $f_{I,C}:I\to C$ respectively. (We have simplified the notation in the obvious way and taken the relevant graded tensor isomorphism $\tau$ to be implicit.) Again, restriction to the relevant dense subspaces in the domains must land in the dense subspaces of the images. That is, $\gamma_{I\sqcup I,I}(\cH_I^\ds\tens\cH_I^{\ds})\subseteq\cH_{I}^{\ds}$ and $\gamma_{I,C}(\cH_I^{\ds})\subseteq\cH_C^{\ds}$.
Finally, according to Axiom~\ref{AmplitudeMapsAxiom}, we assign an amplitude map to the standard disk, $\rho_{D}:\cH_C^{\ds}\to\C$. The generators obtained so far are listed in Table~\ref{tab:elemgen}.

\begin{table}
	\caption{Elementary objects and generating assignments}
	\label{tab:elemgen}
	\begin{center}
			\begin{tabular}{|l|l|}
				\hline
				Elementary object & Generating assignment \\
				\hline
                standard interval $I$ & $\cH_I$ \\
				standard circle $C$ & $\cH_C$ \\
                interval orientation reversal & $\iota_I:\cH_I\to\cH_{\overline{I}}$ \\
                circle orientation reversal & $\iota_C:\cH_C\to\cH_{\overline{C}}$ \\
				2-interval gluing $f_{I\sqcup I,I}$ & $\gamma_{I\sqcup I,I}:\cH_I\tens\cH_I\to\cH_I$ \\
				interval self-gluing $f_{I,C}$ & $\gamma_{I,C}:\cH_I\to\cH_C$ \\
				standard disk $D$ & $\rho_{D}:\cH_C^{\ds}\to\C$ \\
				\hline
			\end{tabular}
	\end{center}
\end{table}

\subsection{Special homeomorphisms}
\label{sec:specialhom}

In this section we consider additional homeomorphic gluing functions that we shall add to the elementary objects as they bring out special properties of already considered elementary hypersurfaces and regions.
Concerning the standard interval $I$ we note that there are orientation preserving homeomorphisms that map $I$ to $\overline{I}$. For example, if we identify $I$ with the unit interval $[0,1]$, then the homeomorphism $t\mapsto (1-t)$ has this property.

\begin{definition}\label{def:twistinghomeomorphism}
	We fix a standard choice of orientation-preserving "twisting homeomorphism" $f_{I,\overline{I}}:I \to \overline{I}$ of the standard interval $I$. We denote the orientation reversal $\overline{f_{I,\overline{I}}}$ of $f_{I,\overline{I}}$ by $f_{\overline{I},I}$. We choose $f_{\overline{I},I}$ in such a way that $f_{\overline{I},I}\circ f_{\overline{I},I}=\id_I$. 
Additionally, we impose as a consistency condition on $f_{I\sqcup I,I}$ the commutativity of the following diagram:
\begin{equation}
	\begin{tikzpicture}
		\matrix (m) [matrix of math nodes,row sep=3em,column sep=3em,minimum width=2em]
		{ & I_1 \sqcup I_2 & \\
		 \overline{I_1}\sqcup\overline{I_2} && I \\
		 & \overline{I} &\\};
		\path[-stealth]
		(m-1-2) edge node [left] {$f_{I_1,\overline{I_1}}\sqcup f_{I_2,\overline{I_2}}$} (m-2-1)
		edge node [right]  {$f_{I_1\sqcup I_2,I}$} (m-2-3)
		(m-2-1) edge node [below]  {$f_{\overline{I_2\sqcup I_1},\overline{I}}$} (m-3-2)
		(m-2-3) edge node [below] {$f_{I,\overline{I}}$} (m-3-2);
	\end{tikzpicture}
	\label{diag:doubleinttwist}
\end{equation}
\end{definition}

We proceed similarly for the circle. If we identify the standard circle $C$ with $\{z\in \C: |z|=1\}$ then, $z\mapsto \overline{z}$ is an example of an oreintation-preserving homeomorphism $C\to\overline{C}$.
\begin{definition}\label{def:twistcircle}
	We fix a standard choice of orientation-preserving "twisting homeomorphism" on $C$, $f_{C,\overline{C}}:C\to\overline{C}$. We denote the orientation reversed version of $f_{C,\overline{C}}:C\to\overline{C}$, which we will assume is also its inverse, by $f_{\overline{C},C}$. Moreover, we require $f_{I,C}$ to satisfy
	\begin{equation}
		f_{\overline{I},\overline{C}}\circ f_{I,\overline{I}}=f_{C,\overline{C}}\circ f_{I,C} .
		\label{eq:itoctwistcompat}
	\end{equation}
\end{definition}

We also consider a half rotation on the standard circle $C$, e.g.\ via the homeomorphism $z\mapsto -z$.
\begin{definition}\label{def:halfrotation}
	We write $f_\pi:C\to C$ for a standard choice of a half-rotation on $C$. This is a homeomorphism of $C$ satisfying $f_\pi^2=\id_C$. Moreover, we require $f_\pi$ to make the following diagram commute:
	\begin{equation}
		\begin{tikzpicture}
			\matrix (m) [matrix of math nodes,row sep=3em,column sep=4em,minimum width=2em]
			{I_1 \sqcup I_2 & I_2 \sqcup I_1 \\
				I & I\\
				C & C\\};
			\path[-stealth]
			(m-1-1) edge node [above]  {$\sigma$} (m-1-2)
			edge node [left] {$f_{I_1\sqcup I_2,I}$} (m-2-1)
			(m-1-2) edge node [right] {$f_{I_2\sqcup I_1,I}$} (m-2-2)
			(m-2-1) edge node [left] {$f_{I,C}$} (m-3-1)
			(m-2-2) edge node [right] {$f_{I,C}$} (m-3-2)
			(m-3-1) edge node [below]  {$f_{\pi}$} (m-3-2);
		\end{tikzpicture}
		\label{diag:doubleintcirc}
	\end{equation}
	Here, $I_1$ and $I_2$ are two copies of the standard interval $I$ and $\sigma$ denotes the identity map that interchanges the components of the disjoint union.
	We also require the orientation reversed version of $f_{\pi}$ to be compatible with the circle twist map. That is,
	\begin{equation}
		f_{\overline{\pi}}\circ f_{C,\overline{C}}=f_{C,\overline{C}}\circ f_{\pi} .
		\label{eq:rottwist}
	\end{equation}
\end{definition} 

There is an orientation preserving homeomorphism that maps the disk to its orientation reversed version. Again, we can choose such a homeomorphism for the standard disk $D$ such that it is inverse to its orientation reversed version. For example, identifying $D$ with $\{z\in\C: |z|\le |\}$ we may consider $z\mapsto \overline{z}$.

\begin{definition}\label{def:twistdisk}
	We fix a standard choice of "twisting homeomorphism" of the standard disk $D$, and we denote it by $f_{D,\overline{D}}:D\to \overline{D}$. We write $f_{\overline{D},D}:\overline{D}\to D$ for the orientation reversal of $f_{D,\overline{D}}$. We require compatibility with $f_{C,\overline{C}}$ in the sense that the latter is the restriction of $f_{\overline{D},D}:\overline{D}\to D$ to the boundary. 
\end{definition}

According to Axiom~\ref{GluingIsometryAxiom} we assign graded epimorphisms $\gamma_{I,\overline{I}}:\cH_I\to\cH_{\overline{I}}$, $\gamma_{C,\overline{C}}:\cH_C\to\cH_{\overline{C}}$ and $\gamma_{\pi}:\cH_C \to \cH_C$ to the gluing maps $f_{I,\overline{I}}$, $f_{C,\overline{C}}$ and $f_{\pi}$ respectively.
Restriction to the relevant dense subspaces in the domains must land in the dense subspaces of the images. That is, $\gamma_{I,\overline{I}}(\cH_I^{\ds})\subseteq \cH_{\overline{I}}^{\ds}$, $\gamma_{C,\overline{C}}(\cH_C^{\ds})\subseteq \cH_{\overline{C}}^{\ds}$ and $\gamma_{\pi}(\cH_C^\ds)\subseteq\cH_{C}^{\ds}$. The additional special homeomorphisms of the present section and their generating assignments are summarized in Table~\ref{tab:elemhomgen}.

\begin{table}
	\caption{Special homeomorphisms and generating assignments}
	\label{tab:elemhomgen}
	\begin{center}
		\begin{tabular}{|l|l|}
			\hline
			Special homeomorphism & Generating assignment \\
			\hline
			interval twist $f_{I,\overline{I}}$ & $\gamma_{I,\overline{I}}:\cH_{I}\to\cH_{\overline{I}}$ \\
			circle twist $f_{C,\overline{C}}$ & $\gamma_{C,\overline{C}}:\cH_{C}\to\cH_{\overline{C}}$ \\
			half rotation $f_{\pi}$ & $\gamma_{\pi}:\cH_{C}\to\cH_{C}$ \\
			disk twist $f_{D,\overline{D}}$ &  \\
			\hline
		\end{tabular}
	\end{center}
\end{table}

\subsection{Relations}
\label{sec:2rel}

We proceed to consider the most important relations that Axioms~\ref{AssociativityTensorIsometryAxiom}--\ref{AmplitudeRelativeGluingAxiom} impose on the generating assignments. Note that we shall use Axiom~\ref{AssociativityTensorIsometryAxiom} implicitly and throughout in accordance with our convention in this section of absorbing the tensor isomorphisms $\tau$ into the gluing epimorphisms $\gamma$. The first piece of structure we consider is the gluing transformation $\gamma_{I\sqcup I, I}$ associated to $f_{I\sqcup I,I}$. 

\begin{lem}\label{lem:algebra}
	The map $\gamma_{I\sqcup I,I}:\cH_I\tens\cH_I\to\cH_I$ is a binary operation, associative up to the isomorphism $\gamma_\delta$. Precisely, the gluing transformations $\gamma_{I\sqcup I,I}$ and $\gamma_\delta$ satisfy the equation:
	\begin{equation}
		\gamma_{I\sqcup I,I}\circ (\id_I\tens\gamma_{I\sqcup I,I})=\gamma_\delta
		\circ \gamma_{I\sqcup I,I}\circ (\gamma_{I\sqcup I,I}\tens\id_I) .
		\label{eq:associator}
	\end{equation}
        Moreover, the gluing transformations $\gamma_{\overline{I_1 \sqcup I_2},\overline{I}}$ associated to two copies $I_1$ and $I_2$ of the interval $I$, $\gamma_{I_1\sqcup I_2,I}$, and the isomorphism $\iota_{I}$ satisfy the equation:
	\begin{equation}
		\gamma_{\overline{I_1\sqcup I_2},\overline{I}}(\iota_I(\psi)\tens\iota_I(\eta))
		=\iota_I(\gamma_{I_1\sqcup I_2,I}(\eta\tens\psi))
		\label{eq:doubleintcopcompat}
	\end{equation}
	for every $\eta\in\cH_{I_1}$ and $\psi\in\cH_{I_2}$.
\end{lem}
\begin{proof}
	Equation~(\ref{eq:associator}) follows from Axiom~\ref{AssociativityofGluingIsometryAxiom} and equation~(\ref{eq:intglueorder}), and equation (\ref{eq:doubleintcopcompat}) follows from Axioms~\ref{ConjugationTensorIsometryAxiom} and \ref{ConjugationGluingIsometryAxiom}.
\end{proof}

\begin{obs}\label{obs:algebra}
	We represent equation~(\ref{eq:associator}) pictorially by the commutativity of the square up to a filler isomorphism $\gamma_\delta$,
	\begin{center}
		\begin{tikzpicture}
			\matrix (m) [matrix of math nodes,row sep=4em,column sep=4em,minimum width=2em]
			{
				\cH_I^{\otimes 3} & \cH_I^{\otimes 2} \\
				\cH_I^{\otimes 2} & \cH_I \\};
			\path[-stealth]
			(m-1-1) edge node [left] {$1\otimes \gamma_{I\sqcup I,I}$} (m-2-1)
			edge node [above] {$\gamma_{I\sqcup I,I}\otimes 1$} (m-1-2)
			(m-1-2) edge node [right] {$\gamma_{I\sqcup I,I}$} (m-2-2)
			(m-2-1) edge node [below] {$\gamma_{I\sqcup I,I}$}(m-2-2)
			(m-1-1) edge [color = white] node [color = black] {$\gamma_\delta$} (m-2-2)
				;
		\end{tikzpicture}
	\end{center}
	The convention on the above square "filled" by $\gamma_\delta$ is the same as in Definition~\ref{def:delta}. Equation~(\ref{eq:doubleintcopcompat}) is equivalent to commutativity of the diagram:
	\begin{center}
		\begin{tikzpicture}
			\matrix (m) [matrix of math nodes,row sep=4em,column sep=4em,minimum width=2em]
			{
				\cH_I^{\otimes 2} & \cH_I \\
				\cH_{\overline{I}}^{\otimes 2} & \cH_{\overline{I}} \\};
			\path[-stealth]
			(m-1-1) edge node [left] {$\iota_{I}\otimes\iota_{I}$} (m-2-1)
			edge node [above] {$\gamma_{I\sqcup I,I}$} (m-1-2)
			(m-1-2) edge node [right] {$\iota_I$} (m-2-2)
			(m-2-1) edge node [below] {$\gamma_{\overline{I_1\sqcup I_2},\overline{I}}$}(m-2-2)
			;
		\end{tikzpicture}
	\end{center}
\end{obs}

We now consider the gluing transformation $\gamma_{I,C}$ associated to our standard choice of "closing an interval" gluing function $f_{I,C}:I\to C$. We also consider the graded isomorphism $\gamma_\pi$ associated to our choice of half-rotation homeomorphism $f_\pi$. 
\begin{lem}\label{obs:productandrotationandiota}
	The transformations $\gamma_{\overline{I},\overline{C}}$, $\gamma_{I,C}$, $\iota_I$, and $\iota_C$ satisfy the equation:
	\begin{equation}
		\gamma_{\overline{I},\overline{C}}\circ \iota_I=\iota_C\circ \gamma_{I,C} .
		\label{eq:itocopcompat}
	\end{equation}
	The transformations $\iota_C$ and $\gamma_\pi$ satisfy the equation:
	\begin{equation}
		\gamma_{\overline{\pi}}\circ\iota_C=\iota_{\overline{C}}\circ \gamma_{\pi}.
		\label{eq:rotinv}
	\end{equation}
\end{lem}
\begin{proof}
	Equations~(\ref{eq:itocopcompat}) and (\ref{eq:rotinv}) follow with Axiom~\ref{ConjugationGluingIsometryAxiom}.
\end{proof}  

\begin{obs}\label{obs:productsandtorationsiota2}
	Equation~(\ref{eq:itocopcompat}) is equivalent to the commutativity of the square:
	\begin{center}
		\begin{tikzpicture}
			\matrix (m) [matrix of math nodes,row sep=4em,column sep=4em,minimum width=2em]
			{
				\cH_I & \cH_C \\
				\cH_{\overline{I}} & \cH_{\overline{C}} \\};
			\path[-stealth]
			(m-1-1) edge node [left] {$\iota_{I}$} (m-2-1)
			edge node [above] {$\gamma_{I,C}$} (m-1-2)
			(m-1-2) edge node [right] {$\iota_C$} (m-2-2)
			(m-2-1) edge node [below] {$\gamma_{\overline{I},\overline{C}}$}(m-2-2)
			;
		\end{tikzpicture}
	\end{center}
	Equation~(\ref{eq:rotinv}) is equivalent to the commutativity of the square:
	\begin{center}
		\begin{center}
			\begin{tikzpicture}
				\matrix (m) [matrix of math nodes,row sep=4em,column sep=4em,minimum width=2em]
				{
					\cH_C & \cH_C \\
					\cH_{\overline{C}} & \cH_{\overline{C}} \\};
				\path[-stealth]
				(m-1-1) edge node [left] {$\iota_{C}$} (m-2-1)
				edge node [above] {$\gamma_\pi$} (m-1-2)
				(m-1-2) edge node [right] {$\iota_{\overline{C}}$} (m-2-2)
				(m-2-1) edge node [below] {$\gamma_{\overline{\pi}}$}(m-2-2)
				;
			\end{tikzpicture}
		\end{center}
	\end{center}
\end{obs}

Next, it turns out to be useful to combine the twist maps with orientation reversal, both for the interval and for the circle.

\begin{notation}
	We denote the anti-linear composite maps $\iota_{\overline{I}}\circ \gamma_{I,\overline{I}}:\cH_I\to\cH_I$ and $\iota_{\overline{C}}\circ\gamma_{C,\overline{C}}:\cH_C\to\cH_C$ by $\alpha_I$ and $\alpha_C$ respectively.   
\end{notation}

\begin{lem}\label{lem:RealStructures}
	$\alpha_I$ and $\alpha_C$ are anti-linear involutions, and thus define \emph{real structures} on $\cH_I$ and $\cH_C$ respectively. Moreover, $\alpha_{\overline{I}}=\iota_I\circ\alpha_I\circ\iota_{\overline{I}}$ and $\alpha_{\overline{C}}=\iota_C\circ\alpha_C\circ\iota_{\overline{C}}$, define real structures on $\cH_{\overline{I}}$ and $\cH_{\overline{C}}$.
\end{lem}
\begin{proof}
	The claim for $\alpha_I$ follows from the fact that the inverse of $\alpha_I$ is the anti-linear isomorphism $\gamma_{\overline{I},I}\circ \iota_{I}$, which by Axiom~\ref{ConjugationGluingIsometryAxiom}, is equal to $\alpha_I$. The claim for $\alpha_C$ follows from an analogous argument. The rest of the statement follows from Lemma \ref{obs:productandrotationandiota}.
\end{proof}

\begin{obs}\label{obs:realstructures}
	The real structures of Lemma \ref{lem:RealStructures} preserve dense subspaces.
\end{obs} 

We now consider relations between the real structures $\alpha_I$ and $\alpha_C$, and the structure maps considered above.
\begin{lem}\label{lem:compatibilityalphagammas}
	$\alpha_I$ and $\gamma_{I\sqcup I,I}$ are compatible, in the sense of the following equation,
	\begin{equation}
		\gamma_{I\sqcup I,I}(\alpha_I(\psi)\tens\alpha_I(\eta))=\alpha_I(\gamma_{I\sqcup I,I}(\eta\tens\psi)) ,
		\label{eq:realanticom}
	\end{equation}
	for every $\psi,\eta\in\cH_I$. Moreover,
	$\alpha_I$, $\alpha_C$ and $\gamma_{I,C}$ are compatible in the sense of the following equation:
	\begin{equation}
		\alpha_C\circ \gamma_{I,C}=\gamma_{I,C}\circ \alpha_I .
		\label{eq:realitoc}
	\end{equation}
	Finally, $\gamma_\pi$ and $\alpha_C$ commute, i.e., the following equation holds:
	\begin{equation}\label{eq:compatibilitygammapi}
		\gamma_\pi\circ\alpha_C=\alpha_C\circ\gamma_{\pi} .
	\end{equation}
\end{lem}
\begin{proof}
	Equation (\ref{eq:realanticom}) follows from Axiom~\ref{AssociativityofGluingIsometryAxiom}, applied to (\ref{diag:doubleinttwist}), together with (\ref{eq:doubleintcopcompat}). Equation~(\ref{eq:realitoc}) follows by applying Axiom~\ref{AssociativityofGluingIsometryAxiom} to (\ref{eq:itoctwistcompat}), together with (\ref{eq:itocopcompat}). Finally, combining Axiom~\ref{AssociativityofGluingIsometryAxiom} applied to (\ref{eq:rottwist}) with (\ref{eq:rotinv}) we obtain Equation~(\ref{eq:compatibilitygammapi}).
\end{proof}

\begin{obs}\label{obs:compatibilityalphasgammas}
	Equation~(\ref{eq:realanticom}) is equivalent to the commutativity of the square:
	\begin{center}
		\begin{tikzpicture}
			\matrix (m) [matrix of math nodes,row sep=4em,column sep=4em,minimum width=2em]
			{
				\cH_I^{\otimes 2} & \cH_I \\
				\cH_{I}^{\otimes 2} & \cH_{I} \\};
			\path[-stealth]
			(m-1-1) edge node [left] {$\alpha_{I}^{\otimes 2}$} (m-2-1)
			edge node [above] {$\gamma_{I\sqcup I,I}$} (m-1-2)
			(m-1-2) edge node [right] {$\alpha_{I}$} (m-2-2)
			(m-2-1) edge node [below] {$\gamma_{I\sqcup I,I}$}(m-2-2)
			;
		\end{tikzpicture}
	\end{center}
	Equation~(\ref{eq:realitoc}) is equivalent to the commutativity of the square:
	\begin{center}
		\begin{tikzpicture}
			\matrix (m) [matrix of math nodes,row sep=4em,column sep=4em,minimum width=2em]
			{
				\cH_I & \cH_C \\
				\cH_{I} & \cH_{C} \\};
			\path[-stealth]
			(m-1-1) edge node [left] {$\alpha_{I}$} (m-2-1)
			edge node [above] {$\gamma_{I,C}$} (m-1-2)
			(m-1-2) edge node [right] {$\alpha_{C}$} (m-2-2)
			(m-2-1) edge node [below] {$\gamma_{I,C}$}(m-2-2)
			;
		\end{tikzpicture}
	\end{center}
	Equation~(\ref{eq:compatibilitygammapi}) is equivalent to the commutativity of the following square:
	\begin{center}
		\begin{tikzpicture}
			\matrix (m) [matrix of math nodes,row sep=4em,column sep=4em,minimum width=2em]
			{
				\cH_C & \cH_C \\
				\cH_C & \cH_C \\};
			\path[-stealth]
			(m-1-1) edge node [left] {$\alpha_{C}$} (m-2-1)
			edge node [above] {$\gamma_\pi$} (m-1-2)
			(m-1-2) edge node [right] {$\alpha_{C}$} (m-2-2)
			(m-2-1) edge node [below] {$\gamma_{\pi}$}(m-2-2)
			;
		\end{tikzpicture}
	\end{center}
\end{obs}

We now consider relations between the binary product operation $\gamma_{I\sqcup I,I}$, the "closing an interval" transformation $\gamma_{I,C}$, the "half rotation" transformation $\gamma_\pi$, and the graded symmetry operation of Krein spaces.
\begin{lem}\label{lem:symmetryproduct}
	The transformations $\gamma_{I\sqcup I,I}$, $\gamma_{I,C}$, and $\gamma_\pi$ satisfy the equation:
	\begin{equation}
		\gamma_\pi(\gamma_{I,C}(\gamma_{I\sqcup I,I}(\psi\tens\eta)))
		=(-1)^{|\psi| |\eta|}\gamma_{I,C}(\gamma_{I\sqcup I,I}(\eta\tens\psi)) .
		\label{eq:picomc}
	\end{equation}
	
	\noindent for every $\psi,\eta\in\cH_I$. 
\end{lem}
\begin{proof}
	The lemma follows directly from diagram~(\ref{diag:doubleintcirc}), together with Axiom~\ref{AssociativityofGluingIsometryAxiom} and Axiom~\ref{TensorIsometryandGradedSymmetryAxiom}.
\end{proof}

\begin{obs}\label{obs:symmetryproduct}
	Equation (\ref{eq:picomc}) says that the "closure" transformation $\gamma_{I,C}$ is a trace for the binary operation $\gamma_{I\sqcup I,I}$, up to the "half rotation" isomorphism $\gamma_\pi$. More precisely, for two copies $I_1$ and $I_2$ of the standard interval $I$, and $s:\cH_{I_1} \otimes \cH_{I_2}\to \cH_{I_2} \otimes \cH_{I_1}$ the graded symmetry isomorphism of Krein spaces, equation (\ref{eq:picomc}) is equivalent to the commutativity of the diagram:
	\begin{center}
		\begin{tikzpicture}
			\matrix (m) [matrix of math nodes,row sep=3em,column sep=4em,minimum width=2em]
			{\cH_{I_1} \otimes \cH_{I_2} & \cH_{I_2} \otimes \cH_{I_1} \\
				\cH_I & \cH_I\\
				\cH_C & \cH_C\\};
			\path[-stealth]
			(m-1-1) edge node [above]  {$s$} (m-1-2)
			edge node [left] {$\gamma_{I_1\sqcup I_2,I}$} (m-2-1)
			(m-1-2) edge node [right] {$\gamma_{I_2\sqcup I_1,I}$} (m-2-2)
			(m-2-1) edge node [left] {$\gamma_{I,C}$} (m-3-1)
			(m-2-2) edge node [right] {$\gamma_{I,C}$} (m-3-2)
			(m-3-1) edge node [below]  {$\gamma_{\pi}$} (m-3-2);
		\end{tikzpicture}
	\end{center}
\end{obs}

We now consider relations of the structure transformations described above and amplitude maps on our choice of standard disk.
\begin{lem}\label{lem:damplinvol}
	The following equations hold for every $\psi\in\cH_C$:
	\begin{align}
		\rho_{\overline{D}}(\iota_C(\psi)) & =\overline{\rho_{D}(\psi)},
		\label{eq:amplconj}\\
		\text{and}\qquad
		\rho_{\overline{D}}(\gamma_{C,\overline{C}}(\psi)) & =\rho_{D}(\psi) .
		\label{eq:amplflip}
	\end{align}
	\noindent Moreover, from this, we obtain the following equation for every $\psi\in \cH_I$:
	\begin{equation}
		\rho_{D}(\alpha_C(\psi))=\overline{\rho_{D}(\psi)} .
		\label{eq:amplcc}
	\end{equation}
\end{lem}
\begin{proof}
	Equations~(\ref{eq:amplconj}) and (\ref{eq:amplflip}) follow directly from Axioms~\ref{AmplitudeandTensorIsometryAxiom} and \ref{AmplitudeConjugationIsometryAxiom}.
\end{proof}

\begin{lem}\label{lem:damplcomp}
	The following equation holds:
	\begin{equation}
		\rho_{D}\circ\gamma_\pi=\rho_{D} .
		\label{eq:amplrot}
	\end{equation}
	Further, $\rho_D$ satisfies the following equation for every $\psi,\eta\in\cH_I$.:
	\begin{equation}
		\rho_{D}(\gamma_{I,C}(\gamma_{I\sqcup I,I}(\psi\tens\eta)))
		=(-1)^{|\psi| |\eta|}\rho_{D}(\gamma_{I,C}(\gamma_{I\sqcup I,I}(\eta\tens\psi))) .
		\label{eq:damplgradcom}
    \end{equation}
\end{lem}
\begin{proof}
	There exists a homeomorphism of the standard disk $D$ that restricts to $f_\pi$ on the boundary of $D$. By Axiom~\ref{AmplitudeConjugationIsometryAxiom} this implies equation~(\ref{eq:amplrot}). This, applied to identity~(\ref{eq:picomc}) yields equation~(\ref{eq:damplgradcom}).
\end{proof}

\begin{prop}\label{prop:trip}
	The inner product of $\cH_I$ is given by,
	\begin{equation}
		\langle\psi,\eta\rangle_I=\rho_{D}(\gamma_{I,C}(\gamma_{I\sqcup I,I}(\alpha_I(\psi)\tens\eta))) .
		\label{eq:intip}
	\end{equation}
	This satisfies,
	\begin{equation}
		\langle\alpha_I(\psi),\alpha_I(\eta)\rangle_I =(-1)^{|\psi| |\eta|} \langle \eta,\psi\rangle_I = (-1)^{|\psi| |\eta|} \overline{\langle \psi,\eta\rangle_I} .
		\label{eq:realipint}
	\end{equation}
	In particular, with respect to the decomposition of $\cH_I$ of Observation~\ref{obs:coev} we have, $\alpha_I(\cH_I^{0,\pm})=\cH_I^{0,\pm}$ and $\alpha_I(\cH_I^{1,\pm})=\cH_I^{1,\mp}$. As a consequence, if $\{\zeta_k\}_{k\in I}$ is an orthonormal basis of $\cH_I$, then $\{\alpha_I(\zeta_k)\}_{k\in I}$ is an orthonormal basis of $\cH_I$ but with $|\alpha_I(\zeta_k)|=|\zeta_k|$ and $[\alpha_I(\zeta_k)]=[\zeta_k] + |\zeta_k|$.
\end{prop}
\begin{proof}	
We start with the notion of evaluation as given by Axiom~\ref{SliceRegionsAxiom}. Consider first a slice region associated to the standard interval hypersurface $I$. Topologically, this slice region is a disk, which we identify with the standard disk $D$. The boundary of the disk $D$ is the circle $C$, and we can decompose this circle into two intervals. These two intervals are a copy of the original interval $I$ and its orientation reversed version $\overline{I}$. In order to construct the necessary gluing function $\overline{I}\sqcup I\to C$ from elementary objects, we compose the 2-interval gluing $f_{I\sqcup I,I}$ with the interval twist $f_{\overline{I},I}$. That is, the required gluing function is $f_{I\sqcup I,I}\circ(f_{\overline{I},I}\sqcup\id_I)$. In terms of assignments, we get $\ev_I:\cH_{\overline{I}}\tens\cH_I\to\C$ as,
\begin{equation}
	\ev_I(\psi\tens\eta)=\rho_{D}(\gamma_{I,C}(\gamma_{I\sqcup I,I}(\gamma_{\overline{I},I}(\psi)\tens\eta))) .
\end{equation}
The inner product on $\cH_I$ is then obtained from (\ref{eq:propip}), yielding (\ref{eq:intip}). Relation (\ref{eq:realipint}) follows with relation (\ref{eq:damplgradcom}).
\end{proof}

\begin{prop}
	The inner product of $\cH_C$ admits the following expression,
	\begin{equation}
		\langle \gamma_{I,C}(v), \gamma_{I,C}(v')\rangle_C=c_{\mathrm{Cyl}}\sum_{k\in I}(-1)^{[\zeta_k]+|\zeta_k| |v|}\rho_{D}\circ\gamma_{I,C}\circ\gamma_{I\sqcup I\sqcup I\sqcup I,I}(\zeta_k\tens \alpha_I(v)\tens \alpha_I(\zeta_k)\tens v') .
		\label{eq:ipcircle}
	\end{equation}
    Here $v,v'\in\cH_I$ and $c_{\mathrm{Cyl}}$ is the gluing anomaly of the relative gluing diagram of Example~\ref{mainexamplerelativegluing}.
	
	In addition, the real structure $\alpha_C$ is a conjugate linear graded isometry. That is,
	\begin{equation}
		\langle\alpha_C(\psi),\alpha_C(\eta)\rangle_C =(-1)^{|\psi| |\eta|} \langle \eta,\psi\rangle_C = (-1)^{|\psi| |\eta|} \overline{\langle \psi,\eta\rangle_C} .
		\label{eq:realipcircle}
	\end{equation}
\end{prop}
\begin{proof}
	We first work out the amplitude of the cylinder, which is the slice region for the circle. To this end we produce the cylinder by gluing a disk to itself following Example~\ref{mainexamplerelativegluing}. The challenge is here to construct the maps $\alpha: I_1\sqcup \overline{I_2} \sqcup \overline{I_3}\sqcup I_4\to\ C$ and $\beta: \overline{I_3}\sqcup I_4\to \overline{C_3} \sqcup C_4$ in terms of elementary standard maps. For $\beta$ this is straightforward, it is simply $\beta=f_{\overline{I_3},\overline{C_3}} \sqcup f_{I_4,C_4}$. The map $\alpha$ can be constructed as follows,
\begin{equation}
	\alpha=f_{I,C}\circ f_{I_1\sqcup I_3\sqcup I_2\sqcup I_4,I}
	 \circ (\id_1\sqcup \sigma_{2 3}\sqcup \id_4)
	 \circ (\id_1\sqcup f_{\overline{I_2},I_2} \sqcup f_{\overline{I_3},I_3} \sqcup \id_4) .
\end{equation}
Applying Axiom~\ref{AmplitudeRelativeGluingAxiom} in the form of identity (\ref{EquationAmplitudeRelativeGluingData}) to $(f,\beta)$ and $\iota_{\overline{C}}(v)\tens v'$ we obtain for the cylinder amplitude,
\begin{equation}
	\rho_{\mathrm{Cyl}}(\iota_C(\gamma_{I,C}(v))\tens \gamma_{I,C}(v'))=c_{\mathrm{Cyl}}\sum_{k\in I}(-1)^{[\zeta_k]+|\zeta_k| |v|}\rho_{D}\circ\gamma_{I,C}\circ\gamma_{I\sqcup I\sqcup I\sqcup I,I}(\zeta_k\tens \alpha_I(v)\tens \alpha_I(\zeta_k)\tens v') .
\end{equation}
By Axiom~\ref{SliceRegionsAxiom} the left-hand side can be read as an evaluation map, which by equation (\ref{eq:propip}) in turn yields the inner product given precisely by equation (\ref{eq:ipcircle}).
Finally, using (iterations of) relation (\ref{eq:associator}) with Axioms~\ref{AssociativityofGluingIsometryAxiom} and \ref{AmplitudeEquivarianceAxiom} as well as relation (\ref{eq:damplgradcom}) we obtain the graded isometry property (\ref{eq:realipcircle}).
\end{proof}

\begin{prop}\label{prop:diskamplautoglue}
The following relation is satisfied by the disk amplitude,
\begin{equation}
	\rho_{D}(\gamma_{I,C}(v))=c \sum_{k\in I}(-1)^{[\zeta_k]}\rho_{D}\circ\gamma_{I,C}\circ\gamma_{I\sqcup I\sqcup I,I}(\zeta_k\tens\alpha_I(\zeta_k)\tens v) .
	\label{eq:diskamplautoglue}
\end{equation}
Here, $v\in\cH_I$ and $c$ is the gluing anomaly of the relative gluing diagram of Example~\ref{diskondiskrelativegluing}.
\end{prop}
\begin{proof}
	Consider gluing the standard disk to itself as in Example~\ref{diskondiskrelativegluing}. Take the gluing map $\beta:I\to C$ to be $f_{I,C}$ and construct the gluing map $\alpha$ out of standard gluing maps as follows,
\begin{equation}
	\alpha= f_{I,C}\circ f_{I_1\sqcup I_2 \sqcup I_3}\circ (\id_{I_1}\sqcup f_{\overline{I_2},I_2}\sqcup \id_{I_3}).
\end{equation}
Applying Axiom~\ref{AmplitudeRelativeGluingAxiom} in the form of identity (\ref{EquationAmplitudeRelativeGluingData}) to $(f,\beta)$ we obtain relation (\ref{eq:diskamplautoglue}).
\end{proof}

\subsection{Associative algebra}
\label{sec:assocalg}

In this section we shall make two additional assumptions on CQFTs, \emph{strictness} and \emph{anomaly-freedom}. As a consequence of these we obtain general structural results on the state spaces associated to the open and closed strings $I$ and $C$. 

\begin{definition}
	A CQFT is called \emph{strict}, if for every homeomorphism $f:\Sigma\to\Sigma$ of a hypersurface $\Sigma$ to itself, the gluing isomorphism $\gamma_f$ is the identity $\id_{\cH_{\Sigma}}$.
\end{definition}

Under the assumption of strictness, the automorphism $\gamma_\delta:\cH_I\to\cH_I$ associated to the self-homeomorphism $\delta:I\to I$ of the standard interval (recall equation (\ref{eq:intglueorder})) is the identity. This has the notable consequence that the algebra structure on $\cH_I$ induced by the epimorphism $\gamma_{I\sqcup I,I}$ is \emph{associative}, compare the associator equation (\ref{eq:associator}). Another affected homeomorphism is the half-rotation of the circle $f_{\pi}:C\to C$ with the associated isomorphism $\gamma_{\pi}:\cH_C\to\cH_C$ becoming the identity.

\begin{definition}
	A CQFT is called \emph{anomaly-free}, if the anomaly factor associated to every gluing diagram is equal to $1$. 
\end{definition}

We assume all CQFTs in this subsection are 2-dimensional, strict and anomaly-free. We adopt the following notation, which clarifies the algebra structure defined in Section~\ref{sec:2rel}.

\begin{notation}\label{notation2dim}
	Let $\Sigma$ be a hypersurface. We write,
	\begin{enumerate}
		\item $\psi\bullet\eta\defeq\gamma_{I\sqcup I, I}(\psi\tens\eta)$ for every  $\psi,\eta\in\cH_I$.
		\item $\psi^\star\defeq\alpha_I(\psi)$ for every $\psi\in\cH_I$.
		\item $\ptr(\psi)\defeq\gamma_{I,C}(\psi)$ for every $\psi\in\cH_I$.
		\item $\psi^\star\defeq\alpha_C(\psi)$ for every $\psi\in\cH_C$.
		\item $\tr(\psi)\defeq\rho_D\gamma_{I,C}(\psi)$ for every $\psi\in\cH_I$ 
	\end{enumerate}
\end{notation}
The product $\bullet$ is thus an \emph{associative}, and a priori, non-unital binary operation on $\cH_I$, and the  $\star$-operation is a conjugate linear involution on $\cH_I$. Lemma~\ref{lem:compatibilityalphagammas} relates these two operations in a clean way and takes the following form in the present setting. (Note also that the last statement of the Lemma, equation~(\ref{eq:compatibilitygammapi}), becomes trivial due to strictness.)
\begin{obs}\label{obs:ProductStarStrict}
	The following equation holds for every $\psi,\eta\in\cH_I$:
	\begin{equation}\label{eq:starprod}
		(\psi\bullet\eta)^\star=\eta^\star\bullet\psi^\star.
	\end{equation}
Compare equation (\ref{eq:realanticom}). The operations $\bullet$ and $\star$ thus provide $\cH_I$ with the structure of an associative $\star$-algebra. Moreover, the $\star$-operation is compatible with the map $\ptr$,
\begin{equation}
   \ptr(\psi^\star)=\ptr(\psi)^\star \quad\forall\psi\in\cH_I .
   \label{eq:ptrstar}
\end{equation}
Compare equation (\ref{eq:realitoc}).
\end{obs}

Lemma~\ref{lem:symmetryproduct} simplifies to the following observation.
\begin{obs}\label{obs:sproduct}
	The following equation holds for every $\psi,\eta\in\cH_I$: 
	\begin{equation}\label{eq:ptrcyclic}
		\ptr(\psi\bullet\eta)=(-1)^{|\psi||\psi|}\ptr(\eta\bullet\psi).	\end{equation}
\end{obs}

The map $\ptr:\cH_I\to\cH_C$ thus makes $\cH_I$ into a vector-valued tracial $\star$-algebra, with graded trace valued in $\cH_C$. The map $\tr:\cH_I\to\C$ inherits the cyclic commutativity property (\ref{eq:ptrcyclic}) from the partial trace, see equation~(\ref{eq:damplgradcom}) of Lemma~\ref{lem:damplcomp}. Also, with equation~(\ref{eq:amplcc}) of Lemma~\ref{lem:damplinvol} it satisfies compatibility with the $\star$-structure. We summarize this as follows.
\begin{obs}
	\begin{equation}
		\tr(\psi\bullet\eta)=(-1)^{|\psi||\psi|}\tr(\eta\bullet\psi) \quad\forall \psi,\eta\in\cH_I .
	\end{equation}
\begin{equation}
	\tr(\psi^\star)=\overline{\tr(\psi)} \quad\forall\psi\in\cH_I .
\end{equation}
\end{obs}

Proposition~\ref{prop:trip}, equation~(\ref{eq:intip}) allows us to recover the inner product on $\cH_I$ as a Hilbert-Schmidt inner product.
\begin{cor}\label{cor:InnerProdKreinAlg}
\begin{equation}
	\langle \psi,\eta\rangle_I=\tr(\psi^\star\bullet \eta) \qquad\forall\psi,\eta\in\cH_I .
	\label{eq:hsip}
\end{equation}
Note also that we have from this,
\begin{equation}
	\langle \psi,\xi\bullet\eta\rangle_I =\langle \xi^\star\bullet\psi,\eta\rangle_I \qquad\forall\psi,\eta,\xi\in\cH_I .
	\label{eq:staradj}
\end{equation}
\end{cor}

A Krein algebra is the Krein-space analog of a Hilbert algebra, i.e.\ it is simultaneously a Krein space, and a $\star$-algebra, and the $\star$-structure coincides with the adjoint of its regular representation. The following theorem follows with the previous observations.

\begin{thm}\label{thm:KreinAlgebra}
The state space $\cH_I$ of the open string $I$ is a Krein algebra, the trace $\tr$ is precisely the usual (graded) trace of the operator algebra on a (graded) Krein space, and the inner product (\ref{eq:hsip}) is precisely the Hilbert-Schmidt inner product.
\end{thm}

The following proposition says that the algebra $\cH_I$ is approximately unital.
\begin{prop}\label{prop:approxid}
	Let $\left\{\zeta_k:k\in \mathbb{N}\right\}$ be an orthonormal basis for $\cH_I$. The sequence
	\begin{equation}
		\one_{I,n}=\sum_{k=1}^n (-1)^{[\zeta_k]}\zeta_k\bullet\zeta_k^\star 
		\label{eq:assocsid}
	\end{equation}
	forms a \emph{left-right approximate identity} for $\cH_I$. If the sum
	\begin{equation}
		\one_I=\sum_{k=1}^\infty (-1)^{[\zeta_k]}\zeta_k\bullet\zeta_k^\star
                \label{eq:associd}
	\end{equation}
	converges, then $\one_I$ is a \emph{unit} for $\cH_I$.
\end{prop}
\begin{proof}
	Let $\psi\in\cH_I$. Rewrite equation (\ref{eq:diskamplautoglue}) as,
	\begin{equation}
		\tr(\psi)=\sum_{k\in I}(-1)^{[\zeta_k]}\tr(\zeta_k\bullet\zeta_k^\star\bullet \psi) .
		\label{eq:assocaglue}
	\end{equation}
	Combining this and equation~(\ref{eq:hsip}) we obtain,
	\begin{equation}
		\langle\psi,\eta\rangle_I=\sum_{k\in I}(-1)^{[\zeta_k]}\tr(\zeta_k\bullet\zeta_k^\star\bullet \psi^\star\bullet\eta) .
	\end{equation}
	Due to the non-degeneracy of the inner product (Axiom~\ref{SliceRegionsAxiom}), this implies
	\begin{equation}
		\psi=\sum_{k\in I}(-1)^{[\zeta_k]}\zeta_k\bullet\zeta_k^\star\bullet \psi .
	\end{equation}
	This, together with (graded) commutativity of the trace provides us with the identity,
	\begin{equation}
		\psi=\sum_{k\in I}(-1)^{[\zeta_k]}\psi\bullet\zeta_k\bullet\zeta_k^\star .
	\end{equation}
	This concludes the proof.
\end{proof}

We also find that we can express the trace as the usual (graded) operator algebra trace.

\begin{cor}
\begin{equation}
	\tr(\psi)=\sum_{k\in I} (-1)^{[\zeta_k]+|\zeta_k|} \langle \zeta_k, \psi\bullet \zeta_k\rangle_I .
	\label{eq:optrace}
\end{equation}
\end{cor}
\begin{proof}
This may be obtained by rewriting equation~(\ref{eq:assocaglue}) using equation~(\ref{eq:hsip}) on the right-hand side.
\end{proof}

The following proposition allows us to identify the state space $\cH_C$ with a subspace of $\cH_I$ and, moreover, the partial trace map $\ptr$ with the orthogonal projection of $\cH_I$ onto $\cH_C$.
\begin{prop}\label{prop:subalg}
	There exists a $\star$-structure- and grade-preserving linear isometry $j:\cH_C\to\cH_I$ such that the composition $P\defeq j\circ$pTr is the orthogonal projection of $\cH_I$ onto $\mathrm{Im}(j)$.
\end{prop}
\begin{proof}
	Rewrite formula (\ref{eq:ipcircle}) for the inner product on $\cH_C$ as follows,
	\begin{equation}
		\langle \ptr(v), \ptr(v')\rangle_C=\sum_{k\in I}(-1)^{[\zeta_k]+|\zeta_k| |v|}\tr(\zeta_k\bullet v^\star\bullet \zeta_k^\star\bullet v') .
		\label{eq:associpc}
	\end{equation}
	for any $v,v'\in\cH_C$. This, together with formula (\ref{eq:hsip}) implies the identity,
	\begin{equation}
		\langle \ptr(v), \ptr(v')\rangle_C=\sum_{k\in I}(-1)^{[\zeta_k]+|\zeta_k| |v|}\langle\zeta_k\bullet v\bullet \zeta_k^\star, v'\rangle_I .
		\label{eq:ipid}
	\end{equation}
	From this and from non-degeneracy of the inner products on $\cH_I$ and $\cH_C$ it follows that if $P:\cH_I\to\cH_I$ denotes the map
	\begin{equation}
		v\mapsto \sum_{k\in I}(-1)^{[\zeta_k]+|\zeta_k| |v|}\zeta_k\bullet v\bullet \zeta_k^\star ,
		\label{eq:assocproj}
	\end{equation}
	then the kernel of $\ptr$ and $P$ coincide. In particular, the orthogonal complement of $\ker P$ is isometrically isomorphic to $\cH_C$ and this provides an isometry $j:\cH_C\to\cH_I$ such that $P=j\circ\ptr$. $P$ is an orthogonal projection, and moreover, $P$ and $j$ are $\star$-maps and respect the grading.
\end{proof}

We identify $\cH_C$ with $\mathrm{Im}(j)$ from now on. With this convention the partial trace map $\ptr$ is the orthogonal projection $P$ of $\cH_I$ onto $\cH_C$.
\begin{thm}\label{thm:center}
	With the above convention $\cH_C$ is the \emph{center} of $\cH_I$ and in particular $\cH_C$ is a $\star$-subalgebra of $\cH_I$.
\end{thm}
\begin{proof}
	Under the above conventions the identity (\ref{eq:ipid}) becomes the trivial identity,
	\begin{equation}
		\langle P v,P v'\rangle_I=\langle P v,v'\rangle_I .
	\end{equation}
	The condition (\ref{eq:ptrcyclic}) now acquires the meaning that projection onto the subspace $\cH_C$ after multiplication is graded commutative. What is more, combining this with the defining identity (\ref{eq:hsip}) of the inner product yields the equality,
	\begin{align}
		\langle u,v\bullet P w\rangle_I & =\tr(u^\star\bullet v\bullet P w)
		=\langle v^\star \bullet u, P w\rangle_I
		=\langle P(v^\star \bullet u), P w\rangle_I \nonumber \\
		& = (-1)^{|v| |u|} \langle P(u \bullet v^\star), P w\rangle_I
		= (-1)^{|v| |u|} \langle u \bullet v^\star, P w\rangle_I \nonumber \\
		& = (-1)^{|v| |u|} \tr(v\bullet u^\star \bullet P w)
		= (-1)^{|v| |w|} \tr(u^\star \bullet P w\bullet v) \nonumber \\
		& = (-1)^{|v| |w|} \langle u, P w\bullet v\rangle_I .
	\end{align}
	This means that elements of $\cH_C$ (graded) commute with all elements of $\cH_I$. In other words, $\cH_C$ is contained in the (graded) center of $\cH_I$. On the other hand, suppose that $v$ is in the center of $\cH_I$. Then we have,
	\begin{equation}
		P v=\sum_{k\in I}(-1)^{[\zeta_k]+|\zeta_k| |v|}\zeta_k\bullet v\bullet \zeta_k^\star= \sum_{k\in I}(-1)^{[\zeta_k]}\zeta_k\bullet \zeta_k^\star\bullet v= v .
	\end{equation}
	That is, $v\in\cH_C$. This concludes the proof.
\end{proof}

\begin{lem}
	Let $\cH_C^\perp$ denote the orthogonal complement of $\cH_C$ in $\cH_I$. Then, $\cH_C\bullet \cH_C^\perp\subseteq \cH_C^\perp$ and $\cH_C^\perp\bullet \cH_C\subseteq \cH_C^\perp$.
\end{lem}
\begin{proof}
Let $\xi,\eta\in\cH_C$ and $\psi\in\cH_C^\perp$. By the subalgebra property, we have, $\langle \psi,\xi\bullet\eta\rangle_I=0$. But with relation (\ref{eq:staradj}) we also have $\langle \xi^\star\bullet\psi,\eta\rangle_I=0$. Since $\psi,\xi,\eta$ were arbitrary, this implies $\cH_C\bullet \cH_C^\perp\subseteq \cH_C^\perp$. Due to the center property we equally have $\cH_C^\perp\bullet \cH_C\subseteq \cH_C^\perp$.
\end{proof}

\begin{figure}
	\begin{center}
		\tikzset{every picture/.style={line width=0.75pt}} 
		\begin{tikzpicture}[x=0.75pt,y=0.75pt,yscale=-1.5,xscale=1.5]
			
			\draw  [draw opacity=0] (266.57,253.4) .. controls (250.2,253.16) and (237,239.82) .. (237,223.4) .. controls (237,206.88) and (250.36,193.48) .. (266.86,193.4) -- (267,223.4) -- cycle ; \draw  [color={rgb, 255:red, 0; green, 0; blue, 0 }  ,draw opacity=1 ] (266.57,253.4) .. controls (250.2,253.16) and (237,239.82) .. (237,223.4) .. controls (237,206.88) and (250.36,193.48) .. (266.86,193.4) ;  
			\draw  [draw opacity=0] (266.17,193.41) .. controls (266.26,193.41) and (266.34,193.41) .. (266.42,193.41) .. controls (282.99,193.09) and (296.68,206.26) .. (296.99,222.82) .. controls (297.31,239.39) and (284.14,253.08) .. (267.58,253.39) .. controls (267.27,253.4) and (266.97,253.4) .. (266.67,253.4) -- (267,223.4) -- cycle ; \draw  [color={rgb, 255:red, 0; green, 0; blue, 0 }  ,draw opacity=1 ] (266.17,193.41) .. controls (266.26,193.41) and (266.34,193.41) .. (266.42,193.41) .. controls (282.99,193.09) and (296.68,206.26) .. (296.99,222.82) .. controls (297.31,239.39) and (284.14,253.08) .. (267.58,253.39) .. controls (267.27,253.4) and (266.97,253.4) .. (266.67,253.4) ;  
			\draw  [draw opacity=0] (258.81,194.53) .. controls (261.21,193.85) and (263.72,193.47) .. (266.32,193.41) .. controls (268.81,193.35) and (271.23,193.6) .. (273.55,194.12) -- (267,223.4) -- cycle ; \draw  [color={rgb, 255:red, 208; green, 2; blue, 27 }  ,draw opacity=1 ] (258.81,194.53) .. controls (261.21,193.85) and (263.72,193.47) .. (266.32,193.41) .. controls (268.81,193.35) and (271.23,193.6) .. (273.55,194.12) ;  
			\draw  [draw opacity=0] (246.78,201.24) .. controls (250.19,198.13) and (254.32,195.79) .. (258.9,194.51) -- (267,223.4) -- cycle ; \draw  [color={rgb, 255:red, 74; green, 144; blue, 226 }  ,draw opacity=1 ] (246.78,201.24) .. controls (250.19,198.13) and (254.32,195.79) .. (258.9,194.51) ;  
			\draw  [draw opacity=0] (239.81,210.71) .. controls (241.52,207.04) and (243.96,203.77) .. (246.94,201.09) -- (267,223.4) -- cycle ; \draw  [color={rgb, 255:red, 208; green, 2; blue, 27 }  ,draw opacity=1 ] (239.81,210.71) .. controls (241.52,207.04) and (243.96,203.77) .. (246.94,201.09) ;  
			\draw  [draw opacity=0] (237,223.19) .. controls (237.03,218.7) and (238.05,214.44) .. (239.85,210.63) -- (267,223.4) -- cycle ; \draw  [color={rgb, 255:red, 74; green, 144; blue, 226 }  ,draw opacity=1 ] (237,223.19) .. controls (237.03,218.7) and (238.05,214.44) .. (239.85,210.63) ;  
			\draw  [draw opacity=0] (273.5,194.1) .. controls (278.67,195.25) and (283.37,197.74) .. (287.18,201.2) -- (267,223.4) -- cycle ; \draw  [color={rgb, 255:red, 74; green, 144; blue, 226 }  ,draw opacity=1 ] (273.5,194.1) .. controls (278.67,195.25) and (283.37,197.74) .. (287.18,201.2) ;  
			\draw  [draw opacity=0] (287.07,201.07) .. controls (291.02,204.61) and (294.04,209.19) .. (295.67,214.41) -- (267.05,223.41) -- cycle ; \draw  [color={rgb, 255:red, 208; green, 2; blue, 27 }  ,draw opacity=1 ] (287.07,201.07) .. controls (291.02,204.61) and (294.04,209.19) .. (295.67,214.41) ;  
			\draw  [draw opacity=0] (294.89,212.32) .. controls (296.07,215.3) and (296.8,218.53) .. (296.96,221.92) .. controls (297.02,223.13) and (297.01,224.33) .. (296.93,225.51) -- (267,223.4) -- cycle ; \draw  [color={rgb, 255:red, 74; green, 144; blue, 226 }  ,draw opacity=1 ] (294.89,212.32) .. controls (296.07,215.3) and (296.8,218.53) .. (296.96,221.92) .. controls (297.02,223.13) and (297.01,224.33) .. (296.93,225.51) ;  
			\draw  [draw opacity=0] (274.06,252.59) .. controls (271.64,253.17) and (269.11,253.45) .. (266.51,253.41) .. controls (264.02,253.36) and (261.62,253.02) .. (259.32,252.41) -- (267.05,223.41) -- cycle ; \draw  [color={rgb, 255:red, 208; green, 2; blue, 27 }  ,draw opacity=1 ] (274.06,252.59) .. controls (271.64,253.17) and (269.11,253.45) .. (266.51,253.41) .. controls (264.02,253.36) and (261.62,253.02) .. (259.32,252.41) ;  
			\draw  [draw opacity=0][dash pattern={on 0.84pt off 2.51pt}] (281.88,198.29) .. controls (277.69,200.67) and (272.86,202.07) .. (267.71,202.18) .. controls (262.13,202.29) and (256.87,200.87) .. (252.34,198.31) -- (267.09,172.18) -- cycle ; \draw  [color={rgb, 255:red, 74; green, 144; blue, 226 }  ,draw opacity=1 ][dash pattern={on 0.84pt off 2.51pt}] (281.88,198.29) .. controls (277.69,200.67) and (272.86,202.07) .. (267.71,202.18) .. controls (262.13,202.29) and (256.87,200.87) .. (252.34,198.31) ;  
			\draw  [draw opacity=0][dash pattern={on 0.84pt off 2.51pt}] (295.98,220.66) .. controls (289.78,225.33) and (279.43,228.47) .. (267.69,228.62) .. controls (255.01,228.79) and (243.87,225.44) .. (237.77,220.31) -- (267.45,210.85) -- cycle ; \draw  [color={rgb, 255:red, 74; green, 144; blue, 226 }  ,draw opacity=1 ][dash pattern={on 0.84pt off 2.51pt}] (295.98,220.66) .. controls (289.78,225.33) and (279.43,228.47) .. (267.69,228.62) .. controls (255.01,228.79) and (243.87,225.44) .. (237.77,220.31) ;  
			\draw [color={rgb, 255:red, 208; green, 2; blue, 27 }  ,draw opacity=1 ]   (266.86,193.4) -- (265.35,193.48) ;
			\draw [shift={(263.35,193.59)}, rotate = 356.93] [color={rgb, 255:red, 208; green, 2; blue, 27 }  ,draw opacity=1 ][line width=0.75]    (4.37,-1.32) .. controls (2.78,-0.56) and (1.32,-0.12) .. (0,0) .. controls (1.32,0.12) and (2.78,0.56) .. (4.37,1.32)   ;
			\draw [color={rgb, 255:red, 208; green, 2; blue, 27 }  ,draw opacity=1 ]   (242.41,206.29) -- (242.2,206.58) ;
			\draw [shift={(241,208.18)}, rotate = 306.87] [color={rgb, 255:red, 208; green, 2; blue, 27 }  ,draw opacity=1 ][line width=0.75]    (4.37,-1.32) .. controls (2.78,-0.56) and (1.32,-0.12) .. (0,0) .. controls (1.32,0.12) and (2.78,0.56) .. (4.37,1.32)   ;
			\draw [color={rgb, 255:red, 208; green, 2; blue, 27 }  ,draw opacity=1 ]   (268.53,253.24) ;
			\draw [shift={(269.24,253.24)}, rotate = 180] [color={rgb, 255:red, 208; green, 2; blue, 27 }  ,draw opacity=1 ][line width=0.75]    (4.37,-1.32) .. controls (2.78,-0.56) and (1.32,-0.12) .. (0,0) .. controls (1.32,0.12) and (2.78,0.56) .. (4.37,1.32)   ;
			\draw [color={rgb, 255:red, 208; green, 2; blue, 27 }  ,draw opacity=1 ]   (292.29,207) -- (291.73,206.25) ;
			\draw [shift={(290.53,204.65)}, rotate = 53.13] [color={rgb, 255:red, 208; green, 2; blue, 27 }  ,draw opacity=1 ][line width=0.75]    (4.37,-1.32) .. controls (2.78,-0.56) and (1.32,-0.12) .. (0,0) .. controls (1.32,0.12) and (2.78,0.56) .. (4.37,1.32)   ;
			\draw [color={rgb, 255:red, 74; green, 144; blue, 226 }  ,draw opacity=1 ]   (252.29,197.35) -- (251.92,197.56) ;
			\draw [shift={(250.18,198.53)}, rotate = 330.95] [color={rgb, 255:red, 74; green, 144; blue, 226 }  ,draw opacity=1 ][line width=0.75]    (4.37,-1.32) .. controls (2.78,-0.56) and (1.32,-0.12) .. (0,0) .. controls (1.32,0.12) and (2.78,0.56) .. (4.37,1.32)   ;
			\draw [color={rgb, 255:red, 74; green, 144; blue, 226 }  ,draw opacity=1 ]   (237.94,215.94) -- (237.86,216.24) ;
			\draw [shift={(237.35,218.18)}, rotate = 284.74] [color={rgb, 255:red, 74; green, 144; blue, 226 }  ,draw opacity=1 ][line width=0.75]    (4.37,-1.32) .. controls (2.78,-0.56) and (1.32,-0.12) .. (0,0) .. controls (1.32,0.12) and (2.78,0.56) .. (4.37,1.32)   ;
			\draw [color={rgb, 255:red, 74; green, 144; blue, 226 }  ,draw opacity=1 ]   (279.82,196.29) -- (279.16,195.98) ;
			\draw [shift={(277.35,195.12)}, rotate = 25.46] [color={rgb, 255:red, 74; green, 144; blue, 226 }  ,draw opacity=1 ][line width=0.75]    (4.37,-1.32) .. controls (2.78,-0.56) and (1.32,-0.12) .. (0,0) .. controls (1.32,0.12) and (2.78,0.56) .. (4.37,1.32)   ;
			\draw [color={rgb, 255:red, 74; green, 144; blue, 226 }  ,draw opacity=1 ]   (296.53,217.47) -- (296.43,217.14) ;
			\draw [shift={(295.82,215.24)}, rotate = 72.47] [color={rgb, 255:red, 74; green, 144; blue, 226 }  ,draw opacity=1 ][line width=0.75]    (4.37,-1.32) .. controls (2.78,-0.56) and (1.32,-0.12) .. (0,0) .. controls (1.32,0.12) and (2.78,0.56) .. (4.37,1.32)   ;
			
			\draw (262.91,183.22) node [anchor=north west][inner sep=0.75pt]  [font=\tiny,color={rgb, 255:red, 208; green, 2; blue, 27 }  ,opacity=1 ]  {$1$};
			\draw (233.82,197.85) node [anchor=north west][inner sep=0.75pt]  [font=\tiny,color={rgb, 255:red, 208; green, 2; blue, 27 }  ,opacity=1 ]  {$2$};
			\draw (293.64,199.31) node [anchor=north west][inner sep=0.75pt]  [font=\tiny,color={rgb, 255:red, 208; green, 2; blue, 27 }  ,opacity=1 ]  {$2$};
			\draw (264.36,256.76) node [anchor=north west][inner sep=0.75pt]  [font=\tiny,color={rgb, 255:red, 208; green, 2; blue, 27 }  ,opacity=1 ]  {$k$};

		\end{tikzpicture}
	\end{center}
	\caption{Gluing of the disk to itself to obtain a sphere with $k$ holes. The blue intervals are glued pairwise as indicated by the dotted lines. The red intervals form the boundaries of the holes after gluing. The intervals marked $1$ and $k$ form the boundaries of holes $1$ and $k$ respectively. The intervals marked $2$ to $k-1$ appear in pairs and form the boundaries of the holes $2$ to $k-1$ with each interval contributing a half-circle.}
	\label{fig:disktoholes}
\end{figure}
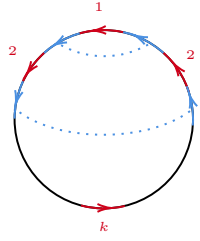

\begin{prop}\label{prop:genus0}
Consider the Riemann surface $S_{0,k}$ of genus $0$ and with $k$ holes, where $k\ge 1$. Assign states $\eta_j\in \cH_{C}$ to the holes, for $1\le j\le k$. Then, the amplitude for $S_{0,k}$ is given by the following formula,
\begin{equation}
		\rho_{0,k}\left(\eta_1\tens \eta_{2}\tens\cdots\tens \eta_{k-1}\tens  \eta_k\right)
		= \rho_D(\eta_1 \bullet \eta_{2}\bullet\cdots\bullet \eta_{k-1}	
		\bullet\eta_k) .
		\label{eq:assocgen0ampl}
\end{equation}
\end{prop}
\begin{proof}
	By Axiom~\ref{AmplitudeEquivarianceAxiom}, the amplitude is invariant under homeomorphisms. Thus, it does not matter how exactly we construct the Riemann surface. We construct $S_{0,k}$ here by gluing a single disk to itself as illustrated in Figure~\ref{fig:disktoholes}. Let $\psi_{j,L}$ and $\psi_{j,R}$ be vectors in the state spaces associated to the left and right intervals labeled by $j$, $1<j<k$, $\psi_1$ and $\psi_k$ be vectors in the state spaces associated to the intervals marked with $1$ and $k$. With Axiom~\ref{AmplitudeRelativeGluingAxiom} this leads to the formula,
	\begin{align}
		& \rho_{0,k}\left(P \psi_k\tens P(\psi_{k-1,L}\bullet\psi_{k-1,R})\tens\cdots\tens P(\psi_{2,L}\bullet\psi_{2,R})\tens P \psi_1\right) \nonumber \\
		= & \sum_{n_1,\ldots,n_{k-1}} s\, (-1)^{[\zeta_{n_1}]+\cdots +[\zeta_{n_{k-1}}]} \nonumber \\
		& \tr(\zeta_{n_1}^\star\bullet \psi_{2,L}\bullet \zeta_{n_2}^\star\bullet \cdots\bullet\psi_{k-1,L}\bullet\zeta_{n_{k-1}}^\star\bullet\psi_k\bullet\zeta_{n_{k-1}}\bullet\psi_{k-1,R}\bullet\cdots\bullet\zeta_{n_2}\bullet\psi_{2,R}\bullet\zeta_{n_1}\bullet\psi_1) .
	\end{align}
	Here, $s$ is a factor of $1$ or $-1$ depending on the fermionic degrees of the vectors appearing on the right-hand side. $s$ arises from the permutations required to change the order as appearing on the left-hand side to that appearing on the right-hand side, taking also into account the basis vectors summed over. We can evaluate the right-hand side as follows. We first notice that the terms around $\psi_k$ imply the application of the projector $P$, then we may move $P\psi_k$ to the extreme left inside the trace using graded commutativity. Next we notice that $\psi_{k-1,L}\bullet\psi_{k-1,R}$ is similarly subject to $P$ and move $P(\psi_{k-1,L}\bullet\psi_{k-1,R})$ to the second position on the left inside the trace. Iterating this procedure yields the identity,
	\begin{align}
		& \rho_{0,k}\left(P \psi_k\tens P(\psi_{k-1,L}\bullet\psi_{k-1,R})\tens\cdots\tens P(\psi_{2,L}\bullet\psi_{2,R})\tens P \psi_1\right) \nonumber \\
		& = \tr(P\psi_k \bullet P(\psi_{k-1,L}\bullet\psi_{k-1,R})\bullet\cdots\bullet P(\psi_{2,L}\bullet\psi_{2,R})	
		\bullet\psi_1) .
	\end{align}
	We note that we may replace $\psi$ with $P \psi$ because the trace implies the projector $P$, and we have the commutativity property of $P$ with left multiplication by elements in $\cH_C$. It is now convenient to work directly with elements in $\cH_C$. We take $\eta_1=P\psi_1$, $\eta_2=P(\psi_{2,L}\bullet \psi_{2,R})$, etc. Thus,
	\begin{equation}
		\rho_{0,k}\left(\eta_k\tens \eta_{k-1}\tens\cdots\tens \eta_{2}\tens  \eta_1\right)
		= \rho_D(\eta_k \bullet \eta_{k-1}\bullet\cdots\bullet \eta_{2}	
		\bullet\eta_1) .
	\end{equation}
\end{proof}

\begin{prop}\label{prop:sphere}
	If $\dim\cH_I<\infty$, then the amplitude of the sphere $S_{0,0}$ is given by the super-dimension of $\cH_I$,
	\begin{equation}
		\rho_{0,0}=\dim\cH_{I}^0-\dim\cH_{I}^1 .
		\label{eq:assocsphereampl}
	\end{equation}
\end{prop}
\begin{proof}
We obtain the sphere by gluing the disk to itself. This yields,
\begin{align}
	\rho_{0,0} & =\sum_{k\in I} (-1)^{[\zeta_k]+|\zeta_k|}\tr(\zeta_k^\star \bullet\zeta_k)
	=\sum_{k\in I} (-1)^{[\zeta_k]+|\zeta_k|}\langle\zeta_k,\zeta_k\rangle_I
	=\sum_{k\in I} (-1)^{|\zeta_k|} \nonumber \\
	& =\dim\cH_{I}^0-\dim\cH_{I}^1 .
\end{align}
\end{proof}

\begin{obs}\label{obs:excludesphere}
If $\dim\cH_I=\infty$, then we exclude the sphere from the admissible regions and exclude gluings that yield the sphere from the admissible gluings.
\end{obs}

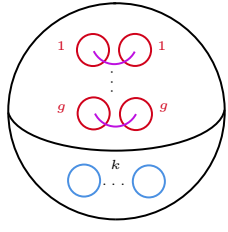
\begin{figure}
\begin{center}
	\tikzset{every picture/.style={line width=0.75pt}} 
	\begin{tikzpicture}[x=0.75pt,y=0.75pt,yscale=-1.5,xscale=1.5]
	
	\draw  [draw opacity=0] (292.44,200.56) .. controls (292.54,200.55) and (292.65,200.55) .. (292.75,200.55) .. controls (312.87,200.16) and (329.5,216.13) .. (329.89,236.21) .. controls (330.27,256.29) and (314.27,272.88) .. (294.14,273.27) .. controls (293.78,273.28) and (293.41,273.28) .. (293.05,273.28) -- (293.45,236.91) -- cycle ; \draw  [color={rgb, 255:red, 0; green, 0; blue, 0 }  ,draw opacity=1 ] (292.44,200.56) .. controls (292.54,200.55) and (292.65,200.55) .. (292.75,200.55) .. controls (312.87,200.16) and (329.5,216.13) .. (329.89,236.21) .. controls (330.27,256.29) and (314.27,272.88) .. (294.14,273.27) .. controls (293.78,273.28) and (293.41,273.28) .. (293.05,273.28) ;  
	\draw  [draw opacity=0] (293.05,273.28) .. controls (273.16,272.99) and (257.12,256.82) .. (257.12,236.91) .. controls (257.12,216.89) and (273.35,200.64) .. (293.4,200.55) -- (293.57,236.91) -- cycle ; \draw  [color={rgb, 255:red, 0; green, 0; blue, 0 }  ,draw opacity=1 ] (293.05,273.28) .. controls (273.16,272.99) and (257.12,256.82) .. (257.12,236.91) .. controls (257.12,216.89) and (273.35,200.64) .. (293.4,200.55) ;  
	\draw  [draw opacity=0] (329.84,236.41) .. controls (328.59,244.04) and (312.77,250.1) .. (293.42,250.14) .. controls (274.7,250.18) and (259.26,244.57) .. (257.1,237.31) -- (293.39,235.5) -- cycle ; \draw   (329.84,236.41) .. controls (328.59,244.04) and (312.77,250.1) .. (293.42,250.14) .. controls (274.7,250.18) and (259.26,244.57) .. (257.1,237.31) ;  
	\draw  [color={rgb, 255:red, 208; green, 2; blue, 27 }  ,draw opacity=1 ] (280.22,216.28) .. controls (280.22,213.3) and (282.63,210.89) .. (285.61,210.89) .. controls (288.59,210.89) and (291,213.3) .. (291,216.28) .. controls (291,219.25) and (288.59,221.67) .. (285.61,221.67) .. controls (282.63,221.67) and (280.22,219.25) .. (280.22,216.28) -- cycle ;
	\draw  [color={rgb, 255:red, 208; green, 2; blue, 27 }  ,draw opacity=1 ] (294.36,216.25) .. controls (294.36,213.27) and (296.77,210.86) .. (299.75,210.86) .. controls (302.73,210.86) and (305.14,213.27) .. (305.14,216.25) .. controls (305.14,219.23) and (302.73,221.64) .. (299.75,221.64) .. controls (296.77,221.64) and (294.36,219.23) .. (294.36,216.25) -- cycle ;
	\draw  [color={rgb, 255:red, 208; green, 2; blue, 27 }  ,draw opacity=1 ] (280.44,237.61) .. controls (280.44,234.63) and (282.86,232.22) .. (285.83,232.22) .. controls (288.81,232.22) and (291.22,234.63) .. (291.22,237.61) .. controls (291.22,240.59) and (288.81,243) .. (285.83,243) .. controls (282.86,243) and (280.44,240.59) .. (280.44,237.61) -- cycle ;
	\draw  [color={rgb, 255:red, 208; green, 2; blue, 27 }  ,draw opacity=1 ] (294.67,237.83) .. controls (294.67,234.86) and (297.08,232.44) .. (300.06,232.44) .. controls (303.03,232.44) and (305.44,234.86) .. (305.44,237.83) .. controls (305.44,240.81) and (303.03,243.22) .. (300.06,243.22) .. controls (297.08,243.22) and (294.67,240.81) .. (294.67,237.83) -- cycle ;
	\draw  [draw opacity=0] (299.64,216.82) .. controls (298.32,219.36) and (295.68,221.07) .. (292.68,221.03) .. controls (289.7,220.98) and (287.13,219.21) .. (285.87,216.66) -- (292.8,213.09) -- cycle ; \draw  [color={rgb, 255:red, 189; green, 16; blue, 224 }  ,draw opacity=1 ] (299.64,216.82) .. controls (298.32,219.36) and (295.68,221.07) .. (292.68,221.03) .. controls (289.7,220.98) and (287.13,219.21) .. (285.87,216.66) ;  
	\draw  [draw opacity=0] (300.06,237.83) .. controls (298.73,240.37) and (296.1,242.08) .. (293.1,242.04) .. controls (290.11,241.99) and (287.54,240.22) .. (286.29,237.67) -- (293.22,234.1) -- cycle ; \draw  [color={rgb, 255:red, 189; green, 16; blue, 224 }  ,draw opacity=1 ] (300.06,237.83) .. controls (298.73,240.37) and (296.1,242.08) .. (293.1,242.04) .. controls (290.11,241.99) and (287.54,240.22) .. (286.29,237.67) ;  
	\draw  [color={rgb, 255:red, 74; green, 144; blue, 226 }  ,draw opacity=1 ] (282.49,265.59) .. controls (279.51,265.52) and (277.16,263.05) .. (277.23,260.08) .. controls (277.29,257.1) and (279.76,254.75) .. (282.74,254.82) .. controls (285.71,254.89) and (288.07,257.35) .. (288,260.33) .. controls (287.93,263.3) and (285.46,265.66) .. (282.49,265.59) -- cycle ;
	\draw  [color={rgb, 255:red, 74; green, 144; blue, 226 }  ,draw opacity=1 ] (304.26,265.88) .. controls (301.29,265.81) and (298.93,263.34) .. (299,260.36) .. controls (299.07,257.39) and (301.54,255.03) .. (304.52,255.1) .. controls (307.49,255.17) and (309.85,257.64) .. (309.78,260.62) .. controls (309.71,263.59) and (307.24,265.95) .. (304.26,265.88) -- cycle ;
	
	\draw (293.29,221.76) node [anchor=north west][inner sep=0.75pt]  [font=\tiny,rotate=-90]  {$\dotsc $};
	\draw (272.47,212.36) node [anchor=north west][inner sep=0.75pt]  [font=\tiny,color={rgb, 255:red, 208; green, 2; blue, 27 }  ,opacity=1 ]  {$1$};
	\draw (306.29,212.45) node [anchor=north west][inner sep=0.75pt]  [font=\tiny,color={rgb, 255:red, 208; green, 2; blue, 27 }  ,opacity=1 ]  {$1$};
	\draw (272.38,233.54) node [anchor=north west][inner sep=0.75pt]  [font=\tiny,color={rgb, 255:red, 208; green, 2; blue, 27 }  ,opacity=1 ]  {$g$};
	\draw (306.75,233.09) node [anchor=north west][inner sep=0.75pt]  [font=\tiny,color={rgb, 255:red, 208; green, 2; blue, 27 }  ,opacity=1 ]  {$g$};
	\draw (287.46,260.06) node [anchor=north west][inner sep=0.75pt]  [font=\tiny,rotate=-1.34]  {$\dotsc $};
	\draw (290.53,252.26) node [anchor=north west][inner sep=0.75pt]  [font=\tiny]  {$k$};
	\end{tikzpicture}
\end{center}
\caption{Gluing of a sphere with $2g+k$ holes into a Riemann surface of genus $g$ with $k$ holes. Red ($2g$) and blue ($k$) circles are holes in the sphere. The $g$ pairs of red circles are glued together as indicated by the purple lines.}
\label{fig:genus0tog}
\end{figure}

\begin{thm}\label{thm:ArbitraryGenus}
Consider a Riemann surface $S_{g,k}$ of genus $g$ with $k$ holes. Assign states $\eta_1,\ldots,\eta_k\in\cH_C$ to the holes. If $k\ge 1$, then its amplitude is given by,
\begin{equation}
	\rho_{g,k}(\eta_1\tens\cdots\tens\eta_k)
	=\sum_{j_1,\ldots,j_g\in I'} (-1)^{[\xi_{j_1}]+\cdots+[\xi_{j_g}]}
	\rho_D(\xi_{j_1}\bullet\xi_{j_1}^\star\bullet\cdots\bullet\xi_{j_g}\bullet\xi_{j_g}^\star\bullet\eta_1\bullet\cdots\bullet\eta_k) .
	\label{eq:assocampl}
\end{equation}
Here $\{\xi_j\}_{j\in I'}$ is an orthonormal basis of the subspace $\cH_C$.
This formula is also valid if $k=0$ and $g\ge 1$ and the sum converges. Moreover, if the element
\begin{equation}
	x\defeq \sum_{j\in I'} (-1)^{[\xi_j]} \xi_j\bullet\xi_j^\star
	\label{eq:assoccbase}
\end{equation}
exists in $\cH_C$, then we may rewrite the amplitude as,
\begin{equation}
	\rho_{g,k}(\eta_1\tens\cdots\tens\eta_k)
	=\rho_D(x^{\bullet g}\bullet\eta_1\bullet\cdots\bullet\eta_k) .
	\label{eq:assocxampl}
\end{equation}
In this case it exists also for $k=0$ and $g\ge 1$.
\end{thm}
\begin{proof}
We obtain $S_{g,k}$ from the Riemann surface of genus $0$ with $2g+k$ holes, $S_{0,2g+k}$, by gluing $g$ pairs of holes together, see Figure~\ref{fig:genus0tog}. We consider the holes as circles with standard orientations, ordered from left to right, and we glue the first $g$ pairs of holes together. To this end we need $g$ pairs of circles $C\sqcup \overline{C}$ (standard circle plus oppositely oriented pair) to glue along, together with $k$ standard circles $C$. In terms of Axiom~\ref{AmplitudeRelativeGluingAxiom}, the gluing function $\beta$ is simply the identity on the disjoint union of the $k$ standard circles. As for the gluing function $\alpha$, we need to twist the oppositely oriented circles so that we can map to the Riemann surface $S_{0,2g+k}$. That is, $\alpha$ is given by,
\begin{equation}
	\alpha=\id_{C_1}\sqcup f_{\overline{C}_2,C_2}\sqcup\cdots\sqcup \id_{C_{2g-1}}\sqcup f_{\overline{C}_{2g-},C_{2g}}\sqcup\id_{C_{2g+1}}\cdots\sqcup\id_{C_{2g+k}} .
\end{equation}
Consequently, the amplitude for $S_{g,k}$ is given by equation~(\ref{eq:assocampl}) as stated.
\end{proof}

\begin{obs}
The Riemann surfaces of genus $g$ without holes, for which the corresponding sum (\ref{eq:assocampl}) does not converge, have to be declared non-admissible. Note that if the sum converges for genus $g$, it also converges for any higher genus.
\end{obs}

\begin{obs}\label{obs:rec2tqft}
	Consider the setting of $2$-dimensional TQFT. The state spaces $\cH_C$ are then finite dimensional. Moreover, they form a commutative Frobenius algebra \cite{Abr:2TQFTfrob}. We sketch how this is recovered from our setting. The algebra structure $A:\cH_C\tens\cH_C\to\cH_C$ arises there from a pair-of-pants cobordism which has two incoming circles and one outgoing circle. This map is obtained here by dualizing one tensor factor in the amplitude for the sphere with three holes. With Proposition~\ref{prop:genus0} and Corollary~\ref{cor:InnerProdKreinAlg} we have,
	\begin{equation}
		A(\eta_1\tens\eta_2)
		=\sum_{j\in I'}\rho_{0,3}(\xi_j^\star\tens\eta_1\tens\eta_2)\, \xi_j
		=\sum_{j\in I'}\rho_D(\xi_j^\star\bullet\eta_1\bullet\eta_2)\, \xi_j
		=\sum_{j\in I'}\langle\xi_j,\eta_1\bullet\eta_2\rangle_I\, \xi_j
		=\eta_1\bullet\eta_2 .
	\end{equation}
	That is, the TQFT product $A$ is precisely the product we have obtained here for $\cH_C$. Similarly, we obtain the coproduct and the Frobenius property is automatic from Proposition~\ref{prop:genus0}.
\end{obs}

\subsection{Area-dependent theory}
\label{sec:2area}

In this section we consider the modified axioms for the area-dependent theory in dimension $n=2$, according to Section~\ref{sec:areaaxioms}. In GBQFT, an area-dependent theory was considered in the example of pure 2-dimensional Yang-Mills theory in \cite{Oe:2dqym}. We consider this example in Section~\ref{sec:2dpqym}. In TQFT, area-dependent theories were classified by Runkel and Szegedy \cite{RuSz:areaqft}.

As for the elementary objects given by hypersurfaces and their gluing functions, there is no change at all. The same applies to the generators and their relations as long as no regions are involved.
As before, any region is an oriented, topological Riemann surface with holes. However, regions now also carry an area $a\ge 0$. We pick a standard Riemann surface for each choice of genus $g\in\N_0$, number of holes $k\in\N_0$ and area $a\in\R^+_0$.

\begin{definition}
	Let $k\in\N_0$ and $a\in\R^+_0$. We make a choice of standard Riemann surface with genus $g$, $k$ holes, and area $a$, and we denote it by $S_{a;g,k}$. We identify the previously introduced standard Riemann surfaces with those of vanishing area, i.e., $S_{g,k}=S_{0;g,k}$.
\end{definition}

As we did in \ref{sec:2rel} we use a separate notation for our choice of standard disk with area.

\begin{definition}
	Let $a\in\R^+_0$. We denote our choice of standard disk with area $a$, $S_{a;0,0}$, by $D_a$. We make it so that the boundary of each disk $D_a$ is the standard circle $C$.
\end{definition} 

As for the generators, according to Axiom~\ref{AmplitudeMapsAxiom}, we have an assignment of an amplitude to the standard disk $D_a$ for value $a\ge 0$ of the area now, instead of a single assignment. That is, we have the $1$-parameter family of maps,
\begin{equation}
\rho_{D_a}:\cH_C^{\ds}\to\C .
\end{equation}
Except for this family of maps, the list of generators of Table~\ref{tab:elemgen} remains unchanged.

We proceed to consider what changes with respect to special homeomorphisms in the area-dependent theory. Since we now have one standard disk $D_a$ for each value of the area $a\ge 0$, we need a standard choice of "orientation twist" homeomorphism that identifies $D_a$ with its orientation-reversed $\overline{D}_a$ for each $a$. 

\begin{definition}
We make a choice of standard "twist" gluing functions as above, and we denote them by $f_{D_a,\overline{D}_a}:D_a\to \overline{D}_a$. We denote the corresponding orientation reversed versions by $f_{\overline{D}_a,D_a}:\overline{D}_a\to D_a$. We require these functions to be mutually inverse. We moreover require compatibility with $f_{C,\overline{C}}$ in the sense that the latter is the restriction of $f_{\overline{D}_a,D_a}:\overline{D}_a\to D_a$ to the boundary. Also, $f_{D_0,\overline{D}_0}$ coincides with the previously defined $f_{D,\overline{D}}$.
\end{definition}

The generating assignments of the special homeomorphisms are still as summarized in Table~\ref{tab:elemhomgen}, except for the disk twists now being parametrized by area.

We turn to consider how the relations (Section~\ref{sec:2rel}) are adapted to the area-dependent theory. Relations that do not involve regions are unchanged and relations that do involve regions continue to hold unchanged for regions of area $a=0$. The latter implies in particular the inner products in $\cH_I$ and in $\cH_C$. Also, many relations generalize straightforwardly to regions of positive area. In particular Lemas~\ref{lem:damplinvol} and \ref{lem:damplcomp} with relations (\ref{eq:amplconj}), (\ref{eq:amplflip}), (\ref{eq:amplcc}), (\ref{eq:amplrot}), (\ref{eq:damplgradcom}) for the standard disk amplitude generalize to the disk amplitude with arbitrary area. 

\begin{lem}\label{lem:aamplcomp}
	The following relations hold:
	\begin{align}
		\rho_{\overline{D}_a}(\iota_C(\psi)) & =\overline{\rho_{D_a}(\psi)},
		\label{eq:amplaconj}\\
		\rho_{\overline{D}_a}(\gamma_{C,\overline{C}}(\psi)) & =\rho_{D_a}(\psi) ,
		\label{eq:amplaflip} \\
		\rho_{D_a}(\alpha_C(\psi)) & =\overline{\rho_{D_a}(\psi)} ,
		\label{eq:amplacc} \\
		\rho_{D_a}\circ\gamma_\pi & = \rho_{D_a} ,
		\label{eq:amplarot} \\
		\rho_{D_a}(\gamma_{I,C}(\gamma_{I\sqcup I,I}(\psi\tens\eta)))
		& =(-1)^{|\psi| |\eta|}\rho_{D_a}(\gamma_{I,C}(\gamma_{I\sqcup I,I}(\eta\tens\psi))) .
		\label{eq:damplagradcom}
	\end{align}
	
\end{lem}
Proposition~\ref{prop:diskamplautoglue} arising from gluing a standard disk to itself along parts of its boundary generalizes to arbitrary area.

\begin{prop}
	For every $v\in\cH_C$ and $\{\zeta_k\}_{k\in I}$ orthonormal basis of $\cH_I$, the following relation holds:
	\begin{equation}
		\rho_{D_a}(\gamma_{I,C}(v))=c_{a}\sum_{k\in I}(-1)^{[\zeta_k]}\rho_{D_a}\circ\gamma_{I,C}\circ\gamma_{I\sqcup I\sqcup I,I}(\zeta_k\tens\alpha_I(\zeta_k)\tens v) .
		\label{eq:diskamplaautoglue}
	\end{equation}
	Note that the gluing anomaly factor $c$ is replaced by a gluing anomaly factor $c_a$ that may depend on the area $a$, with $c_0=c$.
\end{prop}

Another relation for disk amplitudes arises from gluing two disks together to form a single disk, see Example~\ref{ex:twodiskgluing}. If disks have no area this just reduces to the completeness of the inner product. Only in the present area-dependent context does this yield a non-trivial additional relation. In this case the areas of the disks sum.

\begin{lem}\label{lem:2diskgrel}
	Apply the relative gluing Axiom~\ref{AmplitudeRelativeGluingAxiom} to Example~\ref{ex:twodiskgluing} with $\Lambda=I_3\sqcup I_4$, $\Sigma=I$, $X=D_a\sqcup D_b$ and $X'=D_{a+b}$. Moreover, $\beta=f_{I,C}\circ f_{I\sqcup I,I}$ and,
\begin{equation}
	\alpha=(f_{I,C}\sqcup f_{I,C})\circ (f_{I_1\sqcup I_3,I}\sqcup f_{I_2\sqcup I_4,I})
	\circ (\id_{I_1}\sqcup \sigma_{2 3} \sqcup \id_{I_4})
	\circ(\id_{I_1}\sqcup f_{\overline{I_2},I_2}\sqcup\id_{I_3}\sqcup\id_{I_4}) .
\end{equation}
The resulting relation for disk amplitudes is,
\begin{multline}
	\rho_{D_{a+b}}(\gamma_{I,C}\circ\gamma_{I\sqcup I,I}(v\tens v')) \\
	=c_{a,b}\sum_{k\in I}(-1)^{[\zeta_k]}\rho_{D_a}\circ\gamma_{I,C}\circ\gamma_{I\sqcup I,I}(v\tens \zeta_k) \,
	\rho_{D_b}\circ\gamma_{I,C}\circ\gamma_{I\sqcup I,I}(\alpha_I(\zeta_k)\tens v') .
	\label{eq:2diskgrel}
\end{multline}
Here we have applied (\ref{eq:damplagradcom}) in the second step. If $a=0$ or $b=0$ the corresponding disk amplitude can be interpreted as an inner product. This implies, $c_{0,b}=c_{a,0}=1$.
\end{lem}
\begin{proof}
	We have,
	\begin{multline}
		\rho_{D_{a+b}}(\gamma_{I,C}\circ\gamma_{I\sqcup I,I}(v\tens v')) \\
		=c_{a,b}\sum_{k\in I}(-1)^{[\zeta_k]+|\zeta_k| |v|}\rho_{D_a}\circ\gamma_{I,C}\circ\gamma_{I\sqcup I,I}(\zeta_k\tens v) \,
		\rho_{D_b}\circ\gamma_{I,C}\circ\gamma_{I\sqcup I,I}(\alpha_I(\zeta_k)\tens v') .
	\end{multline}
	Applying relation (\ref{eq:damplagradcom}) yields the desired result (\ref{eq:2diskgrel}).
\end{proof}

We proceed to adapt the setting of Section~\ref{sec:assocalg}. That is, we require strictness and anomaly-freedom for the remainder of this section. As a first remark, all results of Section~\ref{sec:assocalg} remain true if we identify the regions considered there with the corresponding regions of vanishing area in the present setting. We shall thus adopt in the following the notations and conventions of that section.

\begin{definition}
	We denote by $\tr_a$ the composition $\rho_{D_a}\circ\gamma_{I,C}$. The maps $\tr_a$ form a one-parameter family of (complete) traces $\tr_a$ on the algebra $\cH_I$. The previously defined trace is the special case $\tr=\tr_0$.
\end{definition}
\begin{obs}
	The trace maps $\tr_a$ inherit the cyclic graded commutativity property (\ref{eq:ptrcyclic}) from the partial trace, see (\ref{eq:damplagradcom}). Also, with equation (\ref{eq:amplacc}) they satisfy compatibility with the $\star$-structure,
	\begin{equation}
		\tr_a(\psi^\star)=\overline{\tr_a(\psi)} .
		\label{eq:trastar}
	\end{equation}
	Relation (\ref{eq:2diskgrel}) of Lemma~\ref{lem:2diskgrel} arising from the gluing of two disks to one may then be rewritten as,
	\begin{equation}
		\tr_{a+b}(v\bullet v')
		=\sum_{k\in I}(-1)^{[\zeta_k]}\tr_a(v\bullet \zeta_k) \,
		\tr_b(\zeta_k^\star\bullet v') .
		\label{eq:assoc2diskgrel}
	\end{equation}
	This in turn is suggestive of generalizing the inner product formula (\ref{eq:hsip}) for arbitrary area. We write this as,
	\begin{equation}
		\langle \psi,\eta\rangle_a=\tr_a(\psi^\star\bullet \eta) .
		\label{eq:ahsip}
	\end{equation}
	Then, $\langle\cdot,\cdot\rangle_a$ is a $1$-parameter family of sesquilinear forms on $\cH_I$, where $\langle\cdot,\cdot\rangle_0=\langle\cdot,\cdot\rangle_I$. These are hermitian as can be seen either by adapting the relevant proof of Theorem~\ref{thm:propip} or by using the $\star$-algebra properties and property (\ref{eq:trastar}).
\end{obs}
\begin{notation}
	We rewrite relation~(\ref{eq:assoc2diskgrel}) as the generalized completeness relation,
	\begin{equation}
		\langle v, v'\rangle_{a+b}
		=\sum_{k\in I}(-1)^{[\zeta_k]}\langle v,\zeta_k\rangle_a
		\langle \zeta_k, v'\rangle_b .
	\end{equation}
\end{notation} 

We consider the amplitudes for Riemann surfaces with area. These are derived in complete analogy to the area-less setting. We exhibit first the generalization of Propositions~\ref{prop:genus0} and Proposition~\ref{prop:sphere}.
\begin{prop}\label{prop:genus0a}
	For the Riemann surface $S_{a;0,k}$ with $k\ge 1$ holes, vanishing genus and area $a\ge 0$, we obtain in analogy to relation~(\ref{eq:assocgen0ampl}),
	\begin{equation}
		\rho_{a;0,k}(\eta_1\tens\cdots\tens\eta_k)=\rho_{D_a}(\eta_1\bullet\cdots\bullet\eta_k) .
		\label{eq:assocgen0ampla}
	\end{equation}
	For the sphere we get in analogy to relation~(\ref{eq:assocsphereampl}),
	\begin{equation}
		\rho_{a;0,0} =\sum_{k\in I} (-1)^{[\zeta_k]+|\zeta_k|}\tr_a(\zeta_k^\star \bullet\zeta_k)
		=\sum_{k\in I} (-1)^{[\zeta_k]+|\zeta_k|}\langle\zeta_k,\zeta_k\rangle_a .
		\label{eq:assocsphereampla}
	\end{equation}
\end{prop}

As in the area-less case the expression (\ref{eq:assocsphereampla}) might or might not exist, compare Observation~\ref{obs:excludesphere}.

Finally, we consider the amplitude for a Riemann surface $S_{a;g,k}$ with genus $g\ge 1$. This generalizes Theorem~\ref{thm:ArbitraryGenus}.
\begin{thm}\label{thm:arbitrarygenusa}
	Consider a Riemann surface $S_{g,k}$ of genus $g$ with $k$ holes. Assign states $\eta_1,\ldots,\eta_k\in\cH_C$ to the holes. If $k\ge 1$, then its amplitude is given by, compare (\ref{eq:assocampl}),
	\begin{equation}
		\rho_{a;g,k}(\eta_1\tens\cdots\tens\eta_k)
		=\sum_{j_1,\ldots,j_g\in I'} (-1)^{[\xi_{j_1}]+\cdots+[\xi_{j_g}]}
		\rho_{D_a}(\xi_{j_1}\bullet\xi_{j_1}^\star\bullet\cdots\bullet\xi_{j_g}\bullet\xi_{j_g}^\star\bullet\eta_1\bullet\cdots\bullet\eta_k) .
	\end{equation}
	Here $\{\xi_j\}_{j\in I'}$ is an orthonormal basis of the subspace $\cH_C$.
	This formula is also valid if $k=0$ and $g\ge 1$ and the sum converges.
	If $x$ as given by expression (\ref{eq:assoccbase}) exists, this is, compare (\ref{eq:assocxampl}),
	\begin{equation}
		\rho_{a;g,k}(\eta_1\tens\cdots\tens\eta_k)
		=\rho_{D_a}(x^{\bullet g}\bullet\eta_1\bullet\cdots\bullet\eta_k) .
		\label{eq:assocxampla}
	\end{equation}
	
\end{thm}


\section{Pure quantum Yang-Mills theory in 2 dimensions}
\label{sec:2dpqym}

In the present section we consider a first example of a symmetric area-dependent CQFT. This is 2-dimensional pure quantum Yang-Mills theory with corners. We take advantage of our findings in Section~\ref{sec:2dim} concerning the specific structure of 2-dimensional CQFTs and particularly the area-dependent setting of Section~\ref{sec:2area}. Our presentation here can be seen as essentially corroborating and making more rigorous the previous treatment of the theory within GBQFT \cite{Oe:2dqym}. A treatment in TQFT was given already around 1990, see in particular Witten's paper \cite{Wit:qgauge2d}. The latter was brought into a rigorous and categorical setting by Runkel and Szegedy as an example of area-dependent quantum field theory \cite{RuSz:areaqft}.

Let $G$ be a group. We introduce convenient notation concerning the rigid symmetric monoidal category of finite-dimensional complex representations of $G$. Given two representations $V$ and $W$, the symmetric structure $s_{V,W}:V\tens W\to W\tens V$ is given by the trivial transposition $v\tens w\mapsto w\tens v$ inherited from the underlying category of finite-dimensional vector spaces. Given a representation $V$ we denote the action $G\times V\to V$ as $(g, v)\mapsto g\act v$. We recall that the action of $G$ on the dual representation $V^*$ is given by, $(g\act w)(v)=w(g^{-1}\act v)$, where $v\in V$, $w\in V^*$, $g\in G$. Functions $G\to\C$ that can be written in the form $g\mapsto w(g\act v)$ are called \emph{representative functions}. Note that in this way $V^*\tens V$ acquires the structure of a vector space of functions on $G$. We denote the vector space of all representative functions on $G$ by $\Calg(G)$. Also, a function $f:G\to\C$ that is invariant under conjugation, i.e., satisfies $f(hgh^{-1})=f(g)$, is called a \emph{class function}.

Let $G$ be a compact Lie group. We consider its category of unitary finite-dimensional representations. This is a semisimple rigid symmetric monoidal category. $G$ admits a bi-invariant normalized Haar measure. This allows to define an invariant positive-definite inner product on functions on $G$ via
\begin{equation}
	\langle \psi,\eta\rangle = \int \overline{\psi(g)} \eta(g)\, \xd g .
\end{equation}
We denote the measurable functions on $G$ that are square-integrable with respect to this inner product by $\Csq(G)$. Crucially, $\Calg(G)$ is a dense subspace of $\Csq(G)$. We also denote the subspace of class functions in $\Csq(G)$ by $\Cclass(G)$.
By the Peter-Weyl Theorem, there is an isomorphism of unitary representations of $G\times G$,
\begin{equation}
	\Calg(G)= \bigoplus_V V^*\tens V .
\end{equation}
$G$ acts on $\Calg(G)$ by multiplication from the left and from the right which gives rise to an action of $G\times G$. The direct sum runs over one irreducible unitary representation $V$ of $G$ per equivalence class. Each copy of $G$ acts on one tensor factor.

We choose an orthonormal basis $\{v_i\}_{i\in I_V}$ for each representation $V$ and introduce a matrix element notation for representative functions,
\begin{equation}
	t^V_{i j}(g)\defeq \langle v_i,g\act v_j\rangle_V .
\end{equation}
This provides an orthogonal basis of $\Csq(G)$ as follows,
\begin{equation}
	\langle t^V_{i j}, t^W_{k l}\rangle=\delta_{V,W}\delta_{i,k}\delta_{j,l}\frac{1}{\dim V} .
	\label{eq:ipcg}
\end{equation}
We also denote $\chi^V\defeq \sum_k t^V_{k k}$. These form an orthonormal basis of $\Cclass(G)$,
\begin{equation}
	\langle \chi^V,\chi^W\rangle=\delta_{V,W} .
	\label{eq:ipccg}
\end{equation}

\begin{table}
	\caption{Generating assignments and 2QPYM model.}
	\label{tab:gen2qpym}
	\begin{center}
			\begin{tabular}{|l|l|}
				\hline
				Generating assignment & 2PQYM \\
				\hline
                $\cH_I$, $\cH_I^\ds$ & $\Csq(G)$, $\Calg(G)$ \\
				$\cH_C$, $\cH_C^\ds$ & $\Cclass(G)$, $\Ccalg(G)$ \\
                $\iota_I:\cH_I\to\cH_{\overline{I}}$ & $(\iota_I(\psi))(g)=\overline{\psi(g)}$ \\
                $\iota_C:\cH_C\to\cH_{\overline{C}}$ & $(\iota_C(\psi))(g)=\overline{\psi(g)}$ \\
				$\gamma_{I\sqcup I,I}:\cH_I\tens\cH_I\to\cH_I$ & $(\gamma_{I\sqcup I,I}(\psi\tens\eta))(g)=\int \psi(g h)\eta(h^{-1})\,\xd h$ \\
				$\gamma_{I,C}:\cH_I\to\cH_C$ & $(\gamma_{I,C}(\psi))(g)=\int \psi(h g h^{-1})\,\xd h$ \\
				$\rho_{D_a}:\cH_C^{\ds}\to\C$ & $\rho_{D_a}(\chi^V)=\dim V \exp(\beta_V a)$ \\
				$\gamma_{I,\overline{I}}:\cH_{I}\to\cH_{\overline{I}}$ & $(\gamma_{I,\overline{I}}(\psi))(g)=\psi(g^{-1})$ \\
				$\gamma_{C,\overline{C}}:\cH_{C}\to\cH_{\overline{C}}$ & $(\gamma_{C,\overline{C}}(\psi))(g)=\psi(g^{-1})$ \\
				$\alpha_I:\cH_I\to\cH_I$ & $(\alpha_I(\psi))(g)=\overline{\psi(g^{-1})}$ \\
				$\alpha_C:\cH_C\to\cH_C$ & $(\alpha_C(\psi))(g)=\overline{\psi(g^{-1})}$ \\
				\hline
			\end{tabular}
	\end{center}
\end{table}

We proceed to consider pure 2-dimensional quantum Yang-Mills (2QPYM) theory with compact gauge group $G$ as a model in the sense of the area-dependent and symmetric version of our axiomatic system. To this end we rely on the work \cite{Oe:2dqym}, but adapt its results to the present setting. We start by specifying the generating assignments listed in Tables~\ref{tab:elemgen} and \ref{tab:elemhomgen}, adapted to the area-dependent setting as explained in Section~\ref{sec:2area}. We add to these the maps $\alpha_I$, $\alpha_C$. The model assignments for 2QPYM are indicated in Table~\ref{tab:gen2qpym}. Note that in this model the gluing automorphisms associated to hypersurface homeomorphisms are identity maps and gluing anomalies are trivial. In particular, we are in the situation of Section~\ref{sec:assocalg} of a strict and anomaly-free CQFT and shall use corresponding notation as we see fit. Note that the inclusion $\cH_C\subseteq \cH_I$ found in that Section is precisely the inclusion $\Cclass(G)\subseteq \Csq(G)$ here. 

The 2-interval gluing map $\gamma_{I\sqcup I,I}$, the interval self-gluing map $\gamma_{I,C}=P$, and the real structures admit explicit expressions in terms of matrix elements,
\begin{align}
	t^V_{i j}\bullet t^W_{k l} & =\delta_{V,W}\delta_{j,k} \frac{t^V_{i l}}{\dim V} , \qquad \chi^V\bullet\chi^W=\delta_{V,W} \frac{\chi^V}{\dim V},\\
	P\, t^V_{i j} & =\delta_{i,j} \frac{\chi^V}{\dim V} , \\
	(t^V_{i j})^\star & =t^V_{j i},\qquad (\chi^V)^\star=\chi^V .
	\label{eq:tijstar}
\end{align}
In particular, we obtain the $\star$-algebra structure of a direct sum of matrix algebras. We have one matrix algebra per unitary irreducible representation $V$. Moreover, $P$ is an orthogonal projector as required. It projects in each matrix algebra on a  multiple of the identity element $\one_V\defeq \dim V \chi^V$.

Compared to the work \cite{Oe:2dqym}, there are some important differences. In particular, the maps $\gamma_{I,\overline{I}}$ and $\gamma_{C,\overline{C}}$ were implicit in that work. Therefore, the maps $\iota_O$ and $\iota_C$ from that work correspond to the maps $\alpha_I$ and $\alpha_C$ here. Moreover, the constants $\beta_V$ were taken to be imaginary numbers in that work. Here, however, we have from relation~(\ref{eq:amplacc}) of Lemma~\ref{lem:aamplcomp} and from relation~(\ref{eq:tijstar}),
\begin{equation}
	\overline{\rho_{D_a}(\chi^V)}=\rho_{D_a}((\chi^V)^\star)
	=\rho_{D_a}(\chi^V) .
\end{equation}
With Table~\ref{tab:gen2qpym} we then get,
\begin{equation}
	\overline{\exp(\beta_V a)}=\exp(\beta_V a)
\end{equation}
for any $a\ge 0$. Thus, we require the constants $\beta_V$ to be real.

We may verify the relations that we have presented in Section~\ref{sec:2rel}, either in the form given there, or in the form as given in Section~\ref{sec:assocalg}. Specifically, it is straightforward to verify relations (\ref{eq:starprod}), (\ref{eq:ptrstar}), (\ref{eq:ptrcyclic}) of Observations~\ref{obs:ProductStarStrict} and \ref{obs:sproduct}. It is also easy to see that relation (\ref{eq:hsip}) of Corollary~\ref{cor:InnerProdKreinAlg} precisely recovers the inner product (\ref{eq:ipcg}) on $\cH_I$. As for the identity (\ref{eq:associd}) in the algebra, described in Proposition~\ref{prop:approxid}, we would have,
\begin{equation}
	\one_I=\sum_{k\in I} \zeta_k\bullet \zeta_k^\star=\sum_{V,i,j} \dim V\, t^V_{i j}\bullet t^V_{j i}=\sum_{V,i,j} t^V_{i i}=\sum_V \dim V\,\chi^V=\sum_V \one_V .
\end{equation}
However, $\|\one_V\|=\dim V$ and so $\|\one_I\|^2=\sum_V (\dim V)^2$. This does not converge, so $\one_I$ does not exist in $\cH_I$. Given an enumeration $\{V_l\}_{l\in \N}$ of one unitary irreducible representation per equivalence class, the \emph{sequential approximate identity} $\{\one_{I,n}\}_{n\in\N}$ given by expression~(\ref{eq:assocsid}) does exist,
\begin{equation}
	\one_{I,n}=\sum_{l=1}^n \one_{V_l} .
\end{equation}
As for the subalgebra $\cH_C\subseteq\cH_I$, the formula (\ref{eq:assocproj}) for the projector $P$ of Proposition~\ref{prop:subalg} is straightforward to verify,
\begin{equation}
	\sum_{V,i,j} \dim V\, t^V_{i j}\bullet t^W_{k l}\bullet t^V_{j i}=
	\delta_{k l}\frac{\chi^W}{\dim W} .
	\label{eq:baseproj}
\end{equation}

We proceed to evaluate the formulas for amplitudes of Riemann surface as described in Proposition~\ref{prop:genus0a} of Section~\ref{sec:2area}. Firstly, for the case of vanishing genus and with at least one hole we have from relation~(\ref{eq:assocgen0ampla}),
\begin{align}
	\rho_{a;0,k}(\chi^{V_1}\tens\cdots\tens\chi^{V_k})
	& =\rho_{D_a}(\chi^{V_1}\bullet\cdots\bullet \chi^{V_k}) 
	=\delta_{V_1,\ldots,V_k}(\dim V_1)^{1-k}\rho_{D_a}(\chi^{V_1}) \nonumber \\
	& =\delta_{V_1,\ldots,V_k}(\dim V_1)^{2-k}\exp(\beta_{V_1} a) .
\end{align}
This is clearly well-defined on $\cH_C^\ds\tens\cdots\tens\cH_C^\ds$.
For the sphere we find with expression~(\ref{eq:assocsphereampla}),
\begin{equation}
	\rho_{a;0,0} =\sum_{k\in I} (-1)^{[\zeta_k]+|\zeta_k|}\tr_a(\zeta_k^\star \bullet\zeta_k)=\sum_V \tr_a(\one_V)=\sum_V \rho_{D_a}(\one_V)
	=\sum_V (\dim V)^2 \exp(\beta_V a).
\end{equation}
For $a=0$ this expression is not well-defined if there are infinitely many non-equivalent irreducible representations of $G$. This is the case if $G$ is any non-trivial compact Lie group. So, we have to exclude the sphere of vanishing area from the admissible regions, as well as any gluing that results in this sphere. For $a>0$ the expression is well-defined if the sequence $\{\beta_V\}$ converges to $-\infty$, and fast enough.

For Riemann surfaces of non-zero genus (Theorem~\ref{thm:arbitrarygenusa}) we first remark that the element $x$ given by expression (\ref{eq:assoccbase}) is well-defined,
\begin{equation}
	x= \sum_{j\in I'} (-1)^{[\xi_j]} \xi_j\bullet\xi_j^\star
	=\sum_V \chi^V\bullet\chi^V=\sum_V \frac{\chi^V}{\dim V} .
\end{equation}
Note $\|x\|^2=\sum_V (\dim V)^{-2}$. Moreover, as is easy to see,
\begin{equation}
	x^{\bullet g}=\sum_V \frac{\chi^V}{(\dim V)^{2g-1}} .
\end{equation}
With this we obtain as the amplitude for a Riemann sphere of genus $g\ge 1$ with $k\ge 1$ holes and area $a$, from expression (\ref{eq:assocxampla}),
\begin{equation}
	\rho_{a;g,k}(\chi^{V_1}\tens\cdots\tens\chi^{V_k})
	=\rho_{D_a}(x^{\bullet g}\bullet\chi^{V_1}\bullet\cdots\bullet\chi^{V_k})
	=\delta_{V_1,\ldots,V_k} (\dim V)^{2-2g-k}\exp(\beta_{V_1} a) .
\end{equation}
This is clearly well-defined on $\cH_C^\ds\tens\cdots\tens\cH_C^\ds$. For genus $g\ge 1$ and $k=0$ holes with area $a$ we obtain,
\begin{equation}
\rho_{a;g,k}
=\rho_{D_a}(x^{\bullet g})=\sum_V  (\dim V)^{2-2g} \exp(\beta_V a) .
\end{equation}
We have already discussed the case of the sphere. The case of the torus ($g=1$) is similar, we must exclude at least the torus of vanishing area. For area $a>0$ the torus amplitude is well-defined if the sequence $\{\beta_V\}$ converges to $-\infty$, and fast enough. For surfaces of higher genus, i.e., $g\ge 2$ it is sufficient for the amplitude to be well-defined that the dimensions of irreducible representations increase stronger than the square-root and that the parameters $\{\beta_V\}$ are bounded from above.


\section{Conclusions and Outlook}

In this work we introduce the notion of \emph{Compositional Quantum Field Theory (CQFT)} in terms of an axiomatic system (Section~\ref{sec:axiomatics}). This is motivated from physics as an attempt towards axiomatizing realistic \emph{Quantum Field Theory (QFT)} and theories that go beyond it, in particular \emph{Quantum Gravity}. As such CQFT is meant as an improvement of \emph{General Boundary Quantum Field Theory (GBQFT)}, which in turn is inspired by \emph{Topological Quantum Field Theory (TQFT)}, see the Introduction. Mathematically, CQGT is motivated by recent developments in Applied Category Theory. In particular, while also alluding to a notion of \emph{locality} in GBQFT and QFT, the attribute \emph{compositional} is to be understood in a categorical sense. The version of CQFT presented in this paper should be understood as the first step towards a functorial formulation of GBQFT. In the accompanying paper \cite{RobertJuan2} we prove that the state space correspondence of a CQFT is precisely the data of an involutive symmetric monoidal functor \cite{Egger,Jacobs}, defined on the category of gluing functions, and thus has the structure of an algebra over an involutive operad \cite{BeniniSchenkelWoike}. Apart from these categorical aspects, the key innovations in CQFT compared to GBQFT (and also TQFT) are a new framework for encoding the notion of gluing, through \emph{gluing functions} and \emph{relative gluing diagrams}, see Section~\ref{usualgluings}. Compared to TQFT (but not to GQBFT) also the \emph{slice regions} are a novelty. As a consequence of this new framework for gluing, the axioms of CQFT automatically have a notion of \emph{equivariance} built-in (Section~\ref{sec:axiomatics}).
In the following we provide a few additional comments and an outlook concerning different aspects of the present work.

In Section~\ref{sec:2dim}, the case of CQFT of dimension $2$ is investigated in some depth. We stop short of a full classification result, although we suspect this might be achieved with not too much additional effort. Remarkably (under mild additional assumptions) the emerging structure is that of a Krein $\star$-superalgebra with a special central commutative subalgebra. The former is associated to the interval or open string and the latter to the circle or closed string. If the symmetric structure of the target category is trivial (i.e., the grading of the state spaces is trivial), this is a Hilbert $\star$-algebra with a special central commutative subalgebra. It is suggestive to try to recover the classifying results of $2$-dimensional TQFT \cite{Kock}, i.e, equivalence with the category of commutative Frobenius algebras. An initial remark in this direction is Observation~\ref{obs:rec2tqft}. While it is not yet completely clear what the algebraic objects parametrizing 2d CQFTs are, these are expected to contain the corresponding classifying objects of 2d TQFT while satisfying higher dimensional relations; as oriented surfaces with corners can be generated by a single disk with multi-directional compositions, whereas the category of 2-dimensional bordisms needs at least two generators and only considers in-out composition and relations. These ideas, as the other ideas presented in this subsection, will be further explored elsewhere. 

In the present article, we consider \emph{topological} manifolds as models for pieces of spacetime. This is only meant as a first and most basic setting. In Sections~\ref{sec:areaaxioms} and \ref{sec:2area} we add a notion of area. In general, however, we are interested in equipping the manifolds with additional structure, such as differentiable structure, conformal structure or a metric. The latter is essential if we are to describe physically realistic QFTs. Of further particular interest are different types of fiber bundles, for example to model gauge theories. By comparison, differential structures have also played an important role in TQFT, but richer structures not so much, in part related to the restriction of TQFT to finite-dimensional state spaces. The traditional way to implement additional structure would be to revise the axiomatic system and modify the axioms one-by-one to make them compatible with the new structure. This has been followed both in TQFT and GBQFT. The categorical point of view suggests a different and more powerful route, however: \emph{enrichments}. Here the structure is added at the category level and the whole axiomatic system of CQFT (encoded as a functor) is lifted at once. This will be detailed in future work by the authors.

An additional structure of importance in QFT is that of \emph{observable}. In textbook accounts of QFT these appear most often in the guise of spacetime point-localized field operators. Spacetime region-localized observables play a fundamental role in Algebraic Quantum Field Theory (AQFT) \cite{HaKa:aqft}. In GBQFT, observables have been implemented as spaces of \emph{observable functions} associated to spacetime regions \cite{Oe:feynobs}. The amplitude then plays the role of a kind of unit object under gluing in the space of observable functions for the respective spacetime region. Extending CQFT to include such a notion of observable should be relatively straightforward. Such a CQFT with observables promises to have interesting relations not only to AQFT, but also to the factorization algebra approach to QFT in the sense of Costello and Gwilliam \cite{CoGw:facalgqft1}.

CQFT as presented here, as well as GBQFT and TQFT are modeled on QFT. In particular, the spaces associated to hypersurfaces are Krein or Hilbert or at least vector spaces that are meant to play the role of \emph{state spaces} of a \emph{quantum theory}. One might, on the other hand, envisage a similar axiomatic system for \emph{classical field theory}. Indeed, a simple axiomatic system, broadly analogous to that of GBQFT for classical field theory was introduced in \cite{Oe:holomorphic} based on earlier work by Kijowski and Tulczyjew \cite{KiTu:symplectic}. This allows to express a notion of \emph{geometric quantization} in a way that should be a functor (or natural transformation) once the categorical underpinnings are properly formalized.

As explained in the introduction, GBQFT is the axiomatic system of the \emph{General Boundary Formulation of Quantum Theory (GBF)}. The latter has seen considerable development inspired by operational approaches to the foundations of quantum theory. This has culminated in the \emph{Positive Formalism (PF)} \cite{Oe:dmf,Oe:locqft,Oe:posfound}, a compositional framework for describing physical theories, based on the concepts of operationalism and locality. The axiomatic system of the PF is very similar to GBQFT, but the state spaces are now \emph{partially ordered vector spaces}. Remarkably, this axiomatic system may encode both classical and quantum theories (plus theories that are neither). In the classical case this is really classical \emph{statistical} field theory and in the quantum case this is a \emph{mixed state} formalism. What is more, a notion mathematically analogous to that of observable of QFT is present, but with a physical interpretation of \emph{probe} or \emph{process}. It is highly desirable and appears very feasible to generalize CQFT to such a setting. Eventually, one would hope that the analog of moving from GBQFT to the axiomatic system of the PF would be as easy as replacing the target category of a functor. This will be fleshed out in future work by the authors.

A final aspect we would like to mention here concerns the physically correct treatment of fermionic field theories. This has been a difficult subject for TQFT. In GBQFT the decisive step in understanding this was made in \cite{Oe:freefermi}. In particular, it was recognized that state spaces in this case are generally Krein spaces. Only with a restriction to hypersurfaces that are exclusively spacelike is the traditional treatment in terms of Hilbert spaces recovered. Recognizing this has required a convergence of different physical and mathematical considerations in \cite{Oe:freefermi}. In contrast, we find here that in CQFT this same treatment of fermions acquires an extremely simple and natural form, making it even more compelling. In the axioms of Definition~\ref{QFTAFfin} the only trace of this is that all state spaces and maps between them are $\Z_2$-graded. Moreover, in Axiom~\ref{TensorIsometryandGradedSymmetryAxiom} this grading is linked to a factor $-1$ when interchanging fermionic particles. As the categorically minded reader immediately recognizes, this axiom just specifies the \emph{symmetric structure} of the target category.

\subsection*{Acknowledgments}

This work was partially supported by UNAM-PAPIIT project grant IN106422. This publication was made possible through the support of the ID\# 61466 grant from the John Templeton Foundation, as part of the “The Quantum Information Structure of Spacetime (QISS)” Project (qiss.fr). The opinions expressed in this publication are those of the authors and do not necessarily reflect the views of the John Templeton Foundation.


\newcommand{\eprint}[1]{\href{https://arxiv.org/abs/#1}{#1}}
\bibliographystyle{stdnodoi}
\bibliography{refs,stdrefsb}

\begin{thebibliography}{10}
\providecommand{\url}[1]{\texttt{#1}}
\providecommand{\urlprefix}{URL }
\providecommand{\selectlanguage}[1]{\relax}
\providecommand{\eprint}[2][]{\url{#2}}

\bibitem{Oe:boundary}
R.~Oeckl, \textit{A ``general boundary'' formulation for quantum mechanics and
  quantum gravity}, Phys. Lett. \textbf{B 575} (2003) 318--324,
  \eprint{hep-th/0306025}.

\bibitem{Oe:gbqft}
R.~Oeckl, \textit{General boundary quantum field theory: Foundations and
  probability interpretation}, Adv. Theor. Math. Phys. \textbf{12} (2008)
  319--352, \eprint{hep-th/0509122}.

\bibitem{Oe:holomorphic}
R.~Oeckl, \textit{Holomorphic Quantization of Linear Field Theory in the
  General Boundary Formulation}, SIGMA \textbf{8} (2012) 050, 31 pages,
  \eprint{1009.5615}.

\bibitem{Oe:freefermi}
R.~Oeckl, \textit{Free Fermi and Bose Fields in TQFT and GBF}, SIGMA \textbf{9}
  (2013) 028, 46 pages, \eprint{1208.5038v2}.

\bibitem{Oe:feynobs}
R.~Oeckl, \textit{Schr\"odinger-Feynman quantization and composition of
  observables in general boundary quantum field theory}, Adv. Theor. Math.
  Phys. \textbf{19} (2015) 451--506, \eprint{1201.1877}.

\bibitem{CoOe:locgenvac}
D.~Colosi, R.~Oeckl, \textit{Locality and General Vacua in Quantum Field
  Theory}, SIGMA \textbf{17} (2021) 073, 83 pages, \eprint{2009.12342}.

\bibitem{Fong}
B.~Fong, \textit{Decorated cospans}, Theory and Applications of Categories
  \textbf{30} (2015) 1096--1120.

\bibitem{BaezCourser}
J.~C. Baez, K.~Courser, \textit{Structured cospans}, Theory and Applications of
  Categories \textbf{35} (2020) 1771--1822.

\bibitem{RobertJuan2}
R.~Oeckl, J.~Orendain, \textit{Compositional Quantum Field Theory: State spaces
  and involutive symmetric monoidal functors}, in preparation.

\bibitem{StWi:pct}
R.~F. Streater, A.~S. Wightman, \textit{PCT, Spin and Statistics, and All
  That}, W. A. Benjamin, New York-Amsterdam, 1964.

\bibitem{HaKa:aqft}
R.~Haag, D.~Kastler, \textit{An Algebraic Approach to Quantum Field Theory}, J.
  Math. Phys. \textbf{5} (1964) 848--861.

\bibitem{Oe:catandclock}
R.~Oeckl, \textit{Schr\"odinger's cat and the clock: Lessons for quantum
  gravity}, Class. Quantum Grav. \textbf{20} (2003) 5371--5380,
  \eprint{gr-qc/0306007}.

\bibitem{Oe:posfound}
R.~Oeckl, \textit{A local and operational framework for the foundations of
  physics}, Adv. Theor. Math. Phys. \textbf{23} (2019) 437--592,
  \eprint{1610.09052v3}.

\bibitem{Tur:qinv}
V.~G. Turaev, \textit{Quantum Invariants of Knots and 3-Manifolds}, de Gruyter,
  Berlin, 1994.

\bibitem{Ati:tqft}
M.~Atiyah, \textit{Topological quantum field theories}, Inst. Hautes \'Etudes
  Sci. Publ. Math. \textbf{68} (1988) 175--186.

\bibitem{Oe:2dqym}
R.~Oeckl, \textit{Two-dimensional quantum Yang-Mills theory with corners}, J.
  Phys. A \textbf{41} (2008) 135401, \eprint{hep-th/0608218}.

\bibitem{RuSz:areaqft}
I.~Runkel, L.~Szegedy, \textit{Area-Dependent Quantum Field Theory}, Commun.
  Math. Phys. \textbf{381} (2021) 83--117, \eprint{1807.08196}.

\bibitem{StolzTeichnerSusy}
S.~Stolz, P.~Teichner, \textit{Supersymmetric field theories and generalized
  cohomology}, 2011, \eprint{1108.0189}.

\bibitem{Bartels1}
A.~Bartels, C.~L. Douglas, A.~Henriques, \textit{Conformal nets I: Coordinate
  free nets}, Int. Math. Res. Not. \textbf{13} (2015) 4975--5052.

\bibitem{MoWa:blobhomology}
S.~Morrison, K.~Walker, \textit{Blob Homology}, Geom. Topol. \textbf{16} (2012)
  1481--1607, \eprint{1009.5025}.

\bibitem{Walker}
K.~Walker, \textit{TQFTs}, https://canyon23.net/math/tc.pdf, 2006.

\bibitem{RunkelCarquevilleDefects}
N.~Carqueville, I.~Runkel, \textit{Orbifolds of n–dimensional defect TQFTs},
  Geometry and Topology \textbf{23} (2019) 781--864.

\bibitem{Abr:2TQFTfrob}
L.~Abrams, \textit{Two-dimensional topological quantum field theories and
  frobenius algebras}, J. Knot Theory Ramifications \textbf{05} (1996)
  569--587.

\bibitem{Wit:qgauge2d}
E.~Witten, \textit{On Quantum Gauge Theories in Two Dimensions}, Commun. Math.
  Phys. \textbf{141} (1991) 153--209.

\bibitem{Egger}
J.~M. Egger, \textit{On involutive monoidal categories}, Theory and
  Applications of Categories \textbf{25} (2011) 368--393.

\bibitem{Jacobs}
B.~Jacobs, \textit{Involutive Categories and Monoids, with a
  GNS-correspondence}, Foundations of Physics \textbf{42} (2012) 874--895.

\bibitem{BeniniSchenkelWoike}
M.~Benini, A.~Schenkel, L.~Woike, \textit{Involutive categories, colored
  operads and quantum field theory}, Theory and Applications of Categories
  \textbf{34} (2019) 13--57.

\bibitem{Kock}
J.~Kock, \textit{Frobenius Algebras and 2-D Topological Quantum Field
  Theories}, no.~59 in London Mathematical Society Student Texts, Cambridge
  University Press, Cambridge, 2003.

\bibitem{CoGw:facalgqft1}
K.~Costello, O.~Gwilliam, \textit{Factorization Algebras in Quantum Field
  Theory}, vol.~1, Cambridge University Press, Cambridge, 2016.

\bibitem{KiTu:symplectic}
J.~Kijowski, W.~M. Tulczyjew, \textit{A Symplectic Framework for Field
  Theories}, Springer, Berlin, 1979.

\bibitem{Oe:dmf}
R.~Oeckl, \textit{A Positive Formalism for Quantum Theory in the General
  Boundary Formulation}, Found. Phys. \textbf{43} (2013) 1206--1232,
  \eprint{1212.5571}.

\bibitem{Oe:locqft}
R.~Oeckl, \textit{Towards state locality in quantum field theory: free
  fermions}, Quantum Stud. Math. Found. \textbf{4} (2017) 59--77,
  \eprint{1307.5031}.

\end{thebibliography}
\end{document}